\documentclass[fleqn,usenatbib]{mnras}

\usepackage{newtxtext,newtxmath}

\usepackage[T1]{fontenc}

\DeclareRobustCommand{\VAN}[3]{#2}
\let\VANthebibliography\thebibliography
\def\thebibliography{\DeclareRobustCommand{\VAN}[3]{##3}\VANthebibliography}

\usepackage{graphicx}	
\usepackage{amsmath}	
\usepackage{array}
\usepackage{cleveref}
\usepackage{ulem}
\usepackage{booktabs} 
\usepackage{amsfonts}
\usepackage{bm}
\usepackage{nicefrac}
\usepackage[hypcap=false]{caption}
\usepackage{subcaption}
\usepackage{epsfig}
\usepackage{wrapfig}
\usepackage{lineno}
\newcommand{\copyrightnotice}{\enlargethispage{24pt}
\let\thefootnote\relax\footnote{\copyright\ 2015 Christopher Heil}}

\newcommand{\svee}{{\hbox{\raise.4ex \hbox{${\scriptscriptstyle{\vee}}$}}}}
\newcommand{\swedge}{{\hbox{\raise.4ex \hbox{${\scriptscriptstyle{\wedge}}$}}}}

\newcommand{\Ac}{{\mathcal{A}}}

\newcommand{\CHI}{\hbox{\raise.4ex \hbox{$\chi$}}}

\newcommand{\scap}{\hbox{\raise.25ex
\hbox{${\operatornamewithlimits{\scriptstyle\bigcap}}$}}}
\newcommand{\scup}{\hbox{\raise.25ex
\hbox{${\operatornamewithlimits{\scriptstyle\bigcup}}$}}}

\newcommand{\Dc}{{\mathcal{D}}}

\newcommand{\deltacheck}
  {\overset{\lower.4ex \hbox{${\scriptscriptstyle{\hskip 2 pt\vee}}$}} \delta}

\newcommand{\Fcheck}
    {\overset{\lower.4ex \hbox{${\scriptscriptstyle{\hskip 2 pt\vee}}$}} F}

\newcommand{\fwedgehat}
    {\overset{\lower.6ex \hbox{${\scriptscriptstyle{\hskip 3 pt\wedge}}$}} f}
\newcommand{\fveecheck}
    {\overset{\lower.4ex \hbox{${\scriptscriptstyle{\hskip 2 pt\vee}}$}} f}
\newcommand{\fkcheck}
   {\overset{\lower.4ex \hbox{${\scriptscriptstyle{\hskip 1 pt\vee}}$}} {f_k}}

\newcommand{\raiseprime}{\hbox{\raise.3ex \hbox{${\scriptstyle{\prime}}$}}}

\newcommand{\Gc}{{\mathcal{G}}}
\newcommand{\Gcheck}
    {\overset{\lower.4ex \hbox{${\scriptscriptstyle{\hskip 2 pt\vee}}$}} G}

\newcommand{\gveecheck}
    {\overset{\lower.4ex \hbox{${\scriptscriptstyle{\hskip 2 pt\vee}}$}} g}

\newcommand{\hcheck}
    {\overset{\lower.4ex \hbox{${\scriptscriptstyle{\hskip 2 pt\vee}}$}} h}

\newcommand{\Kcheck}
  {\overset{\lower.4ex \hbox{${\scriptscriptstyle{\hskip 2 pt\vee}}$}} K}

\newcommand{\mucheck}{\overset{\lower.4ex \hbox{${\scriptscriptstyle{\hskip 2 pt\vee}}$}} \mu}

\newcommand{\norm}[1]{\left\|#1\right\|}

\newcommand{\Oc}{{\mathcal{O}}}

\newcommand{\varphicheck}
    {\overset{\lower.4ex \hbox{${\scriptscriptstyle{\hskip 1 pt\vee}}$}}
              \varphi}

\newcommand{\R}{\mathbb{R}}

\newcommand{\Rbar}
  {{\overset{\hskip -0.9 pt \lower\ 1.5pt \hbox{{\rule{6.7pt}{0.45pt}}}} \R}}
\newcommand{\subRbar}
   {{\overset{\hskip -0.8 pt \lower\ 1.5pt \hbox{{\rule{4.5pt}{0.5pt}}}} \R}}

\newcommand{\Sc}{{\mathcal{S}}}

\newcommand{\SNa}{S_N^{\hskip 0.5 pt
  \hbox{\raise.3ex \hbox{\small\textup{a}}}}}
\newcommand{\SNaa}[1]{S_#1^{\hskip 0.5 pt
  \hbox{\raise.3ex \hbox{\small\textup{a}}}}}
\newcommand{\SNo}{S_N^{\hskip 0.5 pt
  \hbox{\raise.3ex \hbox{\small\textup{o}}}}}
\newcommand{\SNt}{S_N^{\hskip 0.5 pt
  \hbox{\raise.3ex \hbox{\small\textup{t}}}}}

\newcommand{\T}{{\mathbb{T}}}

\newcommand{\zeroveecheck}
    {\overset{\lower.4ex \hbox{${\scriptscriptstyle{\hskip 0.5 pt\vee}}$}} 0}

\crefname{figure}{Figure}{Figures}
\crefname{appendix}{Appendix}{Appendices}
\crefname{align}{Eq.}{Eqs.}
\crefname{equation}{Eq.}{Eqs.}
\crefname{section}{section}{sections}
\crefname{subsection}{section}{sections}
\crefname{subsubsection}{section}{sections}
\crefname{tabular}{Table.}{Tables.}
\newcommand{\coo}{\ensuremath{\mathrm{CO_2}}}

\newtheorem{theorem}{Theorem}[section]

\title[CBP Obliquity Variations]{Low Spin-Axis Variations of Circumbinary Planets}

\author[R. Chen et al.]{Renyi Chen,$^{1}$ Gongjie Li,$^{2}$ and Molei Tao$^{1}$ 
\\
$^{1}$School of Mathematics, Georgia Institute of Technology, Atlanta, GA 30332, USA \\
$^{2}$Center for Relativistic Astrophysics, School of Physics, Georgia Institute of Technology, Atlanta, GA 30332, USA
}

\date{Accepted XXX. Received YYY; in original form ZZZ}

\pubyear{2015}

\begin{document}
\label{firstpage}
\pagerange{\pageref{firstpage}--\pageref{lastpage}}
\maketitle

\begin{abstract}
Having a massive moon has been considered as a primary mechanism for stabilized planetary obliquity, an example of which being our Earth. This is, however, not always consistent with the exoplanetary cases. This article details the discovery of an alternative mechanism, namely that planets orbiting around binary stars tend to have low spin-axis variations. This is because the large quadrupole potential of the stellar binary could speed up the planetary orbital precession, and detune the system out of secular spin-orbit resonances. Consequently, habitable zone planets around the stellar binaries in low inclination orbits hold higher potential for regular seasonal changes comparing to their single star analogues. 
\end{abstract}




\section{Introduction}

Orientation of planetary spin axis plays an important role in shaping the climate of a planet. For Earth, the spin-orbit misalignment, known as obliquity, varies between $21.4^\circ$ and $24.4^\circ$, and this together with the variations in the orbit's eccentricity (oscillating up to $\sim 0.06$) is the main cause of the Milankovitch cycles of Earth climate variations, significant consequences of which include glacial cycles \citep[e.g.,][]{Milankovitch41, Hays76, Imbrie80, Raymo97}. For Mars, the spin axis exhibits wild variations that can reach $\sim 60^\circ$  \citep{Ward73, laskar1993chaotic, Touma93}, which has drastic effects: with a high obliquity of $\sim 45^\circ$, Mars' atmosphere could have partially precipitated due to \coo condensation; this is consistent with the glacier-like land form discovered in the tropics and mid-latitude of Mars  \citep[e.g.,][]{Head05, Forget06}.

Why does Earth have low obliquity variations, different from Mars? It was found that the Moon stabilizes Earth obliquity, and in the absence of the Moon, Earth spin-axis would be chaotically varying with obliquity ranging between $\sim 0 -  50^\circ$ at $\lesssim$Myr timescales \citep{Laskar93a, Li14}. However, this criterion is not necessary for planets in general, such as for the discovered exoplanets. Beyond the Solar System, the rapidly growing number of detected exoplanets calls for a better understanding of the uniqueness of life and searches for habitable worlds in the universe. Obliquity of exoplanets becomes one essential feature related to habitability, as it directly determines the insolation distribution and affects the snowball transitions \citep{Quarles21}. This, for instance, motivated the search for exomoons, as they are thought to guarantee stable obliquities \citep{agol2015center, teachey2018evidence}. Yet, exoplanets may not need a moon to stabilize their obliquity, since they could have very different orbital architectural properties that allow weak coupling with their planetary companions (e.g., Kepler 186-f and Kepler 62-f \citep{Shan18, Quarles20}). Moreover, a massive moon could destabilize obliquities, such as an Earth-like analogue in Alpha Cen \citep{Quarles19}.

What other mechanism(s) can stabilize obliquity then? Here we show that near coplanar circumbinary planets have much more stable obliquities than their single star counterparts. Although different from the common impression based on the Solar System, planets around stellar binaries are not scarce in the universe -- roughly half of Solar-type stars are in binaries \citep{Duquennoy91,Raghavan10}, and planets around binaries with periods less than $\sim 300$days are as prevalent as those around single stars  \citep{Armstrong14, li2016uncovering}. Such circumbinary planets show very different dynamical properties from planets around single star. For instance, it was known that regions of dynamical stability become smaller than those around single stars, due to strong perturbations from the stellar companions \citep{Holman99, quarles2020orbital}. More new dynamical properties still await discovery; for instance, a main technical contribution of this paper is the quantification of intriguing behaviors of spin-axis dynamics of circumbinary planets. This leads to our demonstration of stabilized obliquity.


To explain why binary host stars can stabilize the obliquity of a circumbinary planet, let's start with the physical root of obliquity variations. In general, large obliquity variations manifest due to resonances between gravitational torque from the host star and the orbital perturbations from a planet's companions. Specifically, if there was just one star, the planetary spin axis is torqued by the only star. This would lead to spin-axis precession in the same way as the precession of a top under Earth's gravity. If the spin axis precession frequency coincides with the orbital oscillation frequency, the spin axis can vary with large amplitude due to spin-orbit resonance. For a single star system, a planet's orbit can be perturbed by its planetary companions, resulting in secular spin-orbit resonances and large obliquity variations. Having two stars in the center, however, completely changes the physics. One may wonder if having a second star could correspond to a perturbation similar to that due to a planetary companion, hence creating secular spin-orbit resonances. However, the truth is rather the contrary, as the binary stabilizes rather than destabilizes the planetary obliquity. This is because (i) there are torques from both stars, and (ii) the orbital precession is much larger due to the stellar binary quadrupole potential. Therefore, the obliquity variation can be very different from the well understood single-star case, and the traditional framework for analysis cannot be directly applied. A new framework is constructed in this article and it leads to the discovery of stabilized obliquity. We will first begin with a heuristic demonstration, followed by rigorous analyses based on the development of an analytical secular theory, and then validated by state-of-the-art numerical simulations.

\section{Heuristic Calculation}

\begin{figure}
\centering
\includegraphics[width=0.8\linewidth]{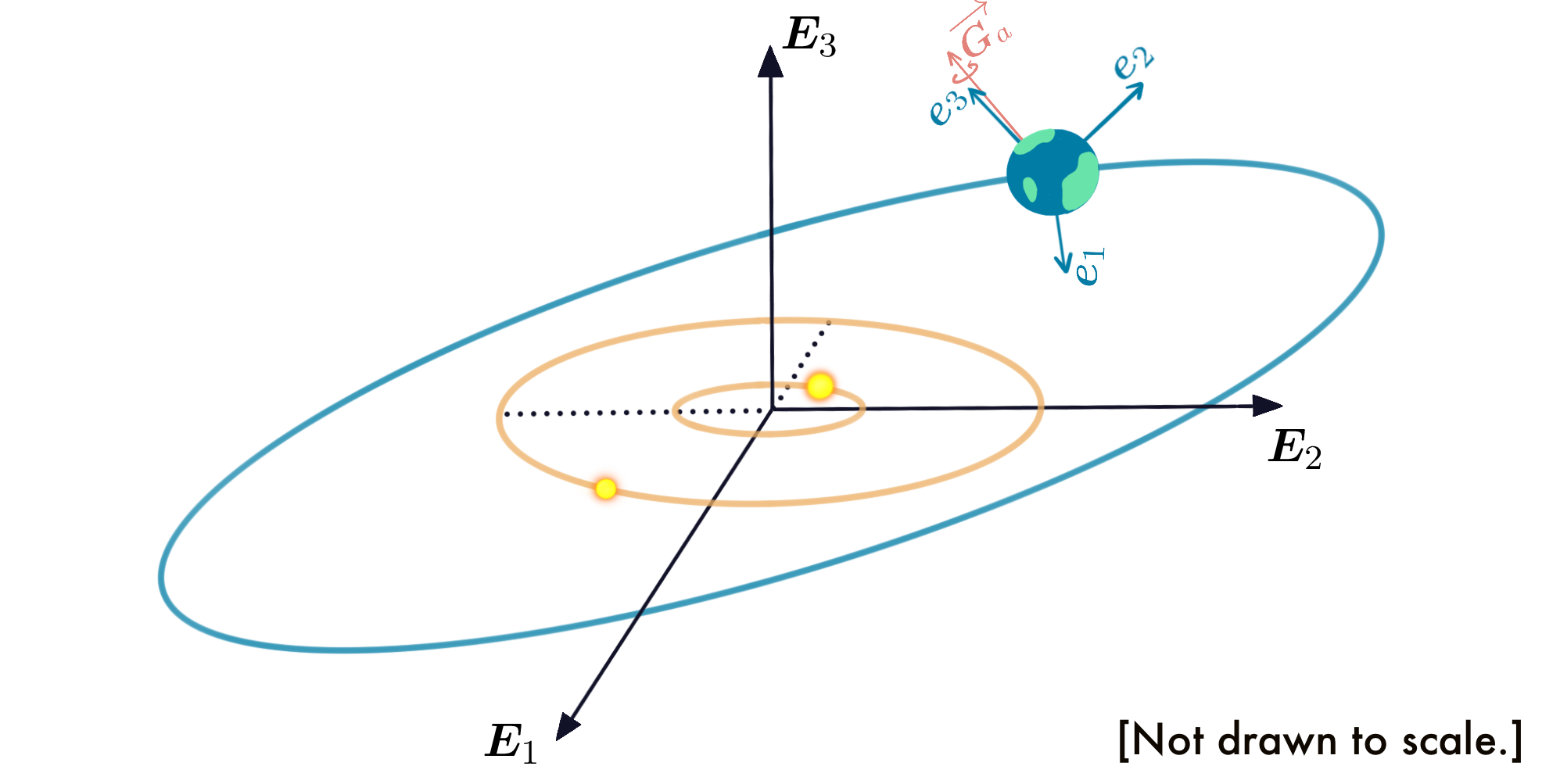}\\
\includegraphics[width=0.8\linewidth]{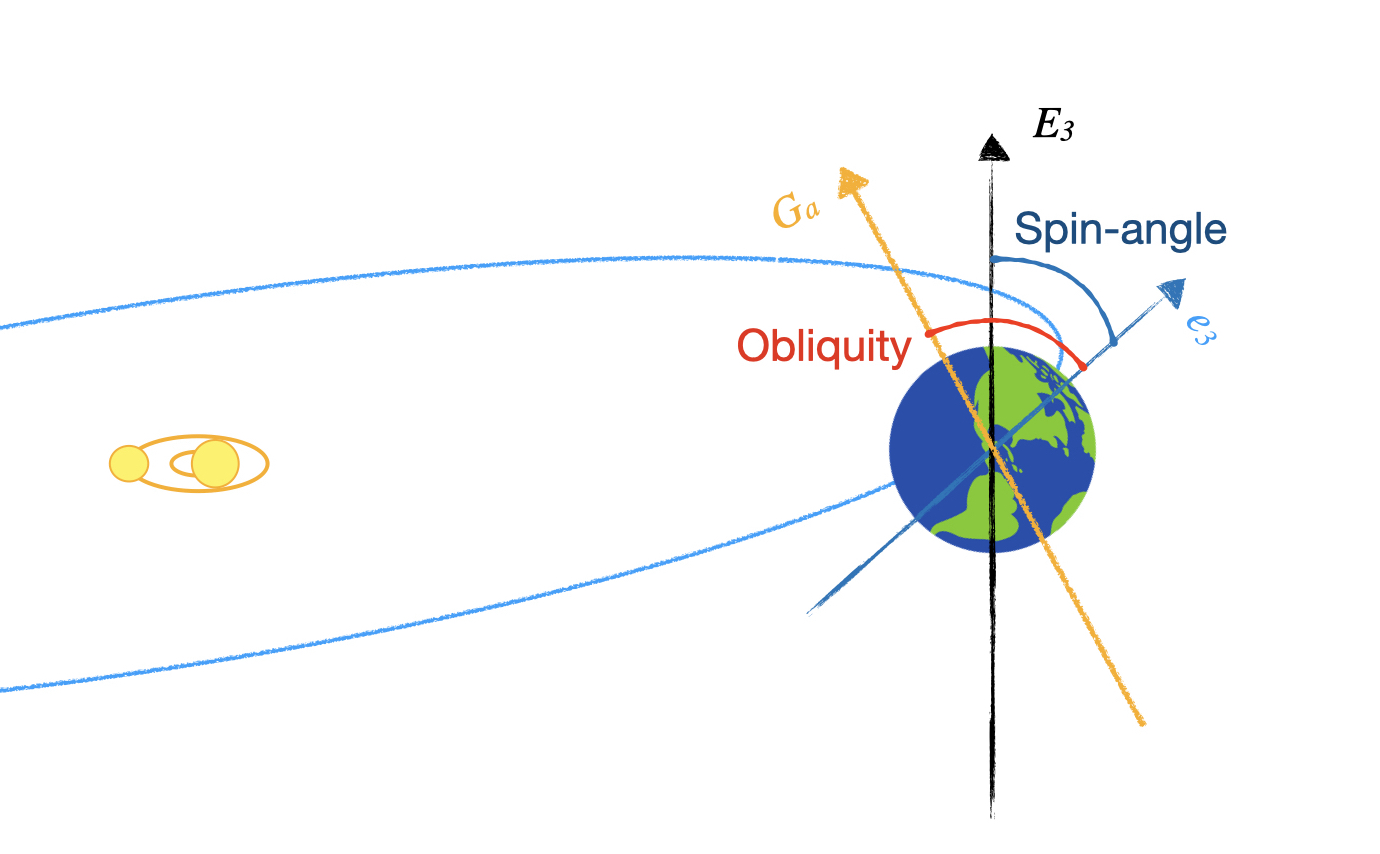}
\caption{\textbf{Upper panel: }The demonstration of a planet orbiting around the stellar binary. \textbf{Lower panel:} spin-axis orientation. The spin-angle measures the tilt of the spin-axis relative to the fixed inertial frame, and the obliquity (spin-orbit misalignment) measures the tilt of the spin-axis relative to the planetary orbital normal direction}. \label{fig:ref_frame}
\end{figure}

We show the set up of the problem in Fig. \ref{fig:ref_frame}, where we consider the spin dynamics of the planet orbiting around a stellar binary. The two yellow dots represent the stellar binary and the blue one represents the planet. The upper panel shows the inertial reference frame: $\{\bm{E_1, E_2, E_3}\}$, as well as the orientation of the spin axis $\{\bm{ e_1, e_2, e_3}\}$. The lower panel shows the parameters we use to characterize the tilt of the spin-axis: the spin-angle, which measures the tilt of the spin axis relative to the inertial frame and the obliquity (spin-orbit misalignment), which measures the tilt of the spin-axis relative to the orbital normal. Obliquity determines the stellar radiation onto the planet and affect the seasonal variations of the planet. Both the orbital variations and the spin-angle variations contribute to the obliquity variations. In this paper, we consider the variations of the spin-angle, as well as the consequent evolution of the obliquity due to spin-angle variations.

Planetary spin-axis could vary due to two mechanisms. First, the spin-axis could vary due to the strong spin-orbit coupling, as the spin-axis rapidly precesses around the orbital plane. This allows the spin-axis to follow the change of the orbit. The second mechanism is due to secular spin-orbit resonances, which occurs when the spin precession frequency matches that of the orbital nodal precession. Different from planets around a single star, we find that spin-axis variations of circumbinary planets are typically low. This is because the stellar binary in the center produces a much larger quadrupole moment, which leads to faster orbital precession comparing to the cases of a single star. This both disallow the planet spin-axis to precess and to follow the orientation change of the orbital normal, and detunes the system out of secular spin-orbit resonances. Thus, this leads to low variations in planetary spin-axis. In this section, we use a heuristic calculation to demonstrate the low spin-axis variations of circumbinary planets. 

The separation of the stellar binary components and their masses determine the quadrupole moment and the possibilities of secular spin-orbit resonances for Earth-like planets. Thus, assuming the planet is Earth-like and located at one Earth flux, we can estimate the period of the stellar binary, which could drive orbital precession to excite secular spin-orbit resonances. For a planet orbiting around a stellar binary, its orbit precesses around the orbits of the stellar binary, and the spin-axis precesses around the planetary orbital orientation. When the orbital precession frequency matches that of the spin-axis, spin-orbit resonance occurs which can drive large obliquity variations. 

The frequency of the orbital precession can be estimated as \citep{li2016uncovering}: 
\begin{align}
    \dot{h_d} = \frac{3}{4 \cos{(\delta i)}} n \Big(\frac{a_b}{a_p}\Big)^2 \frac{M_{*1} M_{*2}}{(M_{*1}+M_{*2})^2} ,
    \label{eqn:orbprec}
\end{align}
where $h_d$ represents the longitude of ascending node of the planetary orbit, and $\delta i$ is the mutual inclination between the orbit of the stellar binary and that of the planet, and $n$ is the orbital frequency, $a_b$ and $a_p$ are the semi-major axes of the stellar binary and the planet, and $M_{*1}$ and $M_{*2}$ are the masses of the stellar binary. This expression is accurate to the second order in $\frac{a_b}{a_p}$.

On the other hand, the planet spin-axis precession can be estimated too. In the case of single host star, this precession frequency due to the torque from a central star can be expressed as $\alpha \cos{\epsilon}$ \citep[e.g.,][]{laskar1993chaotic}, where
\begin{align}
    \alpha = \frac{3 G}{2\omega} \frac{M_{*}}{a_p^3} E_d ,
\end{align}
$\epsilon$ is the spin-orbit misalignment (obliquity), $\omega$ is the spin rate of the planet, $M_*$ is the mass of the central star and $E_d$ is the dynamical ellipticity of the planet. For Earth-like planets (with mass and interior structure similar to Earth), $E_d = E_{d,\oplus} (\omega/\omega_\oplus)^2$, where $E_{d, \oplus} = 3350\times10^{-6}$ and $\omega_\oplus$ are the dynamical ellipticity and rotation rate of the Earth separately.

To obtain a rough estimation of the spin-precession frequency in the case around a stellar binary, we substitute the mass of the central star with the total mass of the stellar binary, while assuming the average distance to each of the stars are roughly $\sim a_p$. Then, the spin-axis precession frequency becomes $3n^2/(2\omega)E_d$.

\begin{figure}
\centering
\includegraphics[width=0.8\linewidth]{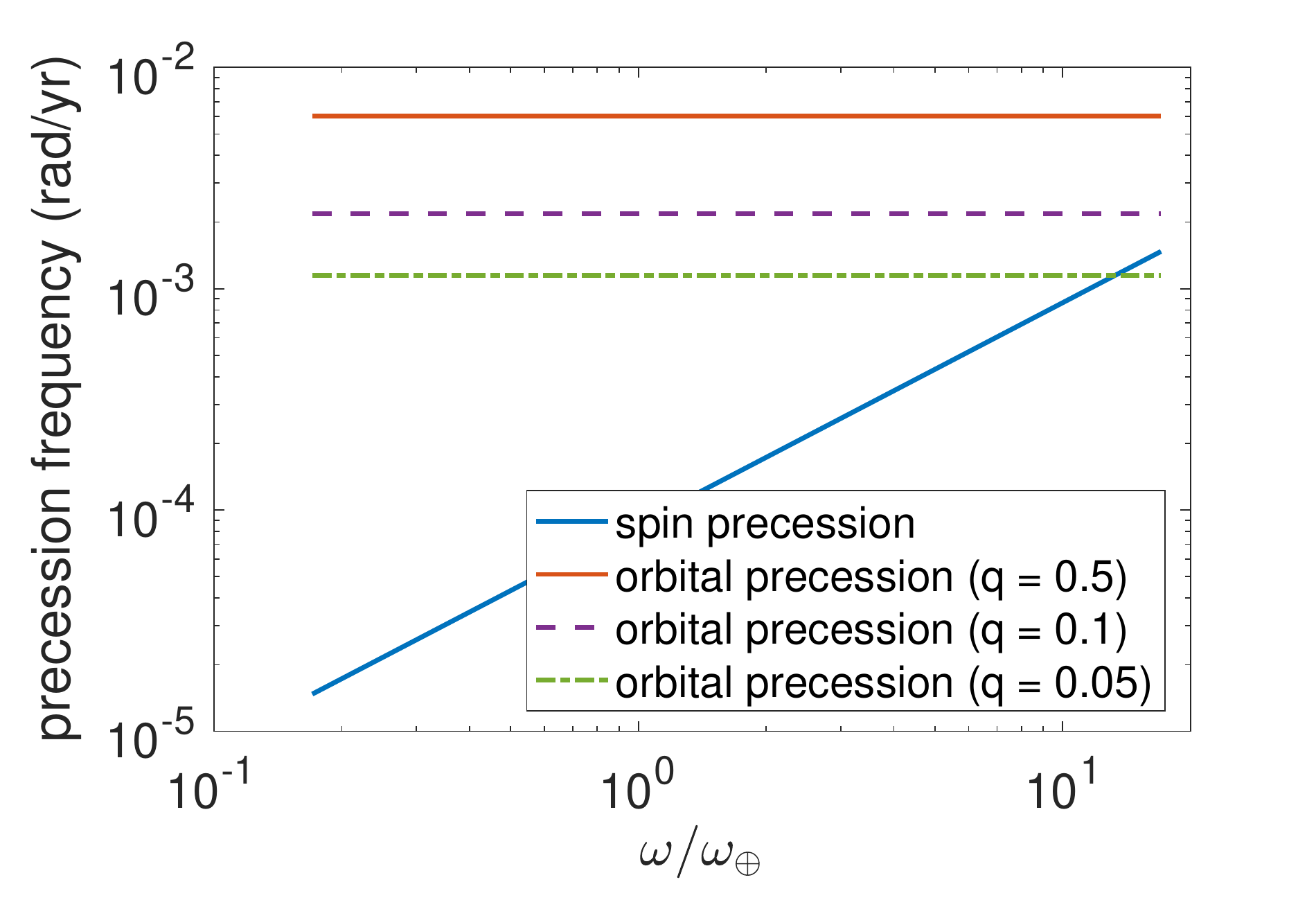}
\vspace{-4mm}
\caption{Spin-axis precession frequency and orbital nodal precession frequencies of an Earth-like planet around a stellar binary of total mass of 1 solar mass with a period of seven days. The x-axis represents the spinning rate of the planet. Spin-axis precession frequency is typically much lower than that of the orbit, unless for extreme mass ratio binaries with fast planet rotators close to the break-up spinning rate.
}\label{fig:freq}
\end{figure}

We plot in Figure \ref{fig:freq} the orbital and spin-axis precession frequencies for an Earth-like planet around a seven-day orbital period stellar binary. Specifically, we set the total mass of the stellar binary to be 1 solar mass, and the planetary mass to be one Earth mass. The separation between the planet and the stellar binary center of mass is set to be 1au. Different line types represent different stellar binary mass ratios ($q = m_1/(m_1 + m_2)$). The x-axis shows the rotation frequency of the planet, and the maximum frequency in the figure corresponds to the break-up spinning rate. It shows that the spin-axis precession rate is typically much lower than that of the orbit, except for extreme mass ratios ($q<0.05$) with a fast planet rotator close to the break-up spinning rate.

\begin{figure}
\centering
\includegraphics[width=0.8\linewidth]{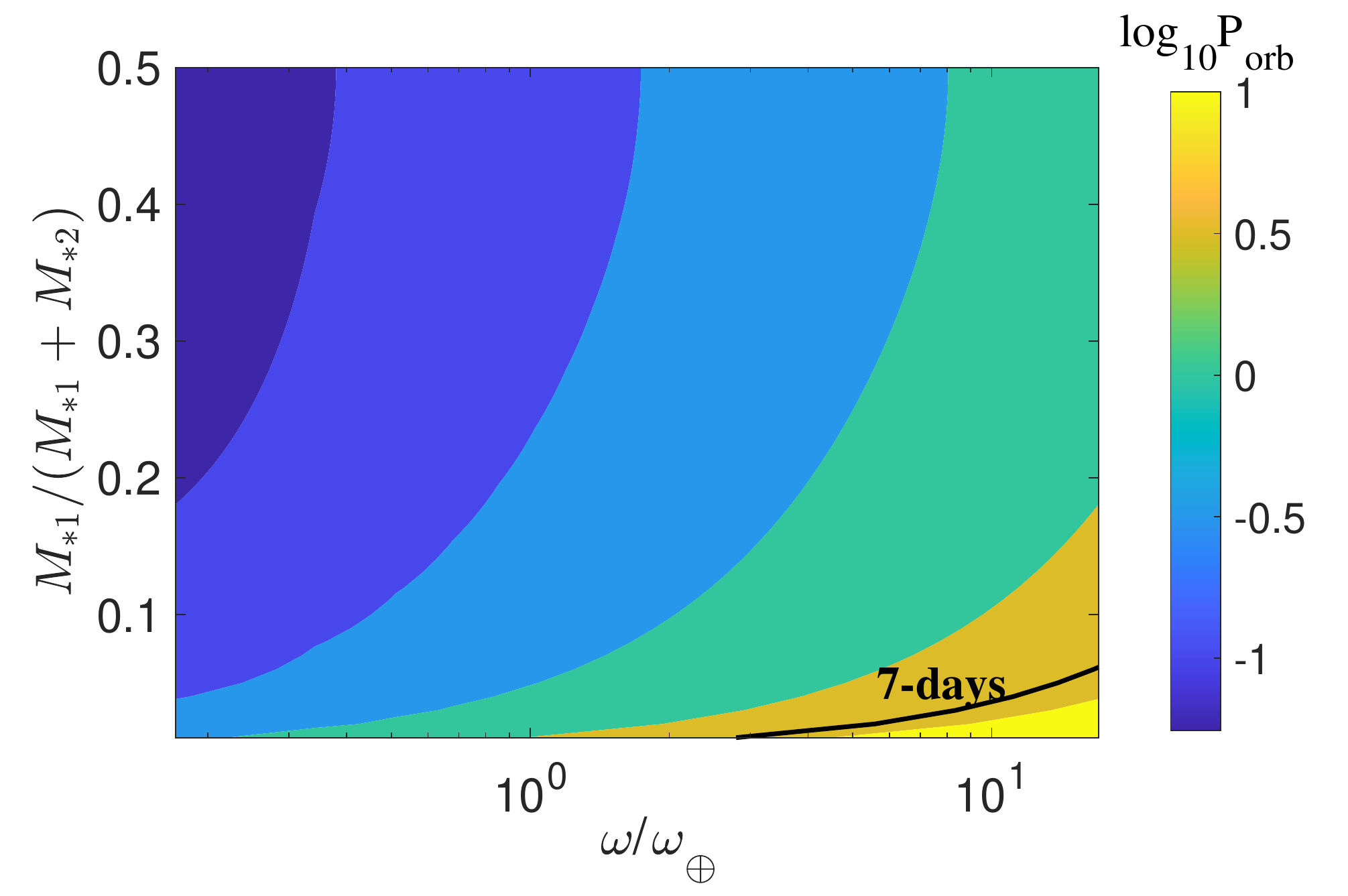}
\vspace{-4mm}
\caption{Stellar binary period ($P_{orb}$ in days) that leads to planet secular spin-orbit resonances. Secular spin-orbit resonances could only occur for closely separated stellar binaries ($P_{orb} \lesssim 1$ day), unless the mass ratio of the stellar binary is below $\lesssim 0.05$ for fast rotating planets with rotating period of $\lesssim 5$ hours.
}\label{fig:spin_orb}
\end{figure}

To explore a larger parameter space, Figure \ref{fig:spin_orb} plots the stellar binary period as a function of the planets rotation rate ($\omega/\omega_\oplus$) and the stellar binary mass ratio when the precession frequencies of the spin-axis and the orbit equal. We set the total mass of the stellar binary to be 1 solar mass. We assume the luminosity of the stars follows the mass-luminosity relation ($L\propto m^{3.5}$) for simplicity. The color represents the stellar binary period, which could generate planetary orbital precession large enough for secular spin-orbit resonances and large spin-axis variations. It illustrates that the stellar binaries that typically host circumbinary planets (with period beyond seven days) are too large to excite spin-axis variations due to both secular spin-orbit resonances and strong spin-orbit coupling, and thus we would expect circumbinary planets to have low spin-axis variations. 

The qualitative results are the same for different total masses of the stellar binaries. The maximum rotation rate we include in the plot is the breakup angular velocity ($\sqrt{G m_p/R_p^3}$). Figure \ref{fig:spin_orb} shows that secular spin-orbit resonances could only occur when stellar binary period is less than 1 day, unless the mass ratio of the stellar binary is very small ($q \le 0.05$) and when the planet rotates much faster (($\gtrsim 5 \omega_\oplus$). However, so far, no circumbinary planets have been detected that orbit stellar binary with a period less than seven days. This is likely due to the formation processes that leads to the orbital decay of the stellar binary to short periods via Kozai-Lidov oscillations \citep{Armstrong14,Miranda15,Martin15}. Thus, it is challenging to excite the obliquity of circumbinary planets via secular spin-orbit resonances.

\section{Single Planet around Binary Stars}
We now develop a rigorous secular theory to gain a deeper understanding of the circumbinary planet's spin dynamics. This will also quantitatively show why having a second star affects planetary spin in ways drastically different from having a planetary companion, even though both setups involve three bodies. 

\subsection{Secular Theory} 
In order to accurately characterize the spin dynamics, the planet is modeled as a rigid body instead of a point mass (as shown in Fig. \ref{fig:ref_frame}). Existing spin-axis secular theories in the literature focused on single star systems, where the mean anomaly increases linearly with time \citep{Farago09}. This is no longer the case for circumbinary planets, whose mean anomalies change at time-varying rates. 

To correctly characterize the long-term effective behavior, we adopt rigorous normal-form-based treatment for averaging over multiple nonlinear angles (i.e., multi-phase averaging), which physically corresponds to averaging over time instead of the mean anomaly and spin oscillations of a planet, in a Hamiltonian setup. For more details on the normal form technique, see for instance \citep{lochak2012multiphase,sanders2007averaging}; we also note two distinctions between this work and the rich field of restricted 3-body dynamics \citep{szebehely2012theory,llibre1990elliptic,koon2000heteroclinic,kumar2021high} to which these tools also apply are: (i) the latter considers point masses only, thus no spin; (ii) orbital resonance is absent or negligible in our considered physical parameter range. 

The original, unaveraged Hamiltonian in canonical variables (namely Delaunay variables for the orbit and Andoyer variables for the spin-axis) is denoted by 
\begin{align}\label{eq:H_org}
\begin{split}
      H \left( \mathcal A, \mathcal D, l_*, L_* \right) = T^{linear}\left( \mathcal D \right) + T^{rot}\left( \mathcal A \right) + V\left( \mathcal A, \mathcal D, l_*, L_* \right)
\end{split}
\end{align}
where $T^{linear}$ is the linear kinetic energy of the rigid body, $T^{rot}$ is the rotational kinetic energy of the rigid body in Andoyer variables (see Appendix C) and $V$ is the gravitational potential generated by the binary stars on the oblate planet, $l_*$ is the mean anomaly of the inner orbit, and $L_*$ is its conjugate momentum. 

Appropriately averaging over the stellar binary's orbital phase, planet's orbital phase and planet spin, we can obtain the secular dynamics of the planetary obliquity, governed by equation \ref{eq:approximated_dynamics} (See Appendix D for derivations). 

The secular dynamics of the planetary obliquity is governed by the following ODE system (see Appendix D):
\begin{align}
\left\{
    \begin{aligned}
        \dot{X} &= \sin(h) \left[ C_1 X \sqrt{1-X^2} + 4 C_2 \cos (h) (1-X^2) \right], \\
        \dot{h} &= \frac{C_1 \cos (h) \left( 1-2X^2 \right)}{\sqrt{1-X^2}} - 2 C_2 X \cos(2 h)  + C_3 + 2 C_4 X, \\
    \end{aligned}
\right.
\label{eq:approximated_dynamics}
\end{align}
where $X=H_a/G_a$ is the cosine value of the angle between spin's angular momentum and $\bm E_3$ (or the spin angle as illustrated in Figure \ref{fig:ref_frame}), and $h=h_d-h_a$ is the phase difference between the precessions of the planet's orbit and the planetary spin. This dynamics (equation \ref{eq:approximated_dynamics}) corresponds to the Hamiltonian 
\begin{align}
\begin{split}
    H(X, h) = & C_1 \sqrt{1-X^2} X \cos (h) + C_2 \left( 1-X^2 \right) \cos(2h) \\
              & +C_3 X+C_4 X^2. \\
\end{split}
 \label{eq:avg_H}
\end{align}

Similar to the single star case \citep{Laskar93a}, $X, h$ together give a secular approximation of the planetary spin dynamics. We note that planetary orbital dynamics is not affected by the spin-axis variations in the secular limit. This can be understood intuitively since the angular momentum of the planetary spin is much lower than that of the orbit. However, the spin-axis evolution is sensitive to the orbital oscillations. In particular, in the secular limit for planets around the stellar binary, planetary orbital shape (semi-major axis $a$ and eccentricity $e$) and inclination are constants over time, and the orbital orientation (argument of pericenter $g_d$ and $h_d$) is changing slowly due to the stellar binary quadrupole potential \citep{Farago09}. $C_{1-4}$ in our secular equations of motion (equation \ref{eq:approximated_dynamics}) depend on the constant parameters of the planetary orbit and the masses of the planet and the stars. Different from the single star case, the Hamiltonian is more complicated, and $C_3$ goes to zero for a single star system.
 
For circumbinary planets in physical parameter range, the value of $C_3$ is much bigger than those of $C_{1,2,4}$ (see e.g., Figure \ref{fig:c1c2c3c4vsip}) unless the planet's orbit is nearly orthogonal to the binary's. These values, according to our secular equations of motion, lead to $X$  being nearly a constant and $h$ linearly changing with time (note $h \in \T$ , not $\R$), which corresponds to precession of the spin axis along the binary orbital normal, and near constant spin angle. This is how the developed secular theory elucidates the low obliquity variation. 

\begin{figure}
    \includegraphics[width=0.8\linewidth]{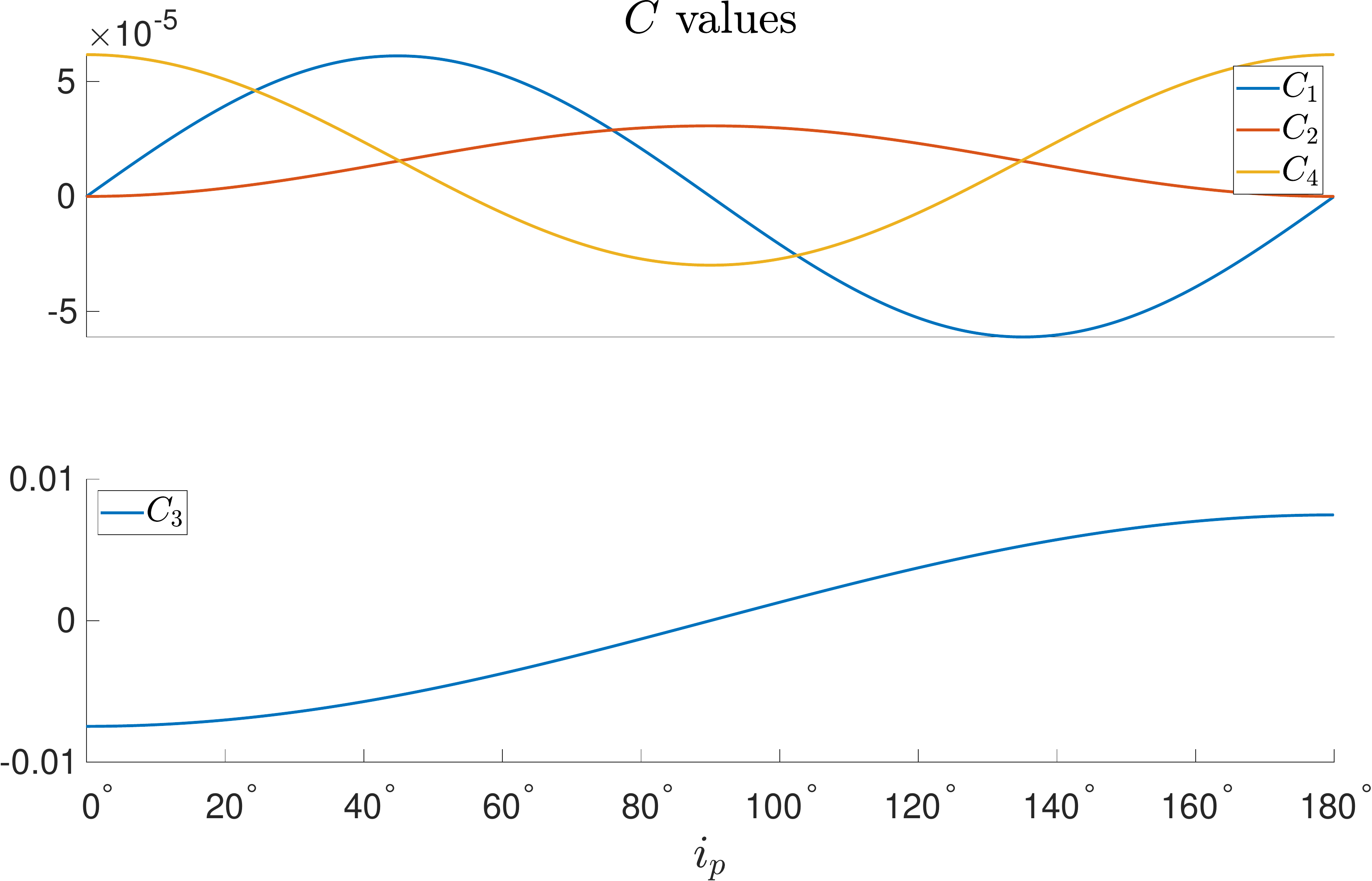}
    \centering
    \caption{How $C_1$,$C_2$,$C_3$,$C_4$ vary with the mutual inclination $i_p$ in \textit{Kepler-47} system.
    }\label{fig:c1c2c3c4vsip}
\end{figure}

The physics of the nearly orthogonal cases, on the other hand, are only more different from that in single star systems (note that it is not yet clear whether such nearly orthogonal cases are common or even existent, as observations are still insufficient; in this sense, this part could just be a theoretical prediction). More precisely, if $i_p$ increases while other parameters of the system remain fixed, the dynamical system given by equation \ref{eq:approximated_dynamics} will undergo two bifurcations sequentially and switch between topologically different dynamics. For Kepler-47 system as an example (see Figure \ref{fig:bifurcation_change_i}), the 1st bifurcation occurs near $i_p=89^\circ$, where the number of fixed points jumps from 2 to 4, and the 2nd bifurcation occurs near $i_p=89.99^\circ$, where the number of fixed points jumps again from 4 to 6. The first bifurcation exists in single star systems as well and is due to the $1:1$ spin-orbit resonance \citep{neron1997long,Saillenfest19}; however, the second bifurcation, and all the dynamical features after $i_p$ exceeds this critical value, are unique to the circumbinary system. 

\begin{figure}
    \centering
    \includegraphics[width=0.8\linewidth]{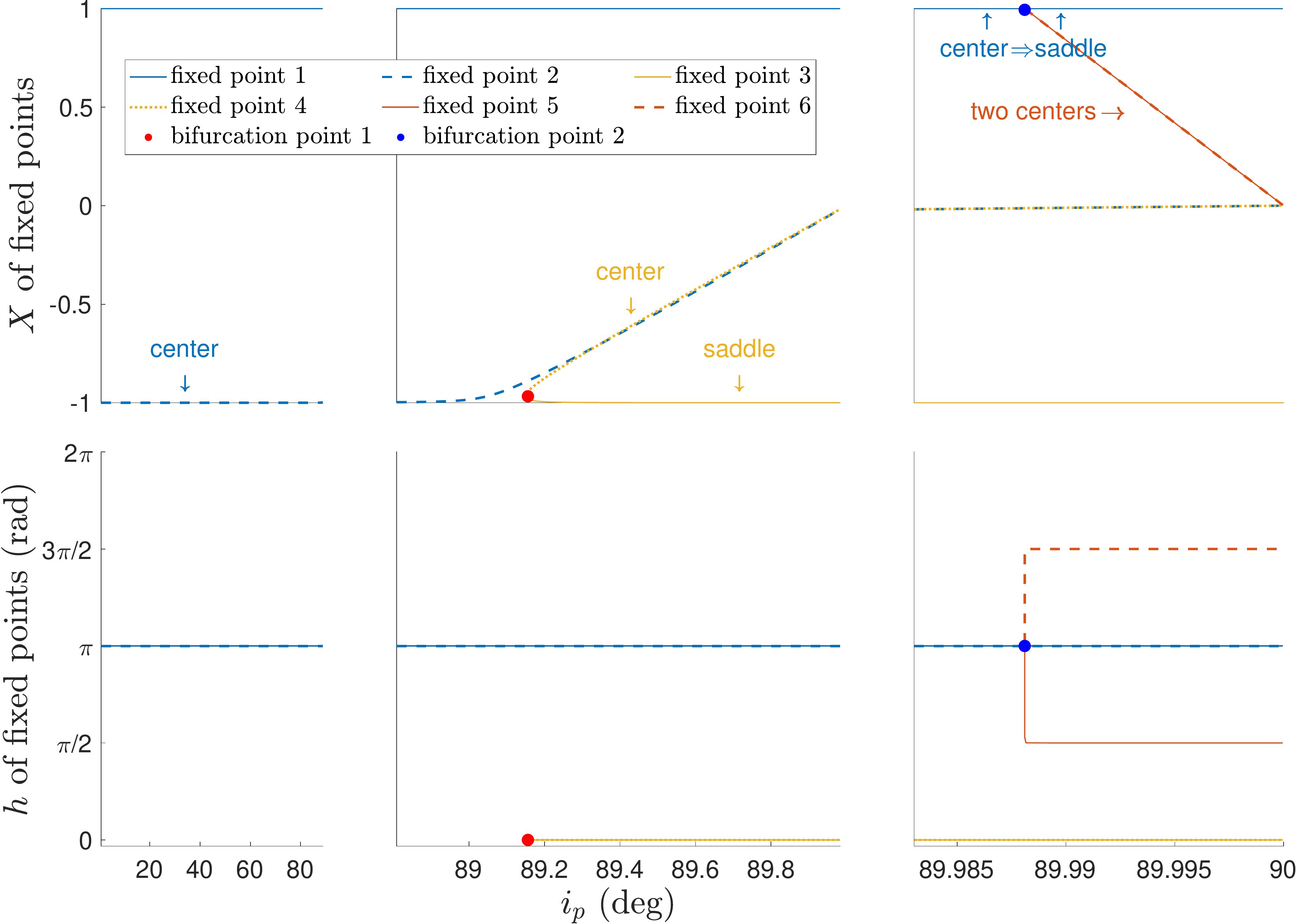}
    \caption{Bifurcation diagram with the varying parameter being $i_p$.
        Three plots with different scalings of $i_p$ axes are concatenated together.
        Each family of fixed points is denoted by a color. Dots indicate bifurcation locations.
    \textbf{Bifurcation point 1} is a Hamiltonian saddle-node bifurcation; \textbf{Bifurcation point 2} is a Hamiltonian pitch-fork bifurcation.
    }\label{fig:bifurcation_change_i}
\end{figure}

In particular, the phenomenon unique to circumbinary planets is the creation of a second saddle point out of a previously stable fixed point, when the mutual inclination between the inner and outer orbits reaches nearly $90^\circ$. Accompanying this change of stability is the emergence of two additional centers, which correspond to configurations for which torques from the two stars cancel out. This phenomenon is analogous to the well-known (supercritical) pitchfork bifurcation, in which a stable fixed point loses its stability and two more stable fixed points emerge nearby, however in a Hamiltonian setting (see \cite{lyu2021hamiltonian} for a similar but more local case in chemical reactions). Physically, this means a previously stable spin-axis configuration changing into an unstable one, accompanied by two libration-like stable regions emerging nearby, encircled by larger scale transports outside. For illustration, the phase portraits corresponding to fixed $i_p$ values respectively in the three different regimes in Figure \ref{fig:bifurcation_change_i} are plotted in Figure \ref{fig:phase_portrait}. To better observe the spin dynamics in 3D space, these portraits are wrapped from 2D phase portraits onto the sphere (Figure \ref{fig:phase_portrait}), demonstrating the trajectories of the spin direction (with respect to the planetary orbit) as the level sets in the sphere. 

\begin{figure}
    \centering
    \includegraphics[width=0.8\linewidth]{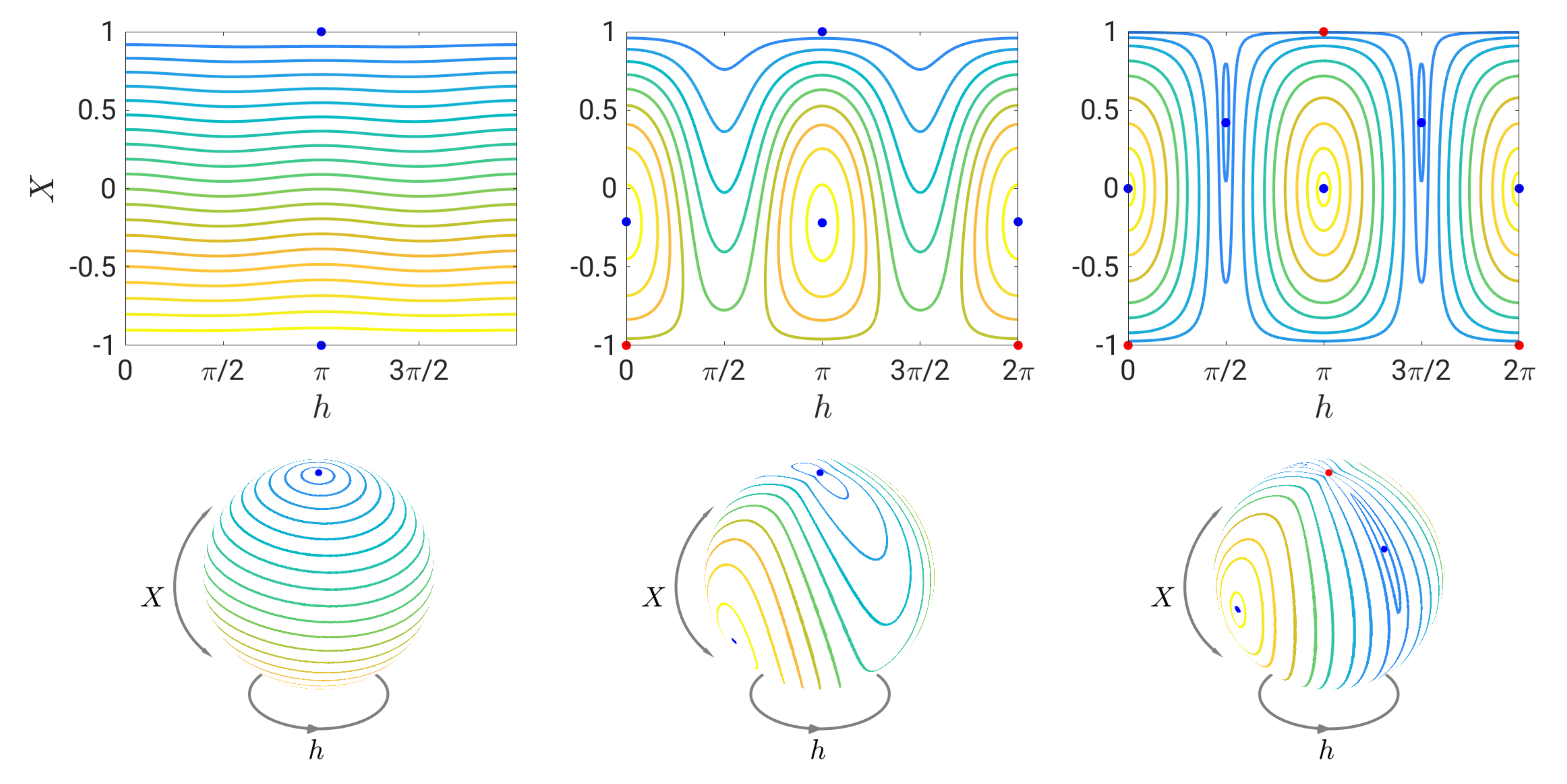}
    \caption{Phase diagrams of the $(X,h)$ dynamics in equation \ref{eq:approximated_dynamics} with the varying parameter $i_p$ (in Figure \ref{fig:bifurcation_change_i}) fixed as $i_p=80^\circ,\,89.8^\circ,\,89.995^\circ$ from the 1st column to the 3rd column  respectively.
    The phase portraits in the second row are the maps of corresponding phase portraits above in spheres with $X$ the latitude coordinate and $h$ the longitude coordinate. Dots indicate fixed pts, blue being nonlinear centers and red being saddles.}
    \label{fig:phase_portrait}
\end{figure}

In summary, our main discovery, namely the reduced spin-angle variation due to the existence of two host stars, is warranted by the quantitative behavior of the secular dynamics (equation \ref{eq:approximated_dynamics}) for inclinations away from $90^\circ$, and this mechanism persists till rather large inclination values (e.g., $\le 88^\circ$). The behavior near $90^\circ$, on the other hand, is not only quantitatively but in fact topologically different from the single star cases, but at this moment it is just a theoretical prediction of our secular theory. 

\subsection{Case study of the Kepler-47 system based on both the secular theory and full rigid-body simulations} 
Now we consider an illustrative example to quantify the spin-axis variations using both the secular theory and the state of art rigid-body simulation package. For relevance to habitability, we investigate in Earth-like planets in near circular orbits. We use the stellar binary Kepler-47AB as the illustrative example ($a_*=0.0836AU$, $M_{*_1} = 1.043 M_\odot$, $M_{*_2} = 0.362 M_\odot$)~\citep{orosz2012kepler}, but only one planet in the system is included for now (see later section for the multi-planetary case). The planet mass, semi-major axis/rotation period, eccentricity, and oblateness are set to be the same as that of the Earth; however, a wide range of inclinations is considered. 

We use both the secular theory and rigid-body simulations to investigate the problem. For the rigid-body simulations, we use the recently developed \texttt{GRIT} package \citep{chen2021grit}. The simulation is accurate and efficient as it is based on high-order symplectic Lie-group integrator, which can be viewed as a modern variation of the seminal Lie-Poisson integrator proposed in \citep{Touma94}. The consideration of rigid bodies instead of point masses allows high-fidelity numerical simulations of the spin-orbit coupling of circumbinary planetary systems, which serve two purposes: (i) to validate our secular theory, and (ii) to demonstrate, by sweeping over the parameter space, that stabilization of the long-term obliquity dynamics is a rather general phenomenon for low inclination planets relative to the stellar binary (note that higher planetary inclination can lead to large obliquity variations, due to precession of the planetary orbit around the orbital normal of the stellar binary).

The upper panel of Figure \ref{fig:spin_angle_vs_t} illustrates the dynamical variation of the spin-axis, measured by the tilt of the spin-axis relative to the invariable plane (the plane normal to the total angular momentum of the system). The solid lines represent the rigid-body simulations, and the dashed lines represent our secular results. It shows that the secular approximation agrees very well with that of the rigid body simulations. More importantly, it can be seen that the spin-angle variations remain very small ($\le 1^\circ$) for all mutual inclinations. This is consistent with our theoretical prediction (see above `Secular Theory' subsection and additional `Secular Theory' subsection below) that the effect of a stellar `perturber' is different from that of planetary perturbers; for example, both Moon-less Earth and Mars are known to have large and chaotic spin-angle variations as a result of secular spin-orbit resonances due to the perturbation of the planetary companions \citep{Touma93,Laskar93a}. These analytical and numerical results also agree with the Heuristic Intuition: when orbiting around a stellar binary, the spin precession frequencies are much slower than the orbital precession frequency due to the large quadrupole potential of to the stellar binary in the center. Thus, the secular spin-orbit resonances are absent, and the spin-angle only has very small oscillations. 

\begin{figure}
\centering
\includegraphics[width=0.8\linewidth]{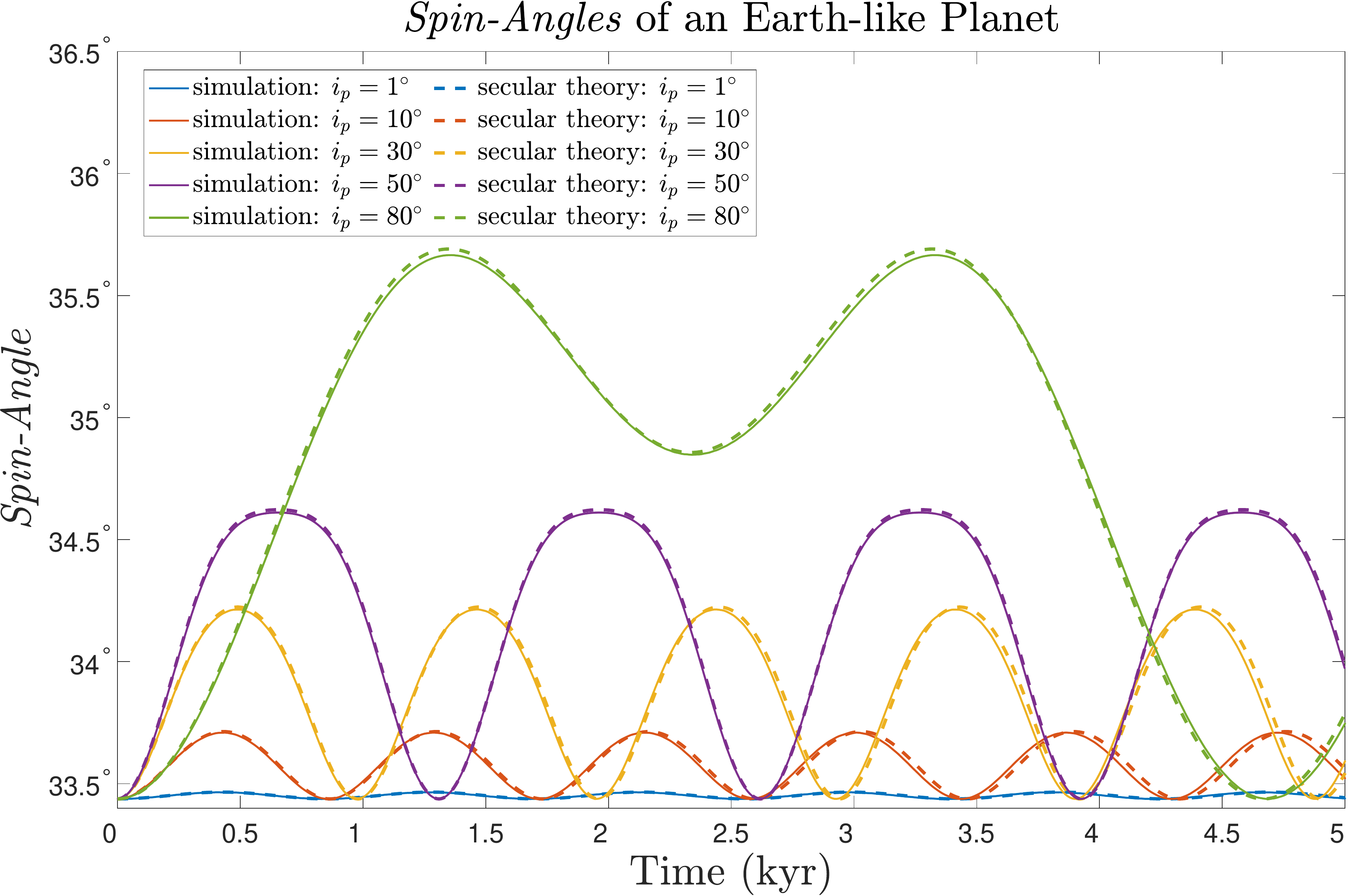}
\includegraphics[width=0.8\linewidth]{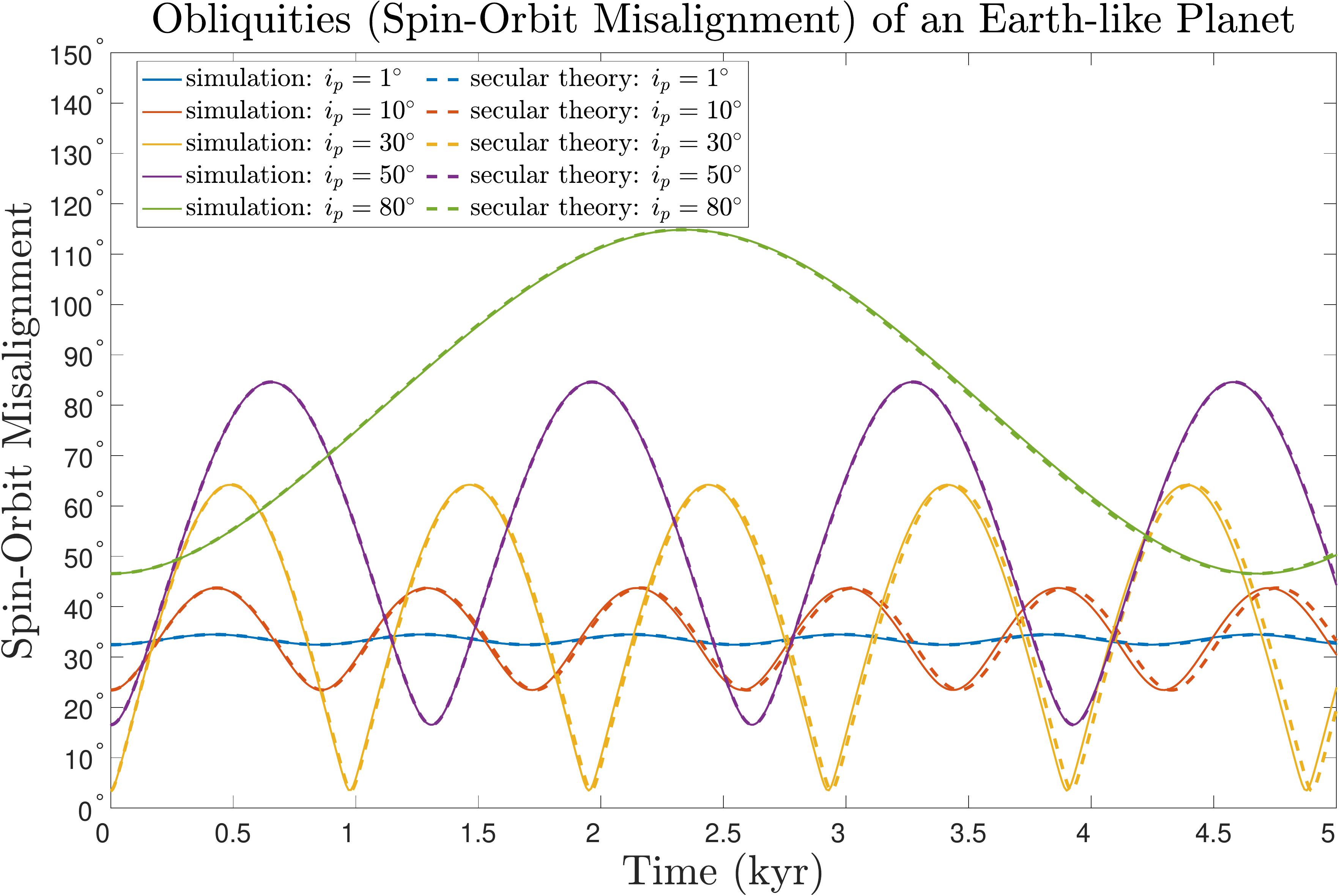}
\caption{\textit{Spin-Angle} (upper panel) and obliquity (lower panel) versus time for an Earth-like planet orbiting around a binary star. The binary stars' orbital parameters were set according to the stellar binary properties of \textit{Kepler-47} system, and the planet was chosen to be \textit{Earth-like}.
Multiple colors represent different mutual inclinations between the planet and the stellar binary,
and the solid / dashed lines represent results from rigid body simulation / secular theory. It shows that the spin angles have only small variations,
and the obliquity variations are mainly due to planetary orbital variations.}\label{fig:spin_angle_vs_t}
\end{figure}


On the other hand, obliquity is affected by both orbital variations and spin-axis variations. Specifically, the spin-axis angle is the angle between the planetary spin-axis (i.e. $\overrightarrow{{\bm G}_a}$) and the normal to the reference plane (i.e. $\bm E_3$) as shown in Figure \ref{fig:ref_frame}. The obliquity is the angle between the spin-axis and the normal to the planetary orbital plane. Variations in both orbital orientation and the spin-axis affect the obliquity. Its variation is shown in the lower panel of Figure \ref{fig:spin_angle_vs_t}. Specifically, obliquity variations are low when planets are nearly coplanar with the stellar binary. However, when planetary orbits are inclined with respect to the stellar binary, orbital precession leads to larger obliquity variations while the spin-angle variation remains small. 

\subsection{Survey in the parameter space of circumbinary planetary systems} The example in Figure \ref{fig:spin_angle_vs_t} shows that for a nearly coplanar circumbinary system analogous to Kepler-47, both the spin-angle and the planetary obliquity variations are small, due to the lack of secular spin-orbit resonances. Does this feature persist for a single planet orbiting around generic stellar binaries? Typical ranges of properties of stellar binaries hosting circumbinary planets are not yet well understood, due to the limited sample size of observed transiting systems (11 so far). However, existing knowledge includes that (i) the stellar binary hosts typically have orbital periods longer than seven days, which are larger than those of eclipsing stellar binaries ($\sim 3$ days) \citep{Armstrong14}; (ii) the mass ratio of the binaries is consistent with that of the stellar binaries in the field (roughly uniform) \citep{Martin19}. In order to make a robust claim, we consider a broad range of binary configurations. We sample through the parameter space of circumbinary planetary systems to quantify the ubiquity of stable obliquity dynamics for low inclination planets.  


We enumerate stellar binary configurations by varying $M_{*_2}/M_{*_1} \in [0.01,1]$ and $a_* \in [0.05,0.3]$. We set the sum of the stellar binaries to be one solar mass for an intuitive comparison with Solar system. The lower limit of the mass ratio is set so that the binary components are both with stellar masses. The minimum semi-major axis corresponds to orbital periods of $\sim$4 days, and the maximum semi-major axis corresponds to orbital period of $\sim$60 days. The maximum semi-major axis is set so that a planet at $1$AU remains stable \citep{Holman99}. Then, we calculate the amplitude of the spin-angle variations of an Earth-like planet with $1^\circ$ inclination (from the orbital plane of the stellar binary) using our secular theory to illustrate the effects quantitatively.

Results are summarized in Figure \ref{fig:survey1}, in the plane of stellar binary semi-major axis and mass ratio. The spin angles mostly have variations less than $1^\circ$  for stellar mass ratio larger than $0.01$ and stellar separation larger than $0.05$ AU (corresponding to a binary period $\ge 4$ days). Note that we include brown dwarfs in this calculation to survey a bigger parameter space for the binaries, in order to assure the robustness of our results. These results suggest that the latitudinal distribution of the stellar radiation (insolation) on the planet has small variations for a wide range of Earth-like planets around stellar binaries in the near co-planar configurations. This phenomenon can be understood intuitively, because higher stellar binary mass ratio and separation both lead to a larger quadrupole momentum at the center of the planetary orbit. This drives fast orbital precession and avoids secular spin-orbit resonances and variations in the planetary spin-axis. Because circumbinary planets are more likely formed in a coplanar configuration around wide orbit stellar binaries above $\sim 7$ days \citep{Miranda15,Martin15}, it is unlikely to have large spin-axis variations for near co-planar Earth-like planets around stellar binaries. 

\begin{figure}
\centering
\includegraphics[width=0.9\linewidth]{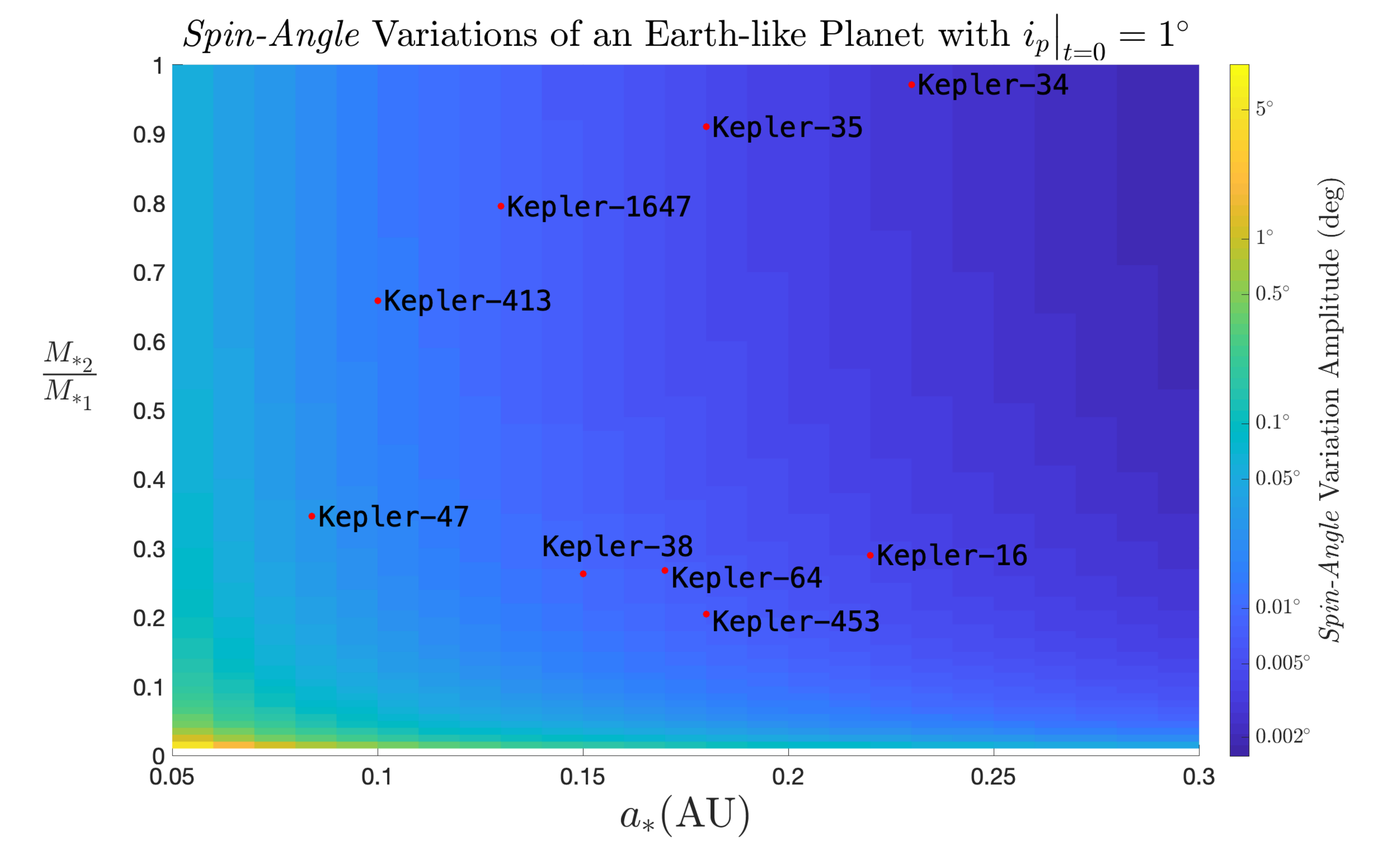}
\vspace{-4mm}
\caption{Spin-Angle variations of circumbinary planets around different types of stellar binaries. The $x$-axis is the semi-major axis of the orbit of the stellar binary, and the $y$-axis is the mass ratio of two stars.
The color represents the largest \textit{spin-angle} variation for an \textit{Earth-like} planet.
Red dots represent currently observed circumbinary planetary systems.
It illustrates that the amplitudes of the spin angle variations are very small ($\lesssim$ the variation of the Earth) for a wide range of binary systems, except when the binary separation and mass ratio are both low and the quadrupole moment of the binary is small, corresponding to the lower left corner of the figure.
}\label{fig:survey1}
\end{figure}

\section{Circumbinary Systems with Multiple Planets}

Similar to single star systems, it is common to have multiple planets in a circumbinary system. An example is the observed Kepler-47 system with at least 3 planets. Interactions between planets can give extra perturbations to their orbital dynamics and increase the likelihood of secular spin-orbit resonances and large amplitude obliquity variations. 
\vspace{-2mm}
\subsection{A case study of a modified Solar System with binary stars} How variable would the spin-axis be, if Earth orbited around a stellar binary with its Solar System planetary companions? To focus on the stabilization effect of the binary, we consider the moon-less Earth with its seven companion planets in the solar system. The stellar binary is assumed to be in the elliptic plane. We arbitrarily set the stellar binary to be composed of a  $0.7M_\odot$ star and a $0.3M_\odot$ star, so the sum of the masses is the same as our own Sun. We set the semi-major axis of the binary to be $0.05$AU, so that the effect of the stellar binary is strong comparing with that of the planets but the separation is not wide enough to create instability in millions of years (this was numerically verified). 

Figure \ref{fig:Earthob} shows the obliquity variations of the Earth-like planet orbiting around a stellar binary (dashed lines) and around a single star, obtained from symplectic rigid-body simulations. We include a wide range of initial obliquities, since exoplanets could have a large range of obliquities due to scattering, collisions and disk turbulence \citep{Hong21, Li21, Jennings21}. Different colors correspond to different initial obliquities. The case of the moonless Earth agrees with the secular results and the rigid body simulations in the literature \citep{laskar1993chaotic,Lissauer12}, which show large amplitude obliquity oscillations when the obliquity is below $\sim 50^\circ$. This is due to the overlap of the secular spin-orbit resonances. Unlike the case of a single-star Moon-less Earth, the obliquity is nearly stable for all initial obliquities around the stellar binary, since the fast orbital nodal precess detunes the system out of the secular spin-orbit resonances. It shows the Earth's obliquity can be stabilized by the stellar binary even without the Moon. Note that we set the planetary orbits to be near coplanar with that of the stellar binary. Thus, the low spin-angle variations leads to low obliquity variations. With higher orbital inclinations, orbital nodal precession would lead to large obliquity variations of the planets while leaving the spin-angles fixed.

\begin{figure}
\centering
\includegraphics[width=0.9\linewidth]{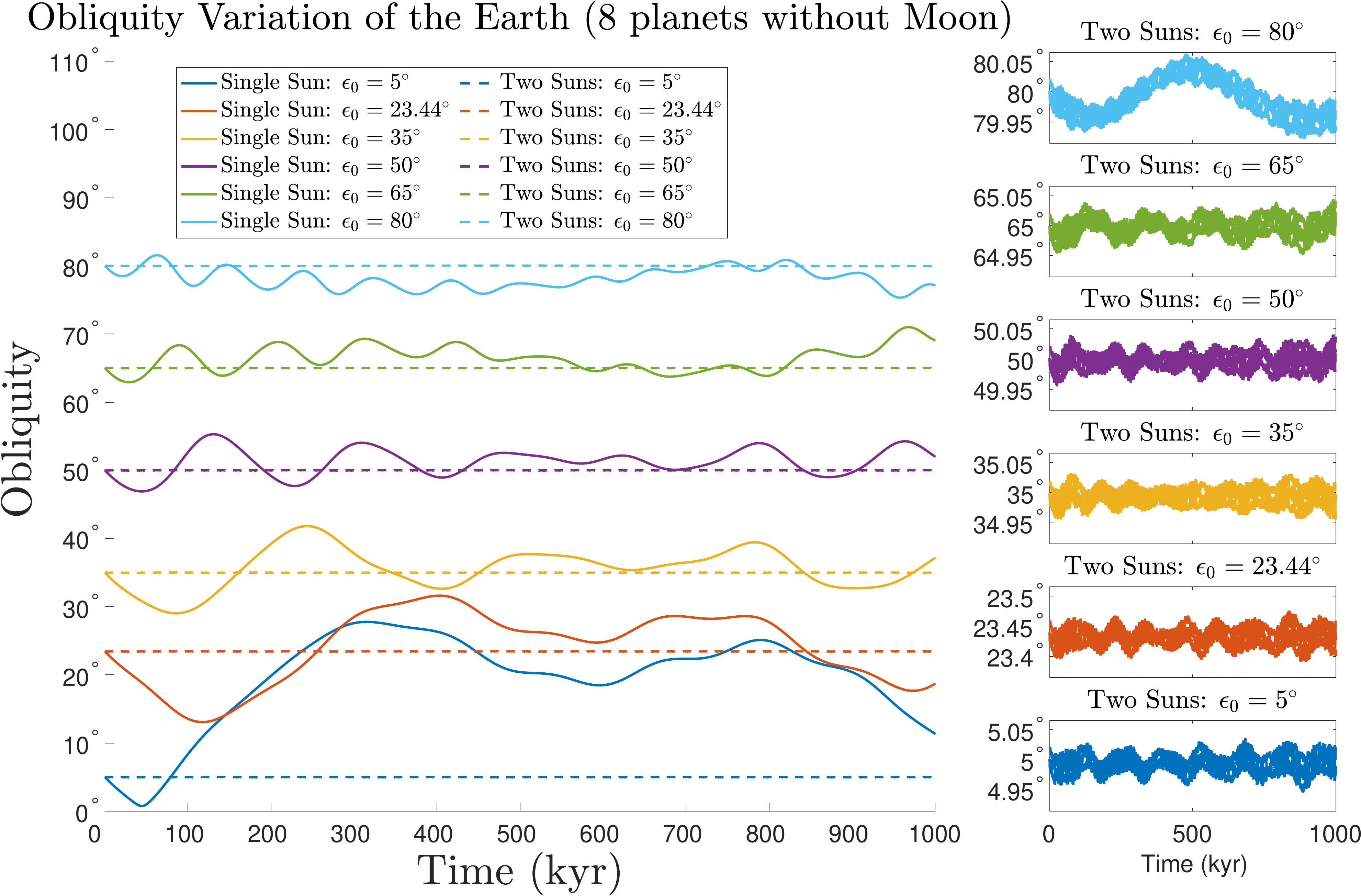}
\vspace{-3mm}
\caption{Obliquity variations of the Moon-less Earth, orbiting Sun (solid) and orbiting a stellar binary (dash), with companion planets. Similar to the results in \citep{laskar1993chaotic,Lissauer12}, the Moon-less Earth has large obliquity variations when obliquity is below $50^\circ$. On the other hand, if orbiting a binary, the Earth's obliquity would be stable even without the help of the Moon.}\label{fig:Earthob}
\end{figure}

\subsection{Secular theory: binary stars + planet spins + Lagrange-Laplace}
How robust is this result when one considers a larger parameter space? We first investigate analytically by combining the Laplace-Lagrange framework with our secular theory. As a result, we obtained an effective Hamiltonian that approximates multi-planetary circumbinary systems in low eccentricity and inclination. Note that we assume the eccentricity and inclination to be low for a more favorable condition for habitability. 

In single star systems (e.g., the Solar System), the oscillation frequency of a planet's orbital orientation can be estimated using the Laplace-Lagrange method \citep{murray1999solar}. Here, we combine Laplace-Lagrange method with our circumbinary secular theory to approximate the orbits of multiple circumbinary planets, so that whether obliquity variations are large can be predicted \citep{Laskar93a}. Under the common physical approximation that planetary perturbation on the stellar binary is negligible, we assume the stellar binary to have fixed Keplerian orbits. The Hamiltonian of the system is the following:  
\begin{align} \label{eq:H_orbit_cirumbinary}
\begin{split}
H(\bm{q},\bm{p})
= & \sum_{i=1}^n \frac{\bm{p}_i^T \bm{p}_i}{2 m_i} 
- \frac{G M_{*_1} m_i}{\norm{\bm{q}_i-\bm{q} _{*_1}}}
- \frac{G M_{*_2} m_i}{\norm{\bm{q}_i-\bm{q} _{*_2}}}
+ \mathfrak{R}_i^{planets},
\end{split}
\end{align}
where $q_i, p_i$ are respectively planet $i$'s position and momentum,
$q_{*_1}$ and $q_{*_2}$ are time-dependent locations of the stars, 
$M_{*_1}$ and $M_{*_2}$ are stellar masses,
and $\mathfrak{R}_i^{planets}$ is known as the disturbing function accounting for the gravity amongst the planets \citep[e.g.,][]{murray1999solar}. 

In order to obtain a nearly-integrable form amenable to analysis, we split the gravitational potential between the planets and the stars into two parts and rewrite the Hamiltonian as:
\begin{align} \label{eq:H_orbit_cirumbinary_2}
H(\bm{q},\bm{p})
= \sum_{i=1}^n H_i^{kepler} + \left[ \mathfrak{R}_i^{*_1} + \mathfrak{R}_i^{*_2} + \mathfrak{R}_i^{planets} \right].
\end{align}
Here $H_i^{kepler}$ is the Hamiltonian of a two-body system, composed of two stars merged at their center of mass and the $i^{th}$ planet, which alone would produce a Keplerian orbit. $\mathfrak{R}_i^{*_1}$ and $\mathfrak{R}_i^{*_2}$ are disturbing functions modeling (exact) corrections of the gravitation potential generated by the binary due to separations from their center of mass. $\mathfrak{R}_i^{planets}$ are again inter-planet potentials.

We then decompose the perturbative non-integrable part $\sum_{i=1}^n \left( \mathfrak{R}_i^{*_1} + \mathfrak{R}_i^{*_2} + \mathfrak{R}_i^{planets} \right)$
into two groups:
the group due to having two stars $\sum_{i=1}^n \left( \mathfrak{R}_i^{*_1} + \mathfrak{R}_i^{*_2} \right)$, and the group due to having companion planets: $\sum_{i=1}^n \mathfrak{R}_i^{planets}$. 

Under the influence of the stellar binary (the first group), the planetary orbit would precess with a near constant angular frequency ($D_i$), in the low eccentricity and low inclination limit (see~\citep{Schneider94}, and Appendix), where the orbital inclination $I_i$ is fixed and the longitude of ascending node $\Omega_i$ decreases at constant rate $D_i$. Then, the dynamical evolution under the disturbing potential of $\sum_{i=1}^n \left( \mathfrak{R}_i^{*_1} + \mathfrak{R}_i^{*_2} \right)$ follows the expression below:
\begin{align}
    \left\{
        \begin{aligned}
            \frac{d h_i}{d t} & = D_i k_i, \\
            \frac{d h_j}{d t} & = 0, \quad j \neq i. \\
        \end{aligned}
    \right.
    \text{ and }
    \left\{
        \begin{aligned}
            \frac{d k_i}{d t} & = -D_i h_i, \\
            \frac{d k_j}{d t} & = 0, \quad j \neq i. \\
        \end{aligned}
    \right.
            \label{eq:dh_i}
\end{align}
where $h_i=I_i \cos\left( \Omega_i \right),\, k_i=I_i \sin\left( \Omega_i \right)$.

On the other hand, the approximated dynamics due to planetary companions (the latter group) can be expressed via Laplace-Lagrange theory as described in~\citep{murray1999solar}:
\begin{equation}
    \frac{d \bm{h}}{dt} = \bm{A} \cdot \bm{k}, \quad
    \frac{d \bm{k}}{dt} = -\bm{A} \cdot \bm{h}
\label{eq:laplace_lagrange}
\end{equation}
where $\bm{A} =
\begin{bmatrix}
    a_{11} & \cdots & a_{1n} \\
    \vdots & \ddots & \vdots \\
    a_{n1} & \cdots & a_{nn} \\
\end{bmatrix} \in \R^{n \times n}$ is a constant matrix that depends on planetary semi-major axes and masses of planets, satisfying $\sum_j a_{ij}=0$. 

For a first-order perturbative approximation, it can be computed that the effective contributions of these two nonresonant fast processes (quadrupole contribution of the binary and planet-planet interaction) are additive, and we have the following approximated dynamics
\begin{equation} \label{eq:LL_circumbinary}
    \frac{d \bm{h}}{dt} = \bm{B} \cdot \bm{k}, \quad
    \frac{d \bm{k}}{dt} = -\bm{B} \cdot \bm{h}
\end{equation}
where $\bm{B}=\bm{A}+\bm{D}$ and $\bm{D} = \begin{bmatrix}
    D_1 & \cdots & 0  \\
    \vdots & \ddots & \vdots  \\
    0 & \cdots & D_n
\end{bmatrix},$
and this system is a linear system with coefficient matrix
$\begin{bmatrix} \bm{0} & \bm{B} \\ -\bm{B} & \bm{0} \end{bmatrix}$.

For planets in the habitable region close to the stellar binary, the effect of the stellar binary dominates that of the planetary companions (as illustrated numerically in the survey of the parameter space in the Results section). Therefore, in most of the cases $A \ll D$ and this will thus be assumed  (the validity of this assumption will be illustrated in the next section). Under this condition, 
$\bm{B}$ is diagonalizable with real eigenvalues $\lambda_1, \lambda_2, \ldots, \lambda_n$ satisfying 
$\lambda_i \approx D_i+A_{i}$, where $A_i$ is the $i^{th}$ diagonal element of $A$ (proof in Appendix B). Denoting the corresponding eigenvectors by $\bm{v}_1, \bm{v}_2, \ldots, \bm{v}_n$, solutions to equation \ref{eq:LL_circumbinary} are in the form
\begin{equation}
\left\{
\begin{aligned}
    \bm{h} &= \sum_j C_j^{(1)} \sin\left(\lambda_j t\right) \bm{v}_j + C_j^{(2)} \cos\left(\lambda_j t\right) \bm{v}_j,  \\
    \bm{k} &= \sum_j -C_j^{(2)} \sin\left(\lambda_j t\right) \bm{v}_j + C_j^{(1)} \cos\left(\lambda_j t\right) \bm{v}_j. \\
\end{aligned}
\right.
\end{equation}

Thus, similar to the single planet case, the inclination oscillation frequencies are dominated by the stellar potential and are typically much larger than that of the spin-axis precession frequencies. Therefore, secular spin-orbit resonances are avoided, and obliquity variations are typically still low for Earth-like circumbinary planets in the habitable zone; it is almost like a circumbinary system with one planet, even though planetary companions are actually present. 

\subsection{Parameter space survey based on the analytical theory}
To show that the mild obliquity variation of circumbinary multi-planetary systems is a robust phenomenon, we conduct an additional systematic study, this time of a binary star system with two planets. We set one of the planets to be an Earth-like planet and vary the mass and location of the other planet. Denote by  $p_1$  the Earth-like planet and  $p_2$  the other planet. For dynamically cold systems, we assume that the eccentricity and inclination of the planetary companion are low. 

To illustrate the difference in magnitude between the orbital precession due to the central stellar binary and the planet-planet interaction, we show in Figure \ref{fig:orbprec} an rough estimation of the precession rates of an Earth-like planet perturbed by a Jupiter mass companion at different semi-major axes. We locate the Earth-like plant at 1au from the central stellar binary, and we set the total mass of the binary to be one solar mass. The nodal precession rate due to the binary can be found in eqn \ref{eqn:orbprec}. The nodal precession rate due to the companion can be obtained following the Laplace-Lagrange secular theory \citep{murray1999solar} as follows: 

\begin{align}
    \dot{h}_{d, pp} = -\frac{n}{4}b^{(1)}_{3/2}(\alpha) \alpha \bar{\alpha}\frac{M'}{M_*}
\end{align}
where $\alpha = a/a'$, $\bar{\alpha} = \alpha$ if the perturber planet is external to the Earth-like planet, and it is unity otherwise. $b^{(1)}_{3/2}$ is a Laplace coefficient defined by:
\begin{align}
    b^{(1)}_{3/2}(\alpha) = \frac{1}{\pi} \int^{2\pi}_{0} \frac{\cos{\psi}}{(1-2\alpha \cos{\psi}+\alpha^2)^{3/2}}d\psi
\end{align}

\begin{figure}
\centering
\includegraphics[width=0.8\linewidth]{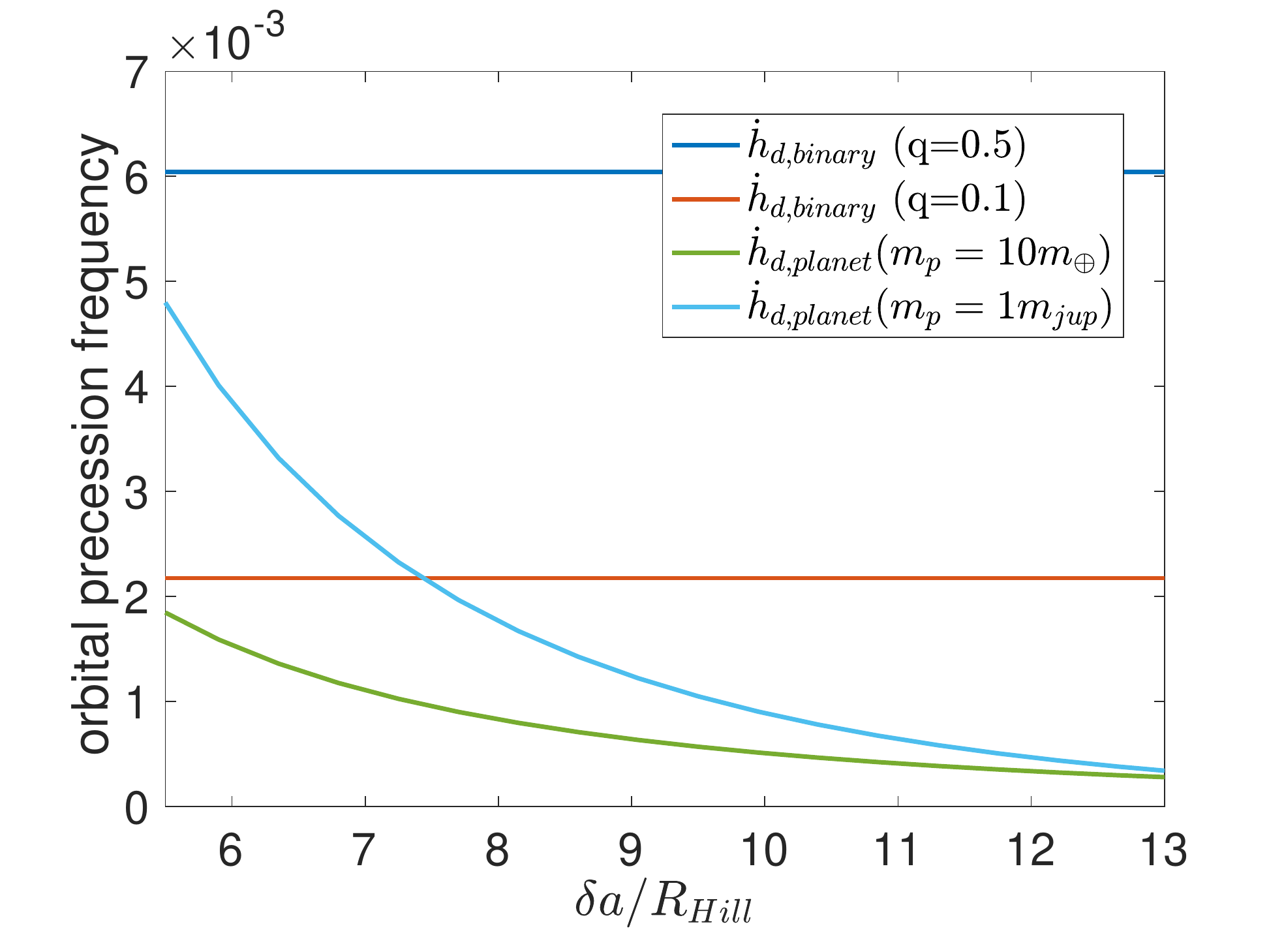}
\vspace{-3mm}
\caption{Orbital precession due to the central stellar binary and due to planetary companion. }
\label{fig:orbprec}
\end{figure}

We place the planetary companion $\delta a$ outside the orbit of the Earth-like planet in unit of mutual Hill radius of the planets ($R_{Hill} = (m_{p1}+m_{p2})/(3(M_{*1} +M_{*2})*(a_{p1}+a_{p2})/2$). It has been shown that planets around binaries generally have separations larger than $\sim 7 a_{Hill}$ to be long term ($10^8$ binary period) stable, and $\beta$ needs to be larger for longer stability timescales \citep{smullen2016planet}. Figure \ref{fig:orbprec} shows the orbital precession rates as a function of the planetary separation. The green line corresponds to the orbital precession due to a ten Earth mass companion, and the cyan line corresponds to that due to a Jupiter mass companion. The blue and red lines correspond to the precession due to the stellar binary with mass ratios of $q = 0.5$ and $q = 0.1$ separately. It shows that the precession due to the planetary perturber could dominate over that of the stellar binary only when the perturber is massive ($m_{p2}> m_{jup}$) and at close separation $\delta a <7.5 R_{Hill}$ mutual Hill radii, and when the stellar binary is having extreme mass ratio ($q<0.1$).

To explore a larger parameter space, we not only sample the parameter space of the stellar binary over possible $a_*$ and  $M_{*_2}/M_{*_1}$  values, but also scan through a wide range of semi-major axes for the planetary companion and consider planetary masses ($p_2$) in the range of $m_{p_2}\in [10^{-3},10^{-6}] M_\odot$. We calculate the inclination oscillation modal frequencies based on the secular approach described in Materials and Methods, and record the interval  ($\max\left(a_{p_2}\right) - \min\left(a_{p_2} \right)$) in which the modal frequencies are close to the spin-axis precession frequencies (for an overestimation, we define `close' as $\sim 0.8-1.2$ times the spin-axis precession frequency). The interval is considered to be an overestimation of $a_{p_2}$ values that may lead to secular spin-orbit resonances --- for this, we note that the exact proximity of the spin-axis precession frequency that could lead to the secular spin-orbit resonances depends on the width of the resonances, which is affected by the configuration of the circumbinary system (e.g., obliquity of $p_1$, the amplitude of the inclination oscillation, as well as the separation of the stellar binary) \citep{Ward73,Shan18}. In general, the larger the interval length, the higher chance the Earth-like planet would experience large spin-angle variations. 

We record in Figure \ref{fig:survey2} the maximum interval of semi-major axis values of the planetary companion that could lead to spin-axis resonances. The likelihood of larger obliquity variation depends on the mass of  $p_2$  and the separation of the stellar binary: when the planets are more massive, they can better compete with the perturbation from the star and allow a higher chance for spin-angle variations. However, including maximum planetary mass up to a Jupiter mass, the planets still need to be very close to each other ($\lesssim 0.1-0.3$AU) to allow large spin-angle variations. 

In short, large spin-angle variations may only occur with high probability in the circumstance of a heavy planet ($m_{p_2} > 10^{-3} M_\odot$) being very close to the Earth-like planet. We note that this is shorter than the typical separations between planets ($\sim 7 R_{hill}$ ) in order to be dynamically stable around stellar binaries \citep{smullen2016planet}. Specifically, $7 R_{hill}$  corresponds to $\sim 0.1-0.4$AU for Earth-like planet around Sun-like stars with Earth-like to Jupiter-like companions, and these minimum separations are larger if the planetary masses are higher. Thus, the yellow region in Figure \ref{fig:survey2} is unlikely to be physical because it won't allow orbital stability.  

\begin{figure}
\centering
\includegraphics[width=0.8\linewidth]{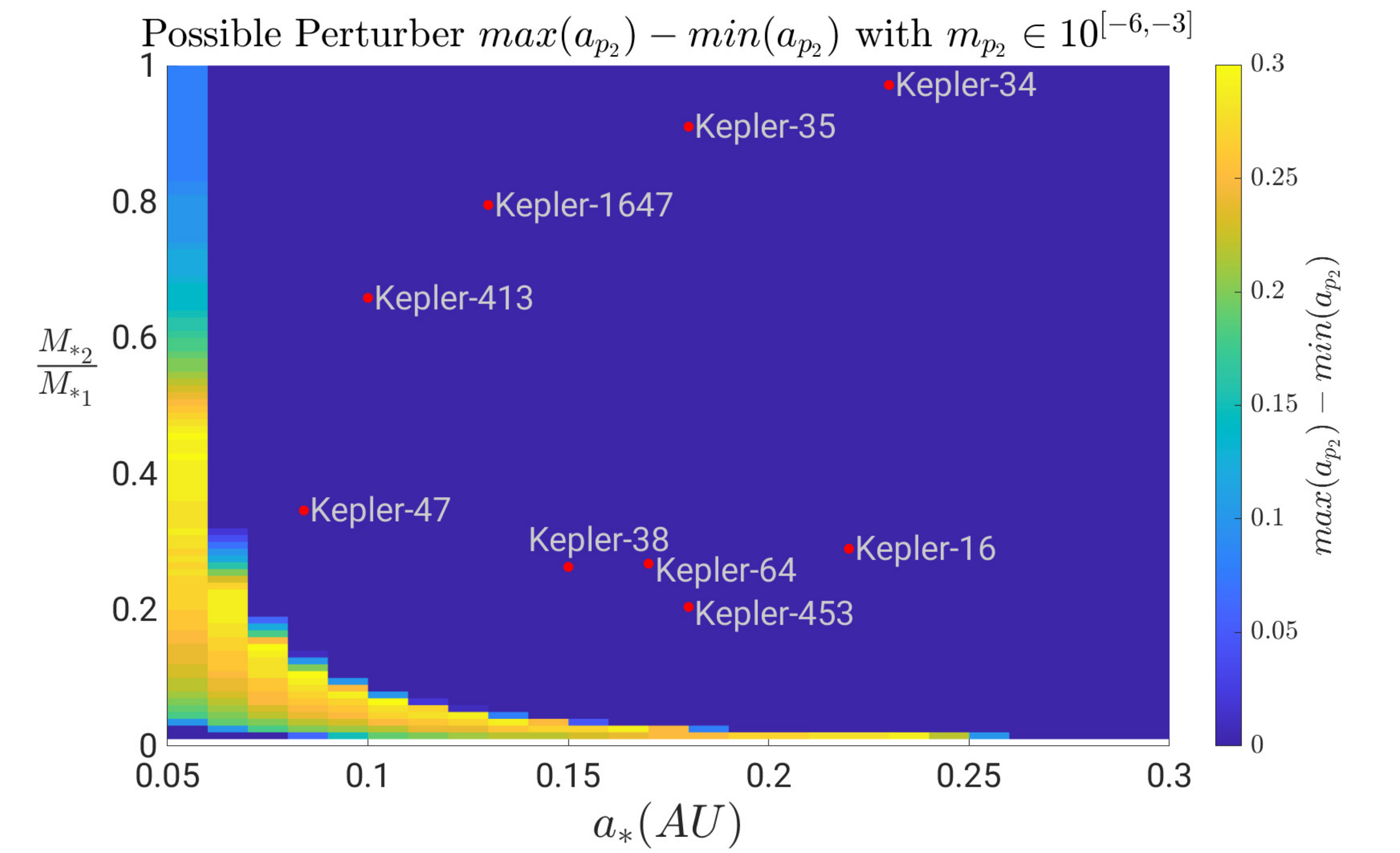}
\vspace{-3mm}
\caption{Region of the companion planet to allow obliquity variations. The $x$-axis is the semi-major axis of the orbit of the stellar binary $a_*$. The $y$-axis is the mass ratio of two stars $\frac{M_{*_2}}{M_{*_1}}$. For each choice of these two values as well as $p_2$'s mass log-uniformly sampled from $\left[ 10^{-3}, 10^{-6} \right] M_\odot$, we compute the interval of $p_2$'s semi-major values that can potentially place $p_1$ in spin-orbit resonance, and use color to represent the largest width of this interval (brighter means higher likelihood of larger variation), maximized over all enumerated $p_2$ mass values. Red dots correspond to observed circumbinary systems.}
\label{fig:survey2}
\end{figure}

\section{Discussion} 

Stellar binaries are common in the Solar neighborhood, and thus to understand the habitability of their planets, if any, is as important as that for single star systems. We adopt a dynamical approach, in which we show that the planetary spin evolution can be very different when a second host star is present. For planets orbiting around a stellar binary, the existing obliquity theory designed for single star systems becomes inapplicable. This is because torques from both stars act on the spin-axis of its planets, and planetary orbits deviate more significantly from being Keplerian. We investigated both analytically and numerically and discovered that the spin angle variations of circumbinary planets are typically much less pronounced than their single star analogue, because the stellar binary leads to much faster orbital nodal precession, which weakens the spin-orbit coupling and detunes the system from secular spin-orbit resonances. Thus, their obliquity (spin-orbit misalignment) variations are mostly low for planets residing near their stellar binary orbital plane (note that for planets with higher orbital inclination relative to the stellar binary, orbital precession around the stellar binary can lead to large obliquity variations). Since obliquity determines the stellar insolation on the planets, milder variations tend to correspond to more regular seasonal variations \citep{Deitrick18}. 

To meet the objective of this study, which is to understand the spin dynamics of Earth-like circumbinary planets, we first conduct a heuristic calculation to illustrate the intuition, then construct an analytical secular theory, from which more accurate and detailed results are obtained, and finally validate the results by high-fidelity numerical simulations. More precisely, the analysis is based on the fact that the spin-axis dynamics admits different behaviors over multiple timescales: the fast timescales correspond to planetary spins, stellar binary orbital rotation, and the planetary orbital rotations, and the slower timescales correspond to spin-axis and orbital variations. We average out the fast timescales to obtain the secular theory. We start with the single planet case for simplicity, where a secular theory is constructed and verified by state-of-the-art full rigid-body numerical simulations. The results are then generalized to multi-planet systems: we first obtain an analytical approach for multi-planetary systems by combining the Lagrange-Laplace method with our circumbinary secular theory, and then conduct a systematic numerical investigation. 

Our study has not considered satellites. Note that large moons ($\ge$ mass of our own Moon) close to the planets (slightly outside its Laplace radius  $\sim 15$ Earth-radius, where the moons orbits precess around the ecliptic) could significantly increase the precession rate \citep{Li14prelate}, and this can lead to enhanced spin-axis variations. Thus, contrary to the Moon of Earth, moons of circumbinary planets could lead to large obliquity variations.  

In our secular theory, we neglected mean motion resonances between the stellar binary and the planetary orbit. This simplification was made because lower order mean motion resonances (within 3:1) reside inside the region of orbital instability \citep{Quarles18}, and higher order mean motion resonances are weak. We also assumed no resonance between the planet's rotation and its orbital mean motion, which could also lead to interesting dynamics \citep{Correia15}. This assumption is reasonable too, because circumbinary planets typically orbit wide stellar binaries beyond orbital period of 7 days, with planet orbital period $\ge 30$ days in the orbital stable region \citep{Armstrong14}. Thus, planet orbital periods tend to be much longer than Earth-like rotation periods, and low order resonances between planet spin and orbital mean motion are thus absent. In addition, over long (~Gyr) timescales, tidal interactions could align the planetary orbit with that of the stellar binary and alter obliquity variation patterns \citep{Correia16}. As we focus on how planet spin affects insolation and hence habitability, the behaviors of which are already important at shorter timescales, our study neglected tidal interactions.   

Our assumption on the interested planet being Earth-like (with Earth's mass, inertia tensor and rotation rate) is not required by either our theory or numerical simulation, but only a choice for concretizing the investigation. Recently, observational techniques of planetary obliquity and oblateness have been proposed \citep{Gaidos04,Carter10,Schwartz16}, and the first constraint on obliquities of planets outside of the Solar System has been made \citep{Bryan20}. In addition, the spin-rate of planets have been measured and it reveals to us important clues on planetary formation \citep{Bryan18}. Such progress in observation will help us better understand planetary spin properties to constrain obliquity variations. However, so far, these observational techniques could only be applied to massive super Jupiter-sized planets, and the spin measurement of terrestrial planets is still beyond the limit of observational sensitivity. For planets whose obliquity observation is available, our work also opens up a possibility of probing otherwise-hard-to-observe planet properties from its obliquity dynamics, as one can apply our tools in an inverse problem setup. 

Moreover, this study is relevant to understanding the habitability of circumbinary planets, because obliquity plays an important role in determining the climate of a planet. Detailed analysis has been conducted for particular exoplanets residing in the habitable zone and with one host star \citep[e.g.,][]{Shan18,Quarles19,Saillenfest19}. However, this is the first time that the more complicated dynamics of obliquities of circumbinary planets, as well as its implications, are investigated. The low obliquity variations of near-coplanar circumbinary planets could indicate a more favorable condition for the emergence of life comparing to their single star analogues.  



\section*{Acknowledgements}
We are grateful for the partial supports by NSF grant DMS-1847802 (RC and MT) and NASA grants
80NSSC20K0641 (GL) and 80NSSC20K0522 (GL). 
We thank David Charbonneau, Alexandre Correia, Matthew Holman, Smadar Naoz, Rafael de la Llave, Avi Loeb, and Billy Quarles for valuable comments.

\section*{Data Availability}
The data underlying this article are available in \texttt{GitHub}, and can be accessed with link: \url{https://github.com/GRIT-RBSim/GRIT}.

\bibliographystyle{mnras}
\bibliography{reference}

\appendix
\onecolumn
\section{Notations}\label{SI:notations}
\begin{tabular}{c c}
    $M_{*_i}$ & The mass of the $i$-th Star ($i=1,2$) \\
    $\delta$ & $\frac{M_{*_2}}{M_{*_1}+M_{*_2}}$ \\
    $m_p$ & The mass of the planet \\
    $\mu_p$ & $\frac{m_p (M_{*_1}+M_{*_2})}{M_{*_1}+M_{*_2}+m_p}$ \\
    $\mu_*$ & $\frac{M_{*_1} M_{*_2}}{M_{*_1}+M_{*_2}}$ \\
    $\bm{I}_p = \begin{bmatrix}
        I_p^{(1)} & 0 & 0 \\
        0 & I_p^{(2)} & 0 \\
        0 & 0 & I_p^{(3)} \\
    \end{bmatrix}$ & The (standard) moment of inertia tensor of the planet \\
    $\mathcal{R}$ & The radius of the equator of the planet \\
    $\mathcal{D}_* = \left\{ l_*, g_*, h_*, L_*, G_*, H_* \right\} $ & The Delaunay variables of the inner orbit \\
    $\mathcal{D}=\left\{ l_d, g_d, h_d, L_d, G_d, H_d \right\} $     & The Delaunay variables of the outer orbit \\
    $\mathcal{A}=\left\{ g_a, h_a, l_a, G_a, H_a, L_a \right\}$      & The Andoyer variables of the planet       \\
\end{tabular}

\section{Perturbative Analysis for Combining Lagrange-Laplace Theory with Our Circumbinary Secular Theory}\label{SI:proof}
\begin{theorem} Given $A=\begin{bmatrix} \lambda_1 & 0 & 0 \\ 0 & \ddots & 0 \\ 0 & 0 & \lambda_n \end{bmatrix}$
    and $B=\begin{bmatrix} b_{11} & \cdots & b_{1n} \\ \vdots & \ddots & \vdots \\ b_{n1} & \cdots & b_{nn} \end{bmatrix}$,
    the eigenvalues of $A+\epsilon B$ are $\lambda_i+\epsilon b_{ii}+\mathcal{O}(\epsilon^2)$ for $i=1,\ldots,n$.
\end{theorem}

\begin{proof}
  Let $v_i$ be the $i$-th unit vector. Obviously $A v_i = \lambda_i v_i$ and $v_j^T A = \lambda_j v_j^T$. Matching $\mathcal{O}(\epsilon)$ terms in
  \[
    (A+\epsilon B)(v_i+\epsilon \delta v_i) = (\lambda_i + \epsilon \delta \lambda_i)(v_i+\epsilon \delta v_i)
  \]
  gives
  \[
    B v_i+A \delta v_i=\delta\lambda_i v_i+\lambda_i \delta v_i.
  \]
  Multiplying $v_i^T$ from the left gives
  \[
    b_{ii} + v_i^T \lambda_i \delta v_i = \delta\lambda_i + \lambda_i v_i^T \delta v_i.
  \]
  Therefore, $\delta\lambda_i = b_{ii}$.
  P.S. One can also find about entries $\delta v_i$ by left multiplication by $v_j^T$.
  Multiplying $v_j^T$ from the left gives
  \[
    b_{ji} + v_j^T \lambda_j \delta v_i = \lambda_i v_j^T \delta v_i \Longrightarrow
    v_j^T \delta v_i = \frac{b_{ji}}{\lambda_j-\lambda_i} \quad (\text{we need } \frac{1}{\lambda_j-\lambda_i}=O(1)).
  \]
\end{proof}

\section{Canonical Variables for Spin-Orbit Dynamics}\label{SI:variables}

Building on the canonical Delaunay variables and the Andoyer variables~\citep{serret1866memoire,tisserand1891traite,andoyer1923cours,gurfil2007},
our secular theory of one planet rotating around binary stars not only reduces the dimensions of a system by averaging out the separated fast variables,
but also provides insights of the physical behaviors of the system as both sets of variables have clear physical meanings.
Specifically, the orbital dynamics are characterized by the Delaunay variables, whereas the spin dynamics are characterized by Andoyer variables.
In the following, we will introduce the procedure for properly constructing these variables in a circumbinary system, where two stars are unaffected by the planet and thus modeled as point masses, and the planet is modeled as a moving and rotating rigid body.
Firstly, three frames (the \textit{reference frame}, the \textit{body frame} and the \textit{angular momentum frame}) involved in our discussion will be introduced.
Secondly, the Delaunay variables will be introduced as a canonical change of coordinates from the spatial positions and momenta of all three bodies' centers of mass.
In the end, the orientation and the spin of the rigid planet will be characterized using Andoyer variables.

\subsection{Three Frames}

The \textit{reference frame} is a fixed frame in $\R^3$ with orthogonal basis $\left( \bm{E}_1, \bm{E}_2, \bm{E}_3 \right)$.
Under the \text{reference frame}, the spatial position and the translational speed of a body can be expressed by vectors in $\R^3$.
In our setup, as the inner orbit's oscillation is relative small compared with the planet's orbit, we make the assumption that the inner orbit is near Keplerian.
Thus, we may choose the reference frame such that
$\bm E_1$-axis matches the semi-major axis of the initial inner orbit;
$\bm E_2$-axis matches the semi-minor axis of the initial inner orbit;
$\bm E_3$-axis matches the normal vector of the initial inner orbit.

On the other hand, the \textit{body frame} (\cref{fig:body_frame}) is the moving frame attached to the rotating body (i.e., the planet), giving each particle of the body fixed coordinates.
As can be seen from \cref{fig:body_frame}, orthogonal bases $\left( \bm e_1, \bm e_2, \bm e_3 \right)$ formed a body frame of this rigid body.
As the body moving along the dashed trajectory following arrows as well as self rotating from time $t_0$ to time $t_1$,
coordinates of points $P, Q$ under the body frame stay the same, despite of the motion of the rigid body.
Since we have assumed that each planet is a spheroid, for the body frame (see \cref{fig:body_frame}),
we fix the body plane (spanned by $\bm e_1$ and $\bm e_2$) the plane of equator and fix the origin the center of the rigid body.

\begin{figure}
\captionsetup{type=figure}
\includegraphics[width=0.5\linewidth]{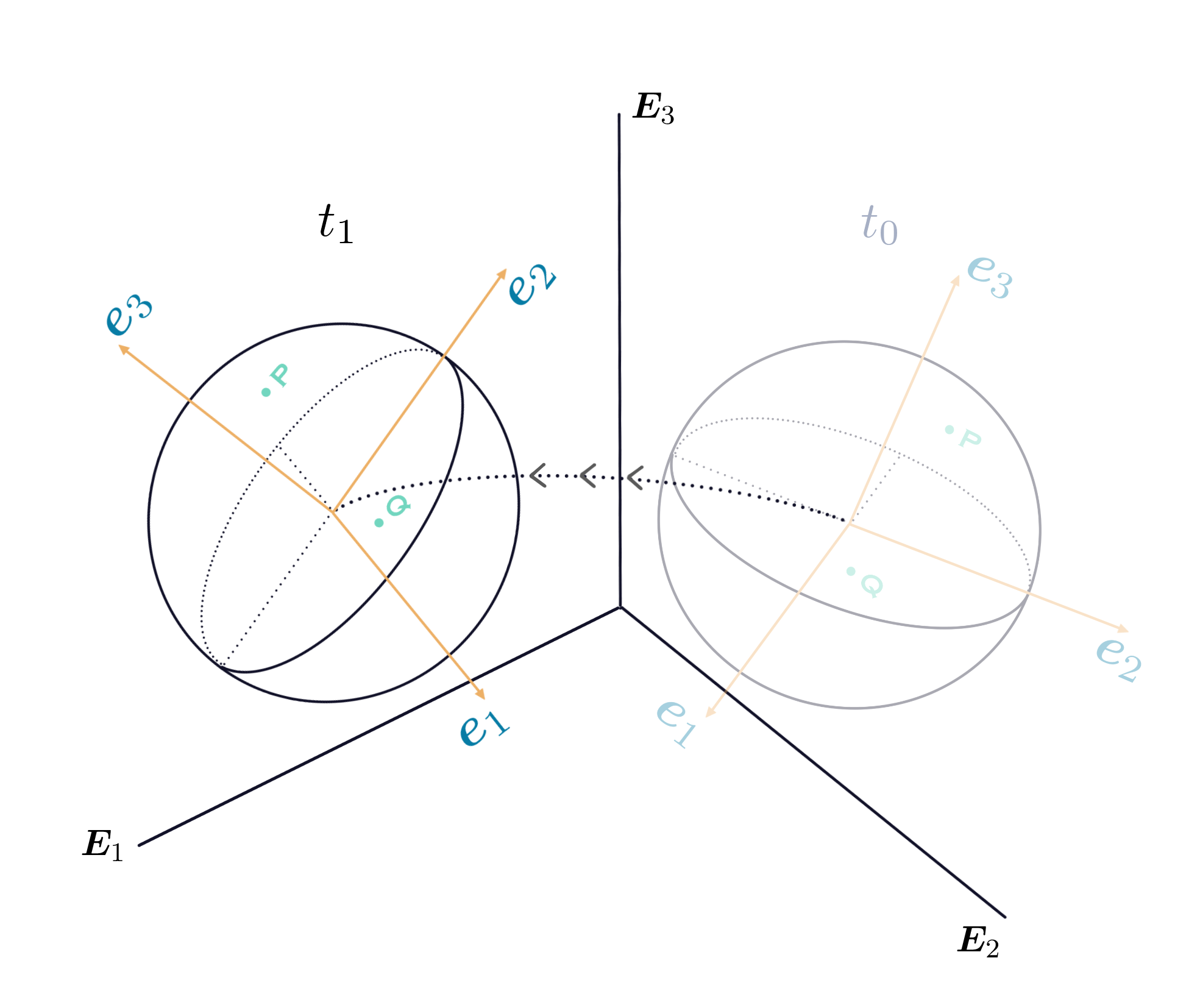}
\centering
\caption{The body frame}
\label{fig:body_frame}
\end{figure}

In addition to the \textit{body frame} and the fixed \textit{reference frame } (inertial frame),
The \textit{angular momentum frame} is any frame that the rotational angular momentum of the rigid body matches the $z$-direction.

\begin{table}\centering
\begin{tabular}{m{5cm}m{2cm}m{2cm}m{2cm}}
\toprule
    & $x$-axis  & $y$-axis & $z$-axis \\
\midrule
\textit{Reference Frame} & $\bm{E}_1$           & $\bm{E}_2$           & $\bm{E}_3$           \\ \midrule
\textit{Body Frame} & $\bm{e}_1$           & $\bm{e}_2$           & $\bm{e}_3$           \\ \midrule
\textit{Angular Momentum Frame} & $-$           & $-$           & $\overrightarrow{\bm G_a}$ \\
\bottomrule
\end{tabular}
\captionsetup{justification=centering}
\caption{$x,y,z$-axis of three frames}\label{table:frames}
\end{table}

In \cref{fig:andoyer}, three frames of a spinning rigid body are shown together with symbols explained in \cref{table:frames}.
For these three frames, the \textit{reference frame} is fixed; the \textit{body frame} is moving with respect to the orientation of the body; while the \textit{angular momentum frame} is moving with respect to the spinning direction of the body.

\begin{figure}
\centering
\includegraphics[width=0.5\linewidth]{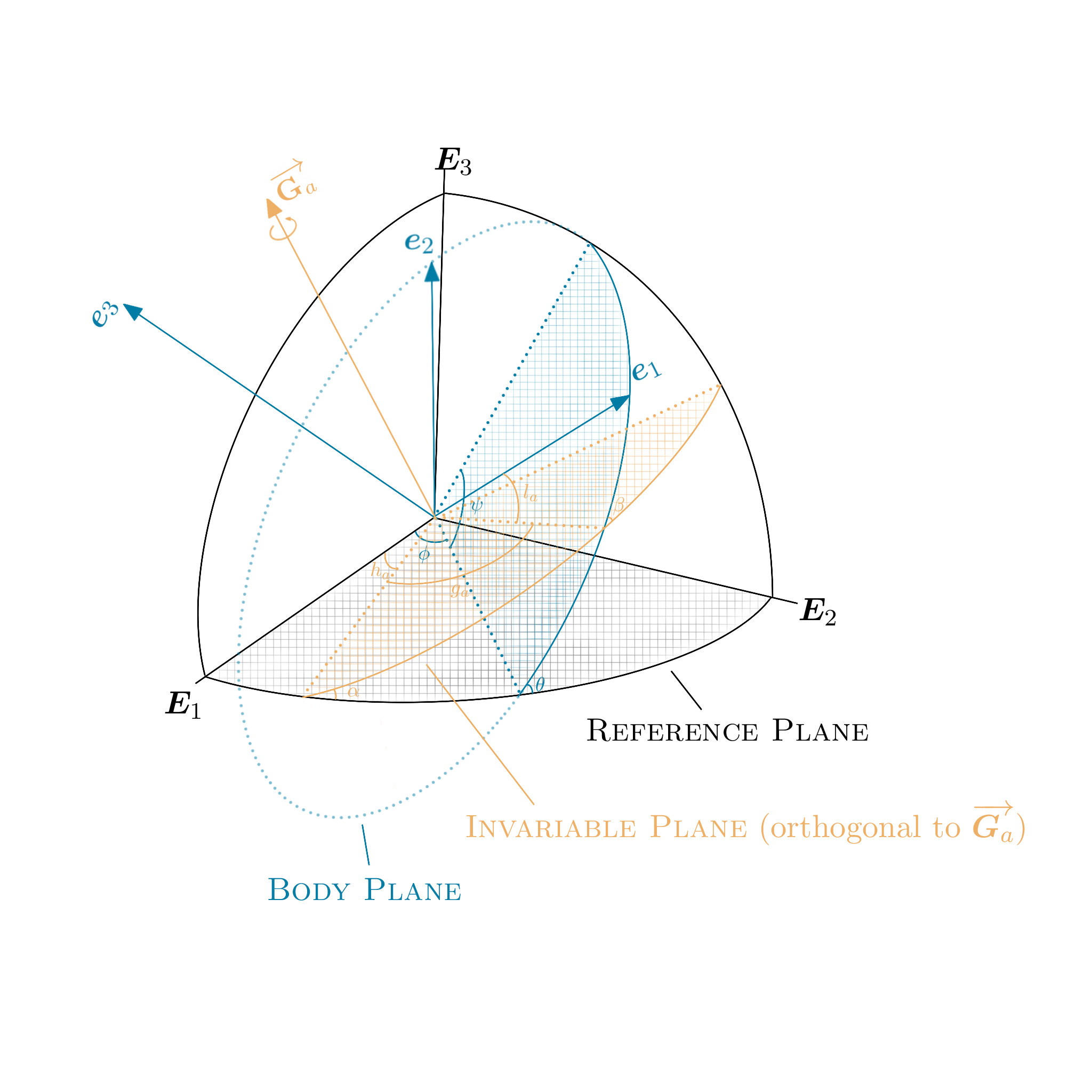}
\caption{Three frames with Andoyer variables}\label{fig:andoyer}
\end{figure}

Explanations:
$\overrightarrow{\bm G_a}$ is the angular momentum vector of the rigid body;
$\alpha$ is the angle between the $z$-axis of the \textit{angular momentum frame} and the $z$-axis of the \textit{reference frame};
$\beta$ is the angle between the $z$-axis of the \textit{angular momentum frame} and the $z$-axis of \textit{body frame}.

\subsection{Converting Cartesian Spatial Positions and Momenta to the Delaunay Variables}

In the system of one planet rotating around binary stars, there are two orbits:
the inner orbit of two stars and the outer orbit of the planet rotating around the center of mass of the binary stars.
As the inner orbit's oscillation is relatively small compared with the planet's orbit, we make the assumption that the inner orbit is near Keplerian,
while the outer orbit has its osculating orbital elements being oscillatory (except for the anomaly variable).
We denote the orbital elements~\citep{morbidelli2002modern} of the inner orbit and the outer orbit as
$\left( a_*, e_*, i_*, \omega_*, \Omega_*, \nu_* \right)$,  
$\left( a_p, e_p, i_p, \omega_p, \Omega_p, \nu_p \right)$ respectively.
Here, for the inner orbit (relative orbit of the $2$nd star around the $1$st star),
$a_*$ is the length of semi-major axis;
$e_*$ is the eccentricity;
$i_*$ is the inclination;
$\omega_*$ is the argument of periapsis;
$\Omega_*$ is the longitude of ascending node;
and we set $\nu_*$ to be the true anomaly.
For the outer orbit of the planet around the center of mass of two stars, $a_p$ is the length of semi-major axis;
$e_p$ is the eccentricity;
$i_p$ is the inclination;
$\omega_p$ is the argument of periapsis;
$\Omega_p$ is the longitude of ascending node;
and $\nu_p$ the true anomaly of the planet.
Expressing the positions of the three bodies using orbital elements, we have
\begin{align}
\left\{
\begin{aligned}
    \bm q_{*_1} & = \delta \cdot \frac{a_* \left( 1-e_*^2 \right)}{1+e_* \cos \left( \nu_* \right)} \cdot R_x\left( \Omega_* \right) \cdot R_z(i_*) \cdot R_x(\omega_*) \cdot
    \begin{bmatrix}
        \cos \nu_* \\
        \sin \nu_* \\
        0 \\
    \end{bmatrix}, \\
    \bm q_{*_2} & = - \frac{M_{*_1}}{M_{*_2}} \bm q_{*_1}, \\
    \bm q_p & = \frac{a_p \left( 1-e_p^2 \right)}{1+e_p \cos \nu_p} \cdot R_x\left( \Omega_p \right) \cdot R_z(i_p) \cdot R_x(\omega_p) \cdot
    \begin{bmatrix}
        \cos \nu_p \\
        \sin \nu_p \\
        0 \\
    \end{bmatrix},\\
\end{aligned}
\right.
    \label{eq:positions}
\end{align}
with $\bm q_{*_i}$, $i=1,2$ the positions of two stars, $\bm q_p$ the position of the planet and $R_x(\cdot)$, $R_y(\cdot)$, $R_z(\cdot)$ are defined in \cref{eq:Rx,eq:Ry,eq:Rz}.

\begin{align}
    R_x(\varphi) := \begin{bmatrix}
        1 & 0             & 0              \\
        0 & \cos{\varphi} & -\sin{\varphi} \\
        0 & \sin{\varphi} & \cos{\varphi}  \\
    \end{bmatrix},
    \label{eq:Rx}
\end{align}
\begin{align}
    R_y(\varphi) := \begin{bmatrix}
        \cos{\varphi}  & 0 & \sin{\varphi} \\
        0              & 1 & 0             \\
        -\sin{\varphi} & 0 & \cos{\varphi} \\
    \end{bmatrix},
    \label{eq:Ry}
\end{align}
\begin{align}
    R_z(\varphi) := \begin{bmatrix}
        \cos{\varphi} & -\sin{\varphi} & 0 \\
        \sin{\varphi} & \cos{\varphi}  & 0 \\
        0             & 0              & 1 \\
    \end{bmatrix}.
    \label{eq:Rz}
\end{align}

Denote the corresponding Delaunay variables of the inner orbit as
\begin{align}\label{eq:delaunay_variables_appendix_star}
    \mathcal{D}_* = & \left\{ l_*, g_*, h_*, L_*, G_*, H_* \right\}
= \bigg\{ M_*, \omega_*, \Omega_*, \mu_* \sqrt{\Gc (M_{*_1}+M_{*_2}) a_*}, L_* \sqrt{1-e_*^2}, G_* \cos i_* \bigg\},
\end{align}
with $\mu_*=\frac{M_{*_1} M_{*_2}}{M_{*_1} + M_{*_2}}$ and $M_*$ the mean anomaly of the inner orbit.
Similarly, the Delaunay variables of the outer orbit is denoted as
\begin{align}\label{eq:delaunay_variables_appendix_planet}
    \mathcal{D} = & \left\{ l_d, g_d, h_d, L_d, G_d, H_d \right\}
    = \bigg\{ M_p, \omega_p, \Omega_p, \mu_p \sqrt{\Gc \left( M_{*_1} + M_{*_2} + m_p \right) a_p}, L_d \sqrt{1-e_p^2}, G_d \cos i_p \bigg\},
\end{align}
with $\mu_p=\frac{m_p (M_{*_1}+M_{*_2})}{M_{*_1} + M_{*_2} + m_p}$ and $M_p$ the mean anomaly of the outer orbit.
In \cref{eq:delaunay_variables_appendix_star,eq:delaunay_variables_appendix_planet}, $l_*,g_*,h_*,l_d,g_d,h_d$ are angle variables
with $L_*,G_*,H_*,L_d,G_d,L_d$ their conjugate momenta.
Mapping the orbital elements in \cref{eq:positions} to Delaunay variables defined above, the positions can be expressed as functions of Delaunay variables.
Noting that different anomalies are used in \cref{eq:positions} and \cref{eq:delaunay_variables_appendix_star}, \cref{eq:delaunay_variables_appendix_planet}, one need to solve the Kepler equation to perform the mapping from the orbital elements to the Delaunay variables.

\subsection{The Orientation and the Rotational Momenta of the Rigid Planet in the Andoyer Variables}
\label{SI:andoyer_variables}

As we are interested in the dynamics of the planet's spin, the planet is modeled as a rigid body to account for its finite size,
thus the orientation and the rotation of the rigid body are necessary to represent the rigid body.
As is well known, each rotation matrix $\bm{R} \in \mathsf{SO}(3)$ rotates 3-dimensional vectors in Euclidean space in a unique way (we will use the convention of column vectors and left multiplication).
Thus we may use a time dependent rotation matrix to represent
the transformation from the \textit{body frame} to the \textit{reference frame} in our dynamics (i.e.\ the orientation of the rigid body).
Angles extracted from the Andoyer variables are used to reflect the orientation by a sequence of \textit{standard rotations} (\cref{eq:Rx,eq:Ry,eq:Rz}).
Also the rotational angular momentum of the rigid body can be expressed with the momentum variables of the Andoyer variables, conjugating to the three angle variables.
Moreover, other than just representing the state of the rigid body, the Andoyer variables also separate the slow and fast scales of the dynamics.

In detail, similar to the Delaunay variables, the Andoyer variables are canonical coordinates. They consist of three angle variables $\left\{ g_a, h_a, l_a \right\}$,
and their conjugate momenta $\left\{ G_a, H_a, L_a \right\}$.
The angle variables are defined by the orientation of three frames (see \cref{fig:andoyer}),
while $G_a, H_a, L_a$ are defined as the following~\citet{kinoshita1972first, gurfil2007},
\begin{align}
    \left\{
        \begin{aligned}
            G_a & = \left| \overrightarrow{\bm G_a} \right|, \\
            H_a & = G_a \cos{\alpha} = \left| \overrightarrow{\bm G_a} \cdot \bm E_3 \right|, \\
            L_a & = G_a \cos{\beta} = \left| \overrightarrow{\bm G_a} \cdot \bm e_3 \right|, \\
        \end{aligned}
    \right.
    \label{GHL_andoyer}
\end{align}
with $\alpha$, $\beta$ described in \cref{fig:andoyer}.  In other words, $G_a$ is the magnitude of the angular momentum;
$H_a$ is the magnitude of $\overrightarrow{\bm G_a}$'s orthogonal projection to $\bm E_3$;
$L_a$ is the magnitude of $\overrightarrow{\bm G_a}$'s orthogonal projection to $\bm e_3$.

Setting the \textit{angular momentum frame} as an intermediate frame,
we may represent the rotation matrix $\bm R$ as a product of a sequence of \textit{standard rotations} using angles in \cref{fig:andoyer}~\citet{gurfil2007},
\begin{align}
    \bm R = R_z(h_a) \cdot R_x(\alpha) \cdot R_z(g_a) \cdot R_x(\beta) \cdot R_z(l_a).
    \label{eq:R_in_angles}
\end{align}

Since $\cos{\alpha} = \frac{H_a}{G_a}$ and $\cos{\beta} = \frac{L_a}{G_a}$, replacing $\alpha,\beta$ in \cref{eq:R_in_angles},
the rotation matrix can be represented purely using the Andoyer variables $\left\{ g_a, h_a, l_a, G_a, H_a, L_a \right\}$ (name it as $\bm R^\mathcal{A}$),
\begin{align}\label{eq:R_a}
    \bm R^\mathcal{A} = R_z\left( h_a \right) \cdot R_x\left( \arccos\left(\frac{H_a}{G_a}\right) \right) \cdot R_z\left( g_a \right) \cdot R_x\left( \arccos\left( \frac{L_a}{G_a} \right) \right) \cdot R_z\left( l_a \right).
\end{align}

\section{Dynamics of the Spin for an Earth-like Planet in Circumbinary Systems: two point mass stars + one rigid body planet}\label{SI:secular_theory}

We will firstly formulate the exact Hamiltonian in slow and fast variables in Section D.1. Then, in Section D.2, fast variables will be averaged in sequence under the "average assumptions" (\cref{eq:appendix_assumption}), resulting in an approximated spin dynamics \cref{eq:avg_H_appendix}.

\subsection{Hamiltonian Formulation}\label{SI:sec_hamiltonian_formulation}

In this section, we will express the Hamiltonian using the canonical Delaunay and Andoyer variables $\Dc,\Dc_*,\Ac$ introduced in Appendix C.
In short, $H=T_{*_1}+T_{*_2}+T_p + V_* + V_1 + V_2$ with $T_{*_i}$ kinetic energies of the $i$th star, $T_p$ the kinetic energy of the planet, $V_*$ potential energy between two stars and $V_i$ potential energies between $i$th star and the planet.We assume that the inner orbit is near Keplerian (as the perturbation of the planet is negligible) such that during the dynamics $T_{*_1}+T_{*_2} + V_* = -\frac{\Gc^2{(M_{*_1}+M_{*_2})}^2 \mu_*^3}{2 L_*^2} + \mathfrak{R}_* = H_0(L_*) + \mathfrak{R}_*$.
We will calculate $T_p$, $V_1$ and $V_2$ terms and conclude the Hamiltonian in action angle variables in the following.
As a result, the exact Hamiltonian \cref{eq:H_approximated_appendix} in slow / fast variables is written as a summation of a small-terms-combined $\tilde{\mathfrak{R}}$ (see \cref{eq:appendix_remainder}) and the effective Hamiltonian.
(note that in the section D.2, $\tilde{\mathfrak{R}}$ is dropped during averaging
 under the "averaging assumptions" (see \cref{eq:appendix_assumption}), but we still keep track of the order of small parameters in $\tilde{\mathfrak{R}}$ in this section to better quantify the effects of small $e_*$, $e_p$, etc.)

\subsubsection{The Kinetic Energy of the Planet}\label{subsubsec:T}

The total kinetic energy of the planet is the sum of its linear kinetic energy $T^{linear}\left( \Dc \right)$
and its rotational kinetic energy $T^{rot}\left(\Ac,L_*,l_*\right)$.
The former is $T^{linear} = \frac{\Gc^2 {(M_{*_1}+M_{*_2}+m_p)}^2 \mu_p^3}{2 L_d^2}$.
For the latter, by the definition of the Andoyer variables, we may express the rotational kinetic energy as
\begin{align}
    \begin{split}
            T^{rot} (\mathcal A)
        = \frac{G_a^2 \sin^2 \beta}{2 I_p^{(1)}} + \frac{G_a^2 \cos^2 \beta}{2 I_p^{(3)}}
        = \frac{H_a^2}{2 I_p^{(1)}} + \frac{G_a^2 - H_a^2}{2 I_p^{(3)}}.
    \end{split}
    \label{eq:T_in_Andoyer_variables_appendix}
\end{align}
with
\begin{align}
\bm I_p = \begin{bmatrix}
    I_p^{(1)} & 0 & 0 \\
    0 & I_p^{(2)} & 0 \\
    0 & 0 & I_p^{(3)} \\
\end{bmatrix}
\end{align}
the (standard) moment of inertia tensor of the planet.

\subsubsection{The Potential Energy of the Planet}\label{subsubsec:V}

The potential energy of the planet is influenced by both stars,
\begin{align}
    V(\mathcal A, \mathcal D, l_*, L_*) = V_1 + V_2.
    \label{V_in_Andoyer_variables}
\end{align}

Without loss of generality, we will only derive $V_1$ here.
Integrating the potential for each mass point of the rigid planet,
\begin{align}
    \begin{split}
        V_1 = & \int_{\mathcal{B}} - \frac{\mathcal{G} M_{*_1} \rho(\bm{x})} {\norm{\left( \bm q_p + \bm R_p\bm{x} \right) - \bm q_{*_1}}} \, d\bm{x} \\
        = & \frac{\Gc M_{*_1} m_p}{\norm{\bm q_p-\bm q_{*_1}}}
        \bigg\{
            - 1
            - \frac{tr[\bm I_p]}{2 m_p \norm{\bm q_p-\bm q_{*_1}}^2}
            + \frac{3 {\left( \bm q_p-\bm q_{*_1} \right)}^T {\bm R_p} \bm I_p {\bm R_p}^T \left( \bm{q}_p-\bm{q}_{*_1} \right)}
            {2 m_p \norm{\bm{q}_p-\bm{q}_{*_1}}^4} 
            + \mathcal{O} \left( {\left(\frac{\mathcal{R}}{\norm{\bm{q}_p-\bm{q}_{*_1}}}\right)}^4 \right)
        \bigg\} \,,
    \end{split}
    \label{eq:true_V1}
\end{align}
with $\bm q_p, \bm q_{*_i}$ in the Delaunay variables $\Dc, \Dc_*$, $\bm R_p$ the rotation matrix in the Andoyer variables $\Ac$ (\cref{eq:R_a}) and $\bm I_p$ the moment of the inertia tensor of a symmetric planet ($I_p^{(1)} = I_p^{(2)} \neq I_p^{(3)}$).
Term by term, we will show how we express \cref{eq:true_V1} in the Delaunay variables and the Andoyer variables with higher order terms separated (the following derivation are simply some high order expansions, feel free to jump to the conclusion of this section).

In detail, to analyze the term in \cref{eq:true_V1} with $\bm R_p$, we assume the angle $\beta$ (the angle between the spin axis and the $\bm e_3$-axis of the planet) small as physically, the oblate shape of the planet is flattened by the spinning of the planet.
Thus we may rewrite $\bm R_p \bm I_p \bm R_p^T$ in \cref{eq:true_V1} by separating the higher order terms of $\Oc\left( \beta \right)$,
\begin{align}
    \begin{split}
          \bm R_p \bm I_p \bm R_p^T
        = &  R_z(h_a) \cdot R_x(\alpha) \cdot R_z(g_a) \cdot R_x(\beta) \cdot R_z(l_a) \cdot \bm I_p
            \cdot R_z^T(l_a) \cdot R_x^T(\beta) \cdot R_z^T(g_a) \cdot R_x^T(\alpha) \cdot R_z^T(h_a) \\
        = &  R_z(h_a) \cdot R_x(\alpha) \cdot R_z(g_a) \cdot R_x(\beta) \cdot \bm I_p
            \cdot R_x^T(\beta) \cdot R_z^T(g_a) \cdot R_x^T(\alpha) \cdot R_z^T(h_a) \\
        = &  R_z(h_a) \cdot R_x(\alpha) \cdot \bm I_p \cdot R_x^T(\alpha) \cdot R_z^T(h_a) + \mathcal{O}(\beta) \\
        = &  I_p^{(1)} \Bigg \{ \bm I_{3 \times 3} + \frac{I_p^{(3)}-I_p^{(1)}}{I_p^{(1)}} \cdot R_z(h_a) \cdot R_x(\alpha)
          \cdot \begin{bmatrix}
            0 & 0 & 0 \\ 0 & 0 & 0 \\ 0 & 0 & 1 \\
        \end{bmatrix} \cdot R_x^T(\alpha) \cdot R_z^T(h_a) + \mathcal{O}(\beta) \Bigg\}. \\
    \end{split}
    \label{eq:RIR}
\end{align}

For $\norm{\bm q_p - \bm q_{*_1}}$, using \cref{eq:positions}, we have
\begin{align}
        \norm{\bm{q}_{*_1}-\bm{q}_p}^2 = \norm{\bm q_p}^2
        \cdot \left( 1 + 2 \cdot \mathfrak{P} (\nu_*,\nu_p,i_p,\omega_p,\Omega_p) \cdot \frac{\norm{\bm q_{*_1}}}{\norm{\bm q_p}}
        + {\left( \frac{\norm{\bm q_{*_1}}}{\norm{\bm q_p}} \right)}^2 \right),
    \label{eq:dis_star_planet}
\end{align}
with
\begin{align}
    \begin{split}
        & \mathfrak{P} (\nu_*,\nu_p,i_p,\omega_p,\Omega_p) \\
    =   & \sin \left( \nu_p \right) \cos \left( \nu_* \right) \big[\cos \left( i_p \right) \cos \left( \omega_p \right) \sin \left( \Omega_p \right)
        +\sin \left( \omega_p \right) \cos \left( \Omega_p \right)\big]
        -\sin \left( \nu_* \right) \big[\cos \left(i_p\right) \cos \left( \Omega_p \right) \\
        & \cdot \sin \left( \nu_p + \omega_p \right) + \sin \left( \Omega_p \right) \cos \left( \nu_p + \omega_p \right)\big]
        +\cos \left( \nu_p \right) \cos \left( \nu_* \right) \big[\cos \left( i_p \right) \sin \left( \omega_p \right) \sin \left( \Omega_p \right)
        - \cos \left( \omega_p \right) \cos \left( \Omega_p \right)\big], \\
    \end{split}
    \label{eq:mathfrakP}
\end{align}
and
\begin{align}
    \frac{\norm{\bm q_{*_1}}}{{\norm{\bm q_p}}} = \delta \frac{a_*(1 + \mathcal O(e_*))}{a_p(1 + \mathcal O(e_p))} 
    = \delta \frac{a_*}{a_p} + \mathcal{O}(e_*) + \mathcal{O}(e_p).
    \label{eq:norm_ratio}
\end{align}

Applying Legendre polynomial for the expansion of $\frac{1}{\norm{\bm q_{*_1} - \bm q_p}}$ with $\norm{\bm q_{*_1} - \bm q_p}$ calculated in \cref{eq:dis_star_planet}
, we have
\begin{align}
    \begin{split}
          \frac{1}{\norm{\bm{q}_{*_1}-\bm{q}_p}}
        = & \frac{1}{\norm{\bm q_{p}}} \Bigg( 1
            - \frac{\norm{\bm q_{*_1}}}{\norm{\bm q_p}} \mathfrak{P} 
            + {\left( \frac{\norm{\bm q_{*_1}}}{\norm{\bm q_p}} \right)}^2 \left( \frac{3}{2} \mathfrak{P}^2 - \frac{1}{2} \right)
          + \mathcal O \left( {\left( \frac{\norm{\bm q_{*_1}}}{\norm{\bm q_p}} \right)}^3 \right)
          \Bigg) \\
          = & \frac{1}{a_p} \Bigg( 1 
              - \delta {\left( \frac{a_*}{a_p} \right)} \mathfrak{P} 
            + \delta^2 {\left( \frac{a_*}{a_p} \right)}^2 \left( \frac{3}{2} \mathfrak{P}^2 - \frac{1}{2} \right)
            + \mathcal O \left( {\left( \frac{a_*}{a_p} \right)}^3 \right)
      \Bigg) + \mathcal{O}(e_*) + \mathcal{O}(e_p). \\
    \end{split}
    \label{eq:legendre_poly}
\end{align}


Let
\begin{align}
    \mathfrak{Q}_{*_1} = \delta \cdot a_* \cdot 
    \begin{bmatrix}
        \cos \nu_* \\
        \sin \nu_* \\
        0 \\
    \end{bmatrix}
\text{ and }
    \mathfrak{Q}_p = a_p \cdot R_x\left( \Omega_p \right) \cdot R_z(i_p) \cdot R_x(\omega_p) \cdot
    \begin{bmatrix}
        \cos \nu_p \\
        \sin \nu_p \\
        0 \\
    \end{bmatrix}.
    \label{eq:mathfrakQ}
\end{align}
According to \cref{eq:positions}, we have
\begin{align}
    \bm q_{*_1} = \mathfrak{Q}_{*_1} + \Oc (e_*),
    \qquad
    \bm q_p = \mathfrak{Q}_p + \mathcal{O}(e_p).
    \label{eq:approximatedq}
\end{align}

Plugging \cref{eq:approximatedq,eq:legendre_poly,eq:RIR},
into \cref{eq:true_V1}
we have
\begin{align}
    \begin{split}
        V_1 = & \frac{\Gc M_{*_1} m_p}{a_p} \bigg\{ - 1 - \frac{\frac{tr[\bm{I}_p]}{2} - \frac{3}{2} I_p^{(1)}}{m_p a_p^2}
            + \frac{I_p^{(3)}-I_p^{(1)}}{m_p a_p^4} \cdot {\left( \mathfrak{Q}_p-\mathfrak{Q}_{*_1} \right)}^T \cdot R_z(h_a) \cdot R_x(\alpha)
         \cdot \begin{bmatrix}
            0 & 0 & 0 \\ 0 & 0 & 0 \\ 0 & 0 & 1 \\
        \end{bmatrix}
          \cdot R_x^T(\alpha) \cdot R_z^T(h_a) \left( \mathfrak{Q}_p-\mathfrak{Q}_{*_1} \right) \\
           &  + \frac{a_*}{a_p} \mathfrak{D}
            + \Oc \left( {\left(\frac{\mathcal{R}}{a_p}\right)}^2 \beta \right)
            + \Oc \left( {\left(\frac{\mathcal{R}}{a_p}\right)}^4 \right)
           + \Oc (e_*) + \Oc (e_p) + \Oc \left( {\left( \frac{a_*}{a_p} \right)}^3 \right) \bigg\} \\
    = & - \frac{\Gc M_{*_1} m_p}{a_p} + \frac{a_*}{a_p} \mathfrak{R}_1
    + \frac{\Gc M_{*_1} m_p}{a_p} \left\{
            \Oc \left( {\left(\frac{\mathcal{R}}{a_p}\right)}^2 \beta \right)
            + \Oc \left( {\left(\frac{\mathcal{R}}{a_p}\right)}^4 \right)
           + \Oc (e_*) + \Oc (e_p) + \Oc \left( {\left( \frac{a_*}{a_p} \right)}^3 \right) \right\} \\
    \end{split}
    \label{eq:V1_final}
\end{align}
where $\mathfrak{D}$ is an $\Oc(1)$ polynomial and
\begin{align}
    \begin{split}
        \mathfrak{R}_1 = & \frac{\Gc M_{*_1} m_p}{a_p} \bigg\{ - \mathfrak{D} - \frac{\frac{tr[\bm{I}_p]}{2} - \frac{3}{2} I_p^{(1)}}{m_p a_p} + \frac{I_p^{(3)}-I_p^{(1)}}{m_p a_p^4}
                \cdot {\left( \mathfrak{Q}_p-\mathfrak{Q}_{*_1} \right)}^T \cdot R_z(h_a) \cdot R_x(\alpha)
             \cdot \begin{bmatrix}
                0 & 0 & 0 \\ 0 & 0 & 0 \\ 0 & 0 & 1 \\
            \end{bmatrix} \cdot R_z^T(h_a) \cdot
            \left( \mathfrak{Q}_p-\mathfrak{Q}_{*_1} \right)
        \bigg\}. \\
        \label{eq:mathfrakR1}
    \end{split}
\end{align}
Similarly, $V_2 = -\frac{\Gc M_{*_2} m_p}{a_p} + \frac{a_*}{a_p} \mathfrak{R}_2 + \frac{\Gc M_{*_1} m_p}{a_p} \left\{
            \Oc \left( {\left(\frac{\mathcal{R}}{a_p}\right)}^2 \beta \right)
            + \Oc \left( {\left(\frac{\mathcal{R}}{a_p}\right)}^4 \right)
           + \Oc (e_*) + \Oc (e_p) + \Oc \left( {\left( \frac{a_*}{a_p} \right)}^3 \right) \right\}$.

Changing coordinates from $\left( a_p, e_p, i_p, \omega_p, \Omega_p, \nu_p \right)$ to $\mathcal D$, we have $V=V_1+V_2$ as a function of $\mathcal A, \mathcal D$, $l_*$ and $L_*$.
Note that $\nu_p$ is the true anomaly and $M_p$ (in $\Dc$) is the mean anomaly, and we expressed $M_p$ as $\nu_p + \Oc(e_p)$ instead of solving the Kepler equation exactly.

\subsubsection{The Hamiltonian}\label{subsec:Hamiltonian}

Summing the kinetic energy 
and the potential energy 
, the full Hamiltonian is:
\begin{align}\label{eq:H_org_appendix}
\begin{split}
    H \left( \mathcal A, \mathcal D, l_*, L_*  \right) & = T^{linear}\left( \mathcal D \right) + T^{rot}\left( \mathcal A \right)
    + V\left( \mathcal A, \mathcal D, l_*, L_* \right) \\
    & = -\frac{\Gc^2 {(M_{*_1}+M_{*_2})}^2 \mu_*^3}{2 L_*^2} + \frac{\Gc^2 {(M_{*_1}+M_{*_2}+m_p)}^2 \mu_p^3}{2 L_d^2} + T^{rot}\left( \mathcal A \right)
    + V\left( \mathcal A, \mathcal D, l_*, L_* \right). \\
\end{split}
\end{align}

Plugging in the potential $V_1$ and $V_2$ to \cref{eq:H_org_appendix}, one sees that the Hamiltonian only depends on

$\Sc = \left\{ l_d, h_d, g_d, L_d, H_d, G_d, l_a, g_a, h_a, G_a, H_a, L_a, l_*, L_*\right\}$,
\begin{align}
    \begin{split}
            H\left( \Sc \right)
        = & -\frac{\Gc^2 {(M_{*_1}+M_{*_2})}^2 \mu_*^3}{2 L_*^2} + \frac{\Gc^2 {(M_{*_1}+M_{*_2}+m_p)}^2 \mu_p^3}{2 L_d^2} + T^{rot} (\Sc) + V_1 + V_2, \\
        = & -\frac{\Gc^2 {(M_{*_1}+M_{*_2})}^2 \mu_*^3}{2 L_*^2} + \frac{\Gc^2 {(M_{*_1}+M_{*_2}+m_p)}^2 \mu_p^3}{2 L_d^2} + \frac{H_a^2}{2 I_p^{(1)}} + \frac{G_a^2 - H_a^2}{2 I_p^{(3)}}
           - \frac{\Gc^2 {(M_{*_1}+M_{*_2}+m_p)}^2 \mu_p^3}{L_d^2} + \frac{a_*(L_*)}{a_p(L_d)} \left( \mathfrak{R}_1 + \mathfrak{R}_2 \right) + \widetilde{\mathfrak{R}}, \\
        = & -\frac{\Gc^2 {(M_{*_1}+M_{*_2})}^2 \mu_*^3}{2 L_*^2} - \frac{\Gc^2 {(M_{*_1}+M_{*_2}+m_p)}^2 \mu_p^3}{2 L_d^2} + \frac{H_a^2}{2 I_p^{(1)}} + \frac{G_a^2 - H_a^2}{2 I_p^{(3)}}
            + \frac{a_*(L_*)}{a_p(L_d)} \left( \mathfrak{R}_1 + \mathfrak{R}_2 \right) + \widetilde{\mathfrak R}, \\
        = & -\frac{\Gc^2 {(M_{*_1}+M_{*_2})}^2 \mu_*^3}{2 L_*^2} - \frac{\Gc^2 {(M_{*_1}+M_{*_2}+m_p)}^2 \mu_p^3}{2 L_d^2} + \frac{H_a^2}{2 I_p^{(1)}} + \frac{G_a^2 - H_a^2}{2 I_p^{(3)}}
            + \frac{a_*(L_*)}{a_p(L_d)} \left( \mathfrak{R}_1 + \mathfrak{R}_2 \right) + \widetilde{\mathfrak R}, \\
        = & H_0(L_*) - \frac{\Gc^2 {(M_{*_1}+M_{*_2}+m_p)}^2 \mu_p^3}{2 L_d^2} + \frac{H_a^2}{2 I_p^{(1)}} + \frac{G_a^2 - H_a^2}{2 I_p^{(3)}}
            + \frac{a_*(L_*)}{a_p(L_d)} \left( \mathfrak{R}_1 + \mathfrak{R}_2 \right) + \widetilde{\mathfrak R}, \\
    \end{split}
    \label{eq:H_approximated_appendix}
\end{align}
with 
\begin{align}
\widetilde{\mathfrak{R}}
= H_0(L_*) \frac{a_*}{a_p} \left\{
            \Oc \left( {\left(\frac{\mathcal{R}}{a_p}\right)}^2 \beta \right)
            + \Oc \left( {\left(\frac{\mathcal{R}}{a_p}\right)}^4 \right)
           + \Oc (e_*) + \Oc (e_p) + \Oc \left( {\left( \frac{a_*}{a_p} \right)}^3 \right) \right\} + \mathfrak{R}_*,
    \label{eq:appendix_remainder}
\end{align}

\subsection{Using Averaging Theory to Construct an Approximated Dynamics}\label{SI:averaging}

The dynamics of the Hamiltonian system above (\cref{eq:H_approximated_appendix}) consists of the slow components $\Sc_{slow}=\{ h_d, g_d, L_d, H_d, G_d, h_a, G_a, H_a, L_* \}$, as well as the fast components $\Sc_{fast}=\{ l_*, l_d, g_a\}$, which correspond to binary's orbital phase, planet's orbital phase, and planet spin.
The fast and slow components are separated by scaling $\eta_1$ and $\eta_2$,
where $\eta_1 = \frac{a_*}{a_p}$ is a small parameter corresponding to the stellar binary are closer to each other than the planet, and $\eta_2 = \frac{I_p^{(1)}-I_p^{(3)}} {I_p^{(3)}}$ is another small parameter modeling the oblateness of the planet. In the observed circumbinary systems, $0.084 \leq \eta_1 \leq 0.23$~\citep{li2016uncovering};
and $\eta_2 \approx 0.00334$ for Earth-like planets.

Let us now compare the relative timescales of the variables, in order to determine and justify the order of a sequence of averaging approximations. Specifically, the slow variables change with rate either $\Oc(\eta_1)$ or $\Oc(\eta_2)$. 
On the other hand, $\dot{l}_*=\Oc\left( \eta_1^{-\nicefrac{3}{2}} \right)$, $\dot{l}_d=\Oc(1)$ and $\dot{g}_a$ equals the spin rate of the planet around the direction of the angular momentum.

Here, we assume that the frequencies of $l_*$, $l_d$ and $g_a$ are not commensurable.
Thus, we may average over $l_*$ and $l_d$ in sequence, assuming no resonance occurs between $l_*$, $l_d$, $g_a$, then we saw that $g_a$ is decoupled from the system.

\subsubsection{Average over binary's orbital phase \texorpdfstring{$l_*$}{}}
A closer look at $H\left( \Sc \right)$ reveals that
all terms dependent on $l_*$ are of order $\mathcal{O}(\eta_1^2)$.
Thus, $H\left( \Sc \right)$ can be written as 
\begin{align}
    \begin{split}
        H\left( \Sc \right) = H_0 \left( L_*  \right) + \eta_1 H_1 \left( \Sc\backslash\left\{ l_*,L_* \right\} \right) + \eta_1^2 H_2 \left( \Sc \right) + \eta_1^3 H_3 \left( \Sc \right)
          & + \Oc (\eta_1^4) H_{\text{remainder}},
    \end{split}
    \label{eq:H_sep_large_small_appendix}
\end{align}
under the assumption that $\widetilde{\mathfrak{R}}/H_0(L_*) = \Oc(\eta_1^4)$. Apply the assumption for each term in $\widetilde{\mathfrak{R}}$, we rewrite the assumption as the "averaging assumptions" below
\begin{equation}
    \label{eq:appendix_assumption}
            {\left(\frac{\mathcal{R}}{a_p}\right)}^2 \beta,
            {\left(\frac{\mathcal{R}}{a_p}\right)}^4, e_*, e_p = \Oc(\eta_1^3)
            \text{  and  }
            \mathfrak R_* / H_0(L_*) = \Oc(\eta_1^4).
\end{equation}

As $H(\Sc)$ is nearly-integrable,
canonical averaging theory (e.g.,~\citet{morbidelli2002modern}) guarantees that averaging $H(\Sc)$ over $l_*$ is equivalent to first-order averaging of the Hamiltonian dynamics.
After $l_*$ is averaged, $L_*$ becomes a conserved quantity. Therefore, excluding $H_0(L_*)$, we obtain the averaged Hamiltonian
\begin{align}
    \begin{split}
        & \overline{H} \left( \Sc \right) = H_1 \left( \Sc\backslash\left\{ l_*,L_* \right\} \right) + \frac{1}{2 \pi} \int_{0}^{2 \pi} \frac{\eta_1 H_2 \left( \Sc \right) + \eta_1^2 H_3 \left( \Sc \right)}{\dot{l}_*} \, d l_* \\
        = & H_1 \left( \Sc\backslash\left\{ l_*,L_* \right\} \right) + \eta_1 \cdot 0 + \eta_1^2 \overline{H}_3 \left( \Sc\backslash\left\{ l_*,L_* \right\} \right) 
        = H_1 \left( \Sc\backslash\left\{ l_*,L_* \right\} \right) + \eta_1^2 \overline{H}_3 \left( \Sc\backslash\left\{ l_*,L_* \right\} \right) 
    \end{split}.
    \label{eq:H_avg_appendix}
\end{align}
which generates an approximated dynamics (see \cref{eq:avg_t}).

\begin{equation} \label{eq:avg_t}
\left\{
\begin{aligned}
    \dot{l}_d & = \frac{\Gc^2 {(M_{*_1}+M_{*_2}+m_p)}^2 \mu_p^3}{L_d^3}-\eta_1^2 \frac{(\delta -1) \delta \Gc^2 {(M_{*_1}+M_{*_2}+m_p)}^2 \mu_p^3 \left(G_d^2 \left(3 \cos \left(2 \left(g_d+l_d\right)\right)-1\right)+6 H_d^2 \sin ^2\left(g_d+l_d\right)\right)}{4 G_d^2 L_d^3}, \\
    \dot{g}_d & = -\frac{3 (\delta -1) \delta  H_d^2 \Gc^2 {(M_{*_1}+M_{*_2}+m_p)}^2 \mu_p^3 \eta_1 ^2 \sin ^2\left(g_d+l_d\right)}{2 G_d^3 L_d^2}, \\
    \dot{h}_d & = \frac{3 (\delta -1) \delta  H_d \Gc^2 {(M_{*_1}+M_{*_2}+m_p)}^2 \mu_p^3 \eta_1 ^2 \sin ^2\left(g_d+l_d\right)}{2 G_d^2 L_d^2}, \\
    \dot{L}_d & = \frac{3 (\delta -1) \delta  \Gc^2 {(M_{*_1}+M_{*_2}+m_p)}^2 \mu_p^3 \eta_1 ^2 \left(G_d^2-H_d^2\right) \sin \left(2 \left(g_d+l_d\right)\right)}{4 G_d^2 L_d^2}, \\
    \dot{G}_d & = \frac{3 (\delta -1) \delta  \Gc^2 {(M_{*_1}+M_{*_2}+m_p)}^2 \mu_p^3 \eta_1 ^2 \left(G_d^2-H_d^2\right) \sin \left(2 \left(g_d+l_d\right)\right)}{4 G_d^2 L_d^2}, \\
    \dot{H}_d & = 0, \\
    \dot{g}_a & = \frac{G_a}{I_p^{(3)}} + \eta_2 \left( f^{(1)}_0 \left( \mathcal S \right) + \eta_1^2 \cdot f^{(1)}_1 \left( \mathcal{S} \right) \right), \\
    \dot{h}_a & = \eta_2 \cdot \left( f^{(2)}_0\left( \mathcal S \right)+ \eta_1^2 \cdot f^{(2)}_1\left( \mathcal S \right) \right), \\
    \dot{G}_a & = 0, \\
    \dot{H}_a & = \eta_2 \cdot \left( f^{(3)}_0\left( \mathcal S \right) + \eta_1^2 \cdot f^{(3)}_1\left( \mathcal S \right) \right). \\
\end{aligned}
\right.
\end{equation}

\subsubsection{Average Over Planet's Orbital Phase \texorpdfstring{$l_d$}{}}

From the averaged Hamiltonian in \cref{eq:H_avg_appendix}, we may derive the corresponding dynamics \cref{eq:avg_t}.
Then, we average the vector field over the fast angle $l_d$. Important to note is, $\dot{l}_d$ actually depends on $l_d$, and therefore one cannot just average $\dot{L}_d$ (for example) uniformly over $l_d$ from 0 to $2\pi$. A proper ergodic averaging can be done either via integration against time, or, essentially equivalently, a weighted average that reflects the non-uniform ergodic measure of $l_d$ on the torus (which can be rigorously derived using normal form).

For example, for $L_d$, the averaged dynamics is
\[
    \dot{\overline{L}}_d
  = \frac{\frac{1}{2\pi} \int_0^{2\pi} \frac{\dot{L}_d}{\dot{l}_d} \, dt}
         {\frac{1}{2\pi} \int_0^{2\pi} \frac{1}{\dot{l}_d} \, dt}.
\]
Performing the same operation for $G_d, H_d, g_d, h_d, G_a, H_a, g_a, h_a$, we obtain the averaged dynamics \cref{eq:avg_ld}.

\subsubsection{Resulting Approximated Dynamics}

After averaging over $l_*$ and $l_d$, another fast angle $g_a$ is decoupled from the system (see \cref{eq:avg_ld}).
\begin{equation} \label{eq:avg_ld}
\left\{
\begin{aligned}
    \dot{\overline{L}}_d =\, & 0, \\
    \dot{\overline{G}}_d =\, & 0, \\
    \dot{\overline{H}}_d =\, & 0, \\
    \dot{\overline{G}}_a =\, & 0, \\
    \dot{\overline{g}}_d =\, & \frac{3 (\delta -1) \delta  H_d^2 \Gc^2 {(M_{*_1}+M_{*_2}+m_p)}^2 \mu_p^3 \eta_1 ^2}{L_d^2 \left(G_d^3 \left((\delta -1) \delta  \eta_1 ^2-4\right)-3 (\delta -1) \delta  G_d H_d^2 \eta_1 ^2\right)}, \\
    \dot{\overline{h}}_d =\, & -\frac{3 (\delta -1) \delta  H_d \Gc^2 {(M_{*_1}+M_{*_2}+m_p)}^2 \mu_p^3 \eta_1 ^2}{L_d^2 \left(G_d^2 \left((\delta -1) \delta  \eta_1 ^2-4\right)-3 (\delta -1) \delta  H_d^2 \eta_1 ^2\right)}, \\ 
    \dot{\overline{h}}_a = \, &
    \frac{3 \eta_2 \Gc^4 I_p^{(3)} \mu_p^7 (m_p+M_{*_1}+M_{*_2})^4}{8 m_p G_d^2 L_d^6 L_a^3 \sqrt{1-\frac{H_a^2}{L_a^2}} \left(G_d^2 \left((\delta -1) \delta  \eta_1 ^2-4\right)-3 (\delta -1) \delta  H_d^2 \eta_1 ^2\right)} \\
    & 
    \Bigg(G_d^4 H_a L_a \sqrt{1-\frac{H_a^2}{L_a^2}} \left(-85 (\delta -1) \delta  \eta_1^2+\left(223 (\delta -1) \delta  \eta_1 ^2-16\right) \cos \left(2 \left(h_d-h_a\right)\right)-16\right) \\
    & + 2 G_d^3 H_d \left(125 (\delta -1) \delta  \eta_1 ^2-16\right) \sqrt{1-\frac{H_d^2}{G_d^2}} \cos \left(h_d-h_a\right) \left(2 H_a^2-L_a^2\right) \\
    & + 2 G_d^2 H_d^2 H_a L_a \sqrt{1-\frac{H_a^2}{L_a^2}} \left(-95 (\delta -1) \delta  \eta_1 ^2+\left(8-68 (\delta -1) \delta  \eta_1 ^2\right) \cos \left(2 \left(h_d-h_a\right)\right)+24\right) \\
    & + 174 (\delta -1) \delta  G_d H_d^3 \eta_1 ^2 \sqrt{1-\frac{H_d^2}{G_d^2}} \cos \left(h_d-h_a\right) \left(2 H_a^2-L_a^2\right)-87 (\delta -1) \delta  H_d^4 H_a L_a \eta_1 ^2 \left(\cos \left(2 \left(h_d-h_a\right)\right)+3\right) \sqrt{1-\frac{H_a^2}{L_a^2}}\Bigg) \\
    \dot{\overline{H}}_a =\, & -\frac{3 \eta_2  \Gc^4 I_p^{(3)} \mu_p^7 \sin \left(h_d-h_a\right) (m_p+M_{*_1}+M_{*_2})^4}{4 m_p G_d^2 L_d^6 L_a^2 \left(G_d^2 \left((\delta -1) \delta  \eta_1 ^2-4\right)-3 (\delta -1) \delta  H_d^2 \eta_1 ^2\right)} \Bigg(G_d^4 \left(223 (\delta -1) \delta  \eta_1 ^2-16\right) \cos \left(h_d-h_a\right) \left(H_a^2-L_a^2\right) \\
    & - 8 G_d^2 H_d^2 \left(17 (\delta -1) \delta  \eta_1 ^2-2\right) \cos \left(h_d-h_a\right) \left(H_a^2-L_a^2\right) + G_d^3 H_d H_a L_a \left(16-125 (\delta -1) \delta  \eta_1 ^2\right) \sqrt{1-\frac{H_d^2}{G_d^2}} \sqrt{1-\frac{H_a^2}{L_a^2}} \\
    & - 87 (\delta -1) \delta  G_d H_d^3 H_a L_a \eta_1 ^2 \sqrt{1-\frac{H_d^2}{G_d^2}} \sqrt{1-\frac{H_a^2}{L_a^2}}+87 (\delta -1) \delta  H_d^4 \eta_1 ^2 \cos \left(h_d-h_a\right) \left(L_a^2-H_a^2\right)\Bigg)  \\
\end{aligned}
\right.
\end{equation}
and $\overline{g}_a$ is decoupled from the system.

Changing coordinates to $X = H_a/G_a$ and $h=h_d-h_a$, from \cref{eq:avg_ld}, we have
\begin{align}
\left\{
    \begin{aligned}
        \dot{X} &= \sin(h) \left[ C_1 X \sqrt{1-X^2} + 4 C_2 \cos (h) (1-X^2) \right], \\
        \dot{h} &= \frac{C_1 \cos (h) \left( 1-2X^2 \right)}{\sqrt{1-X^2}} - 2 C_2 X \cos(2 h)  + C_3 + 2 C_4 X, \\
    \end{aligned}
\right.
    \label{eq:approximated_dynmaics_appendix}
\end{align}
with 
\begin{align} \label{eq:Cs}
  \left\{
    \begin{aligned}
      C_1 & = \frac{3 \eta_2 \Gc^4 I_p^{(3)} \mu_p^7 \sqrt{1-\frac{H_d^2}{G_d^2}} (m_p+M_{*_1}+M_{*_2})^4 \left(G_d^2 H_d \left(125 (\delta -1) \delta  \eta_1 ^2-16\right)+87 (\delta -1) \delta  H_d^3 \eta_1 ^2\right)}{4 m_p G_d L_d^6 L_a \left(G_d^2 \left((\delta -1) \delta  \eta_1 ^2-4\right)-3 (\delta -1) \delta  H_d^2 \eta_1 ^2\right)} \\
      C_2 & = \frac{3 \eta_2 \Gc^4 I_p^{(3)} \mu_p^7 \left(G_d^2-H_d^2\right) (m_p+M_{*_1}+M_{*_2})^4 \left(G_d^2 \left(223 (\delta -1) \delta  \eta_1^2-16\right)+87 (\delta -1) \delta  H_d^2 \eta_1^2\right)}{16 m_p G_d^2 L_d^6 L_a \left(G_d^2 \left((\delta -1) \delta  \eta_1^2-4\right)-3 (\delta -1) \delta  H_d^2 \eta_1^2\right)}, \\
      C_3 & = -\frac{3 (\delta -1) \delta  \Gc^2 H_d \mu_p^3 \eta_1^2 (m_p+M_{*_1}+M_{*_2})^2}{L_d^2 \left(G_d^2 \left((\delta -1) \delta \eta_1 ^2-4\right)-3 (\delta -1) \delta  H_d^2 \eta_1^2\right)}, \\
      C_4 & = \frac{3 \eta_2 \Gc^4 I_p^{(3)} \mu_p^7 (m_p+M_{*_1}+M_{*_2})^4 \left(G_d^4 \left(85 (\delta -1) \delta \eta_1 ^2+16\right)+2 G_d^2 H_d^2 \left(95 (\delta -1) \delta  \eta_1^2-24\right)+261 (\delta -1) \delta  H_d^4 \eta_1^2\right)}{16 m_p G_d^2 L_d^6 L_a \left(G_d^2 \left((\delta -1) \delta  \eta_1^2-4\right) - 3 (\delta-1) \delta H_d^2 \eta_1^2\right)} \\
    \end{aligned}
  \right.
\end{align}
By definition, we know that $H_a$ is the angular momentum's orthogonal projection to $\bm E_3$ axis
and $h_d-h_a$ is the phase difference between the two precessions (the precession of the planet's orbit and the precession of the angular momentum).
Thus $X, h$ together give the approximated dynamics of the spin.

The Hamiltonian corresponding to the dynamics above is \cref{eq:approximated_dynmaics_appendix}
\begin{align} \label{eq:avg_H_appendix}
    H(X, h) = & C_1 \sqrt{1-X^2} X \cos (h) + C_2 \left( 1-X^2 \right) \cos(2h) +C_3 X+C_4 X^2.
\end{align}

The averaged dynamics has an $\mathcal O \left( \max(\eta_1,\sqrt{\eta_2}) \right)$ error at least till time $\mathcal O\left( \frac{1}{\max(\eta_1,\sqrt{\eta_2})} \right)$, and numerically, we observed that the approximation remains accurate over an even longer time span.

\subsection{Properties of the Dynamical System Corresponding to the Secular Theory}\label{SI:secular_analyses}

Note that the bifurcations of a physical system happen in the vicinity of $i_p=90^\circ$, one has to zoom in the $i_p$ to capture the bifurcation. In addition, the topology of the phase portraits across the bifurcations are twisted.
Therefore, as a complementary, we explore bifurcation of parameter $C_4$ by plotting the phase portraits and the bifurcation diagram of $X$, $h$ under a set of unphysical $C_1$, $C_2$, $C_3$ in \cref{fig:bifurcations_fixC1C2C3,fig:bifurcations_phase_portraits}.
We fix $C_1=C_2=C_3=-1$ in \cref{eq:approximated_dynmaics_appendix} and varying $C_4$.
It can be seen from \cref{fig:bifurcations_fixC1C2C3,fig:bifurcations_phase_portraits} that bifurcations happen at $C_4=-1,\,\frac{7\pm\sqrt{17}}{8}$ and the Hamiltonian pitchfork bifurcation at $C_4=\frac{7+\sqrt{17}}{8}$ have similar topological change of the phase portraits (shown from \cref{phase:7} to \cref{phase:8}) with the Hamiltonian bifurcation in the main article.

\begin{figure}
    \centering
    \includegraphics[width=0.9\linewidth]{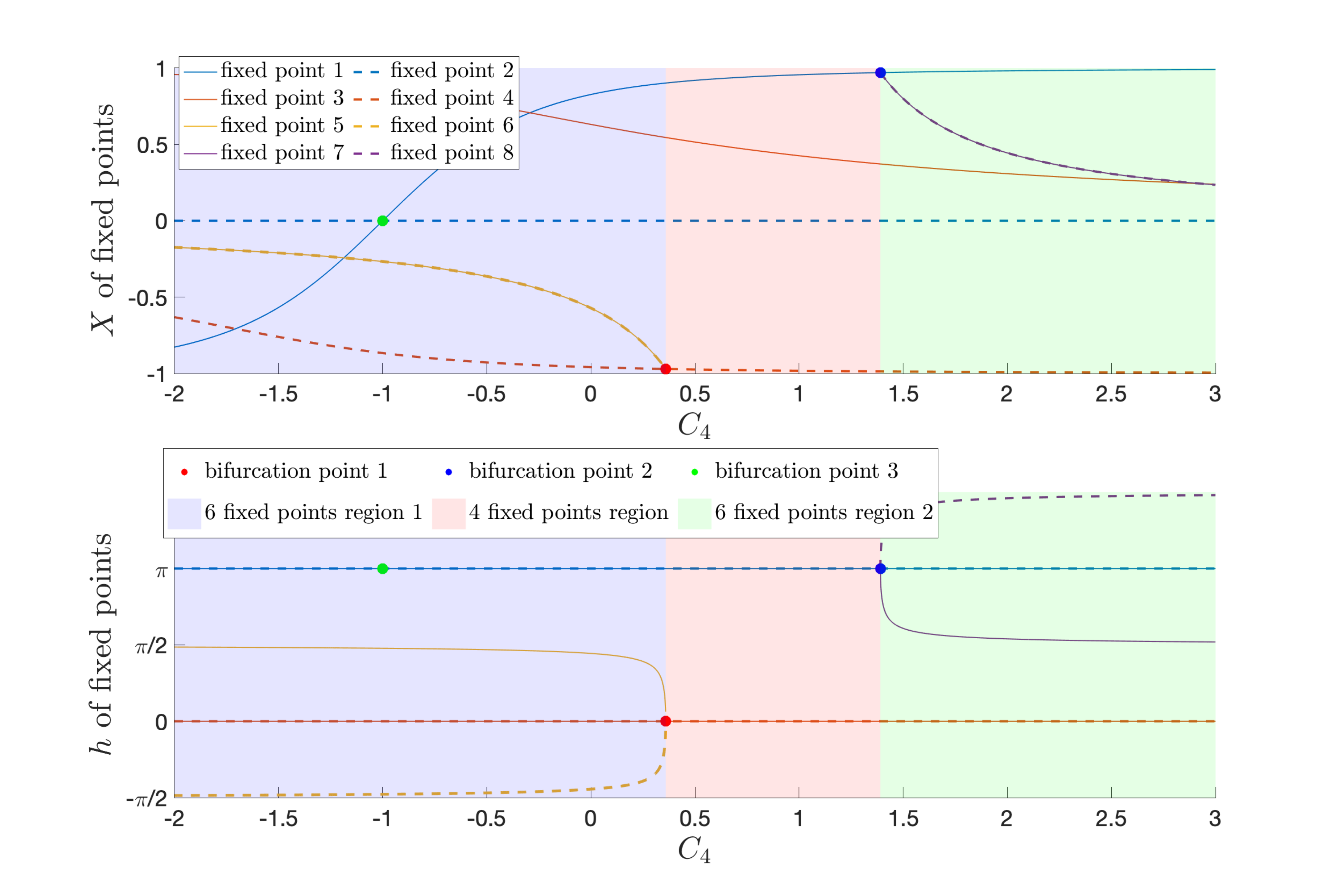}
    \caption{Bifurcation diagram with the varying parameter being $C_4$. Families of fixed points are denoted by different colors. Dots indicate bifurcation locations.
    Fixed point 2,3,5,6,7,8 are centers; fixed point 1 changes from center to saddle at bifurcation point 1; fixed point 4 changes from saddle to center at bifurcation point 2.
    \textbf{Bifurcation point 1} (Hamiltonian pitchfork bifurcation): for $C_4>\nicefrac{7-\sqrt{17}}{8}$, fixed point 4 is a center; for $C_4<\nicefrac{7-\sqrt{17}}{8}$, fixed point 4 is a saddle, there are two more centers (fixed point 5,6).
    \textbf{Bifurcation point 2} (Hamiltonian pitchfork bifurcation): for $C_4<\nicefrac{7+\sqrt{17}}{8}$, fixed point 1 is a center; for $C_4>\nicefrac{7+\sqrt{17}}{8}$, fixed point 1 becomes a saddle, there are two more centers (fixed point 7,8).
    \textbf{Bifurcation point 3} is a saddle-node bifurcation.}
    \label{fig:bifurcations_fixC1C2C3}
\end{figure}

\begin{figure}
\begin{subfigure}{0.3\textwidth}
    \includegraphics[width=\linewidth]{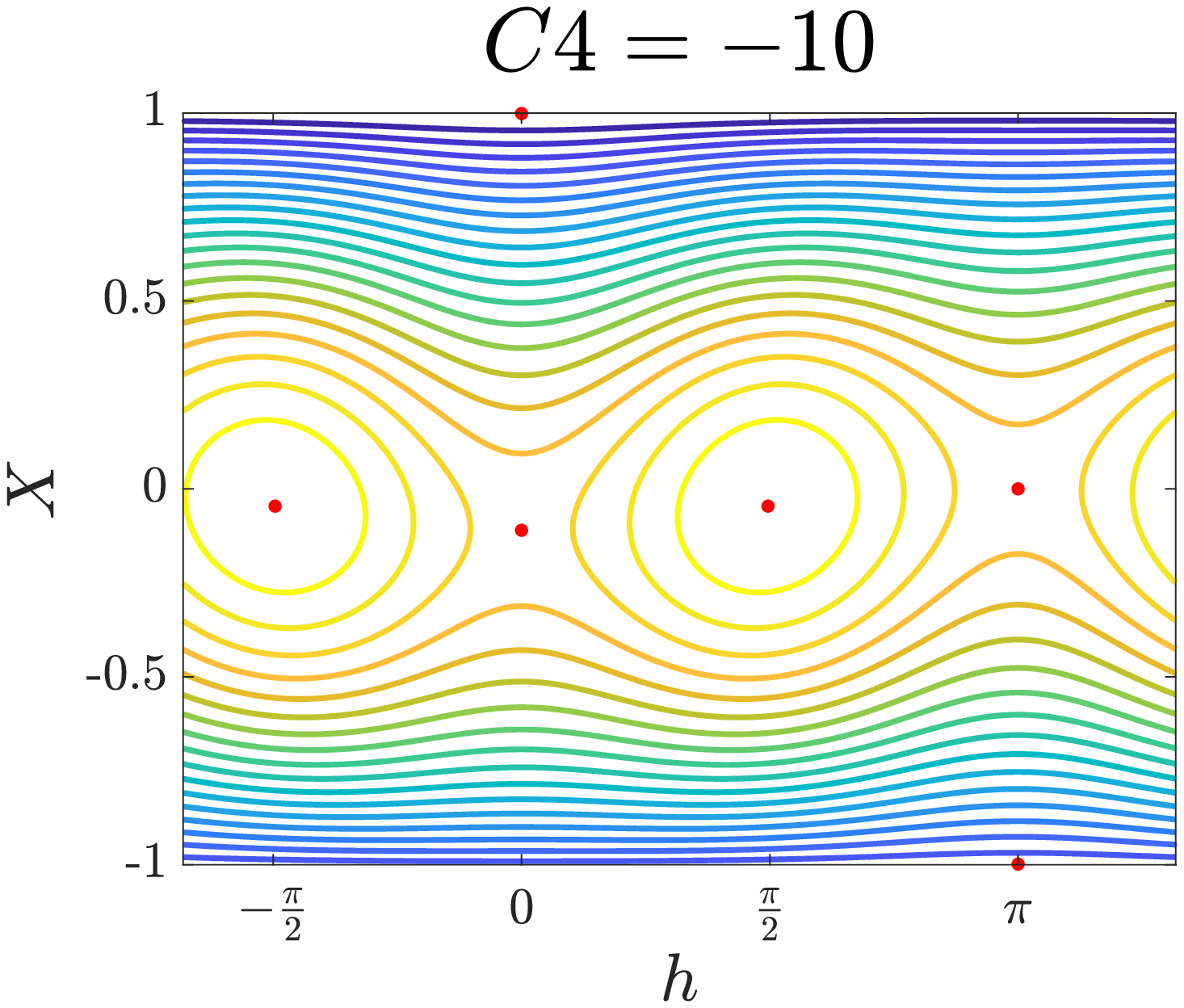}
    \caption{}\label{phase:1}
\end{subfigure}\hspace*{\fill}
\begin{subfigure}{0.3\textwidth}
    \includegraphics[width=\linewidth]{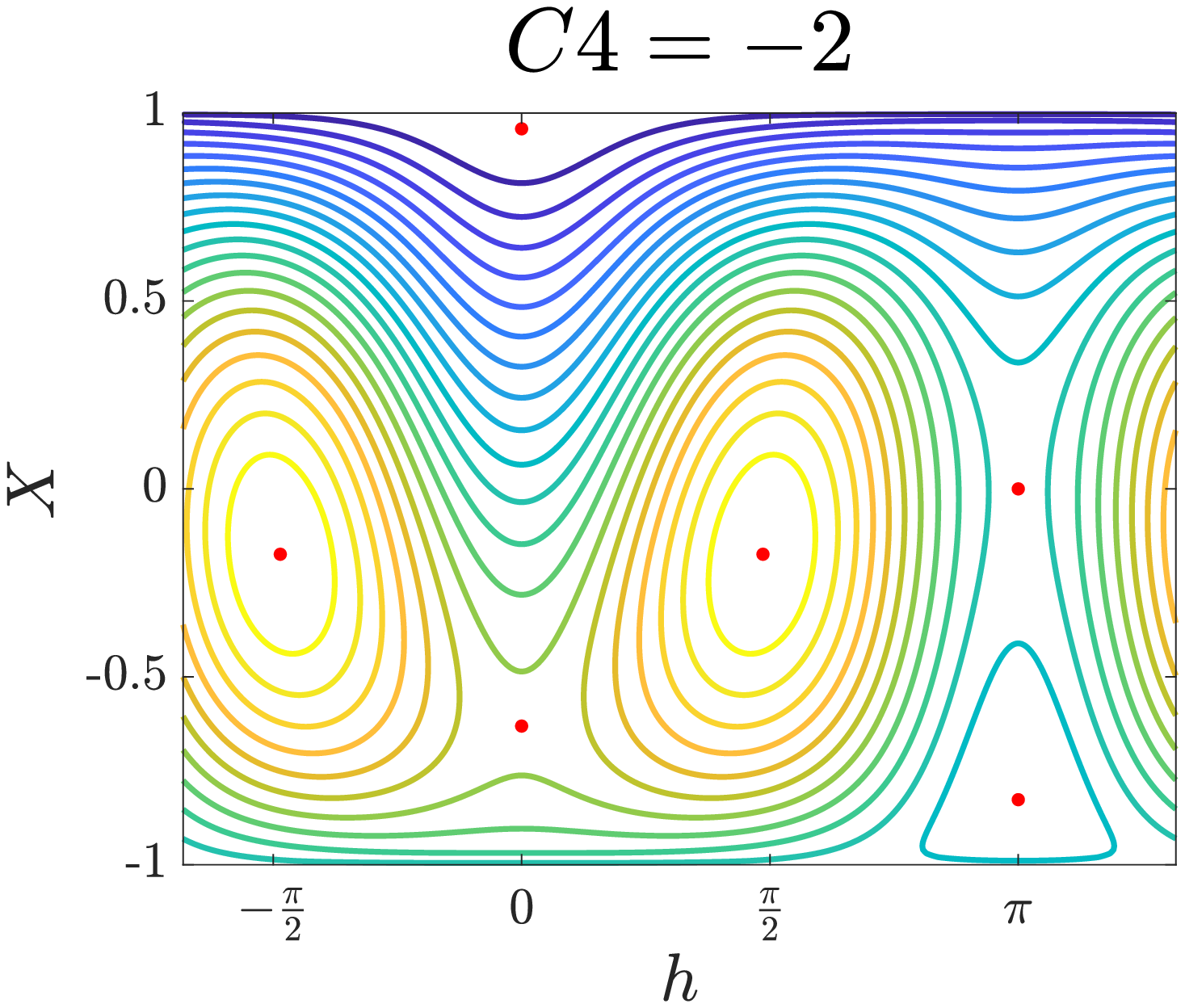}
    \caption{}\label{phase:2}
\end{subfigure}\hspace*{\fill}
\begin{subfigure}{0.3\textwidth}
    \includegraphics[width=\linewidth]{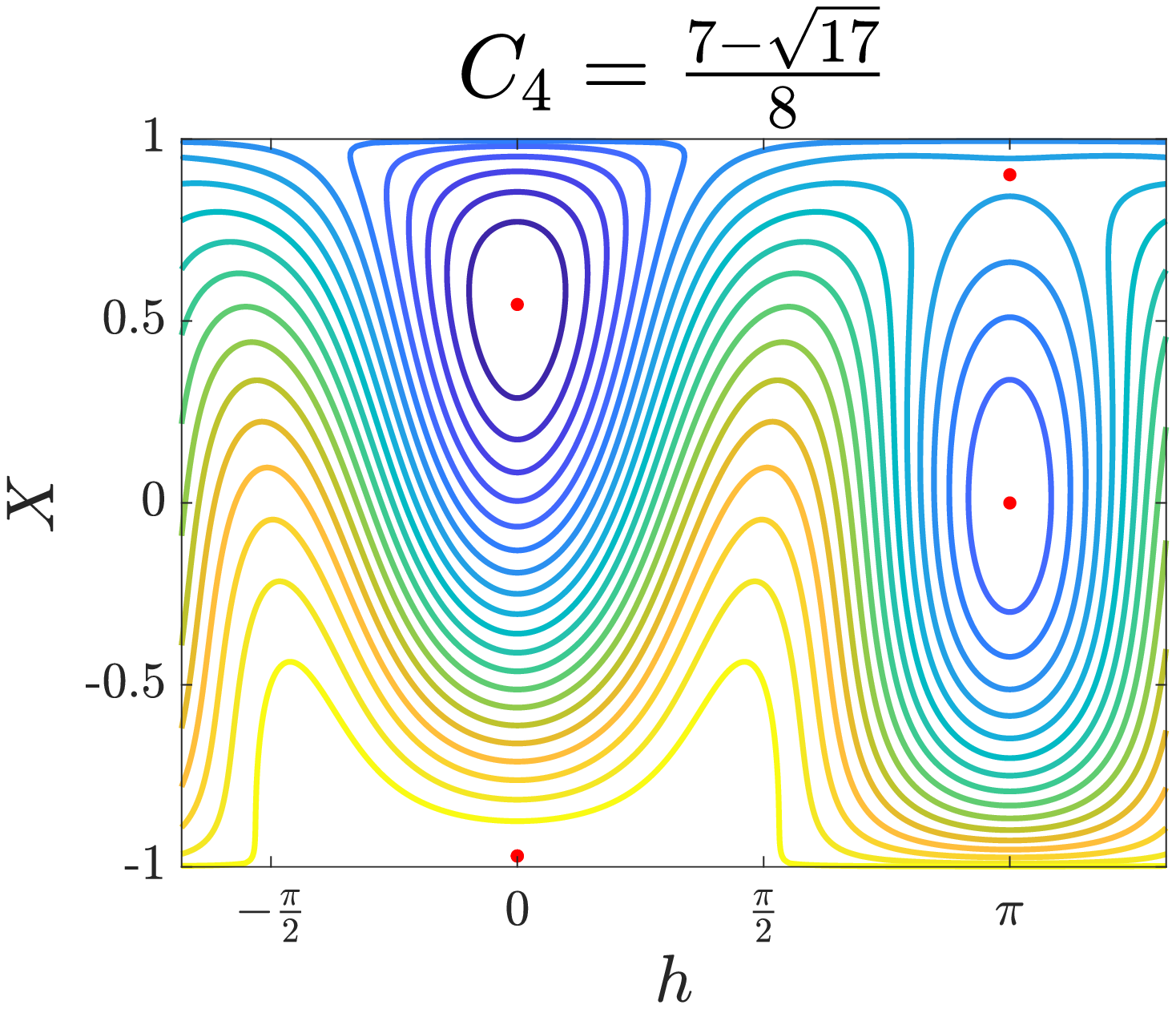}
    \caption{}\label{phase:3}
\end{subfigure}
\medskip
\begin{subfigure}{0.3\textwidth}
    \includegraphics[width=\linewidth]{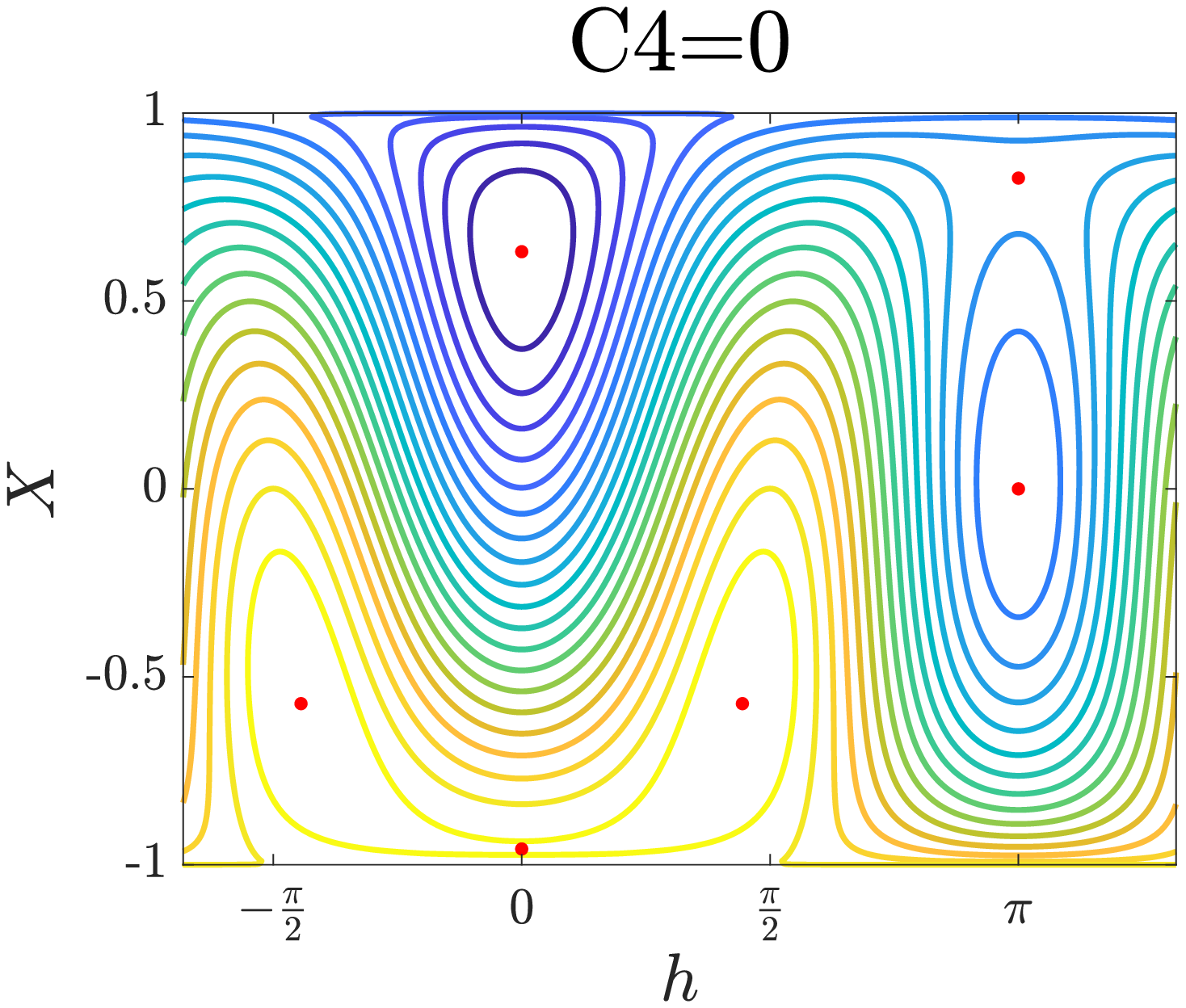}
    \caption{}\label{phase:4}
\end{subfigure}\hspace*{\fill}
\begin{subfigure}{0.3\textwidth}
    \includegraphics[width=\linewidth]{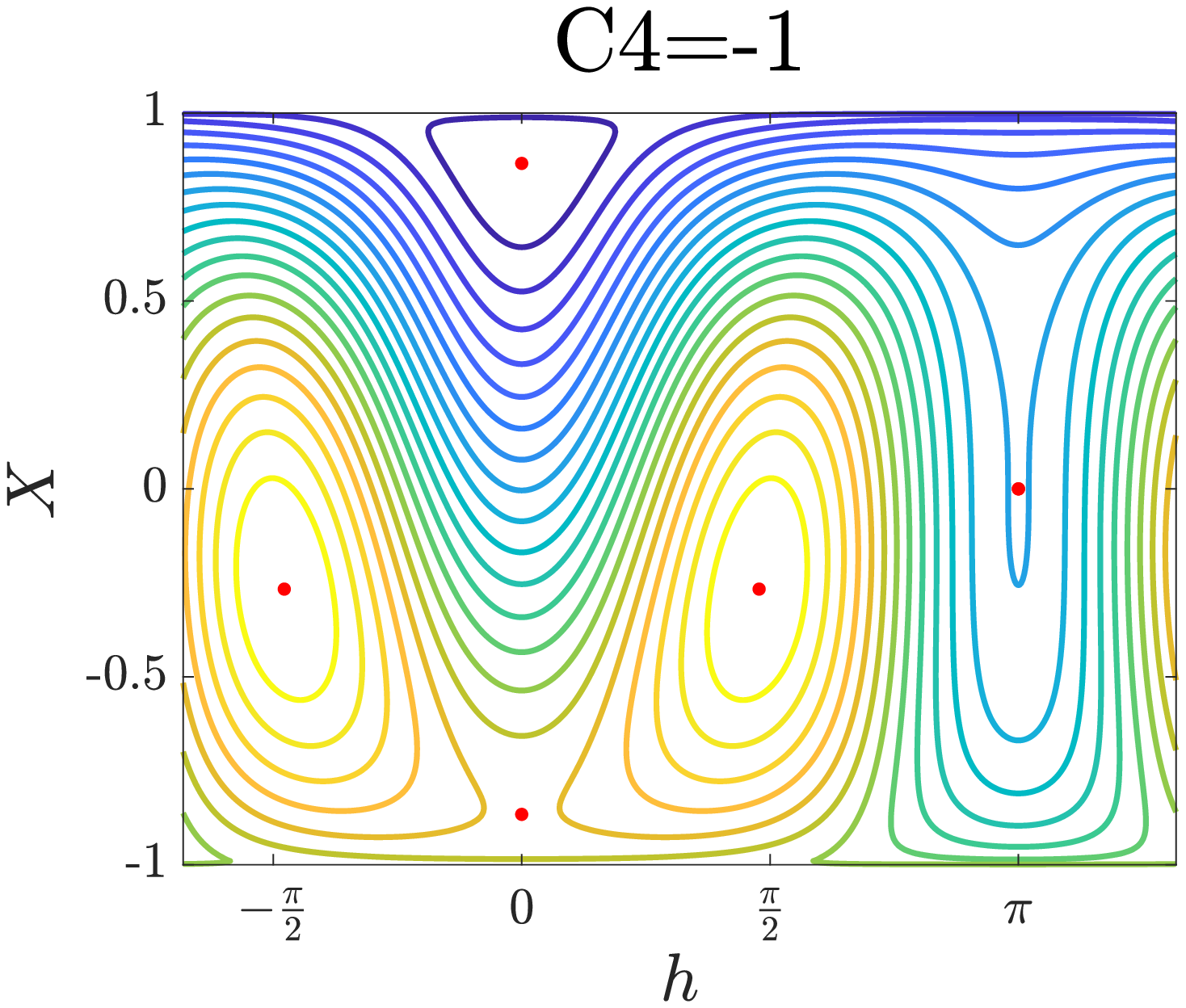}
    \caption{}\label{phase:5}
\end{subfigure}\hspace*{\fill}
\begin{subfigure}{0.3\textwidth}
    \includegraphics[width=\linewidth]{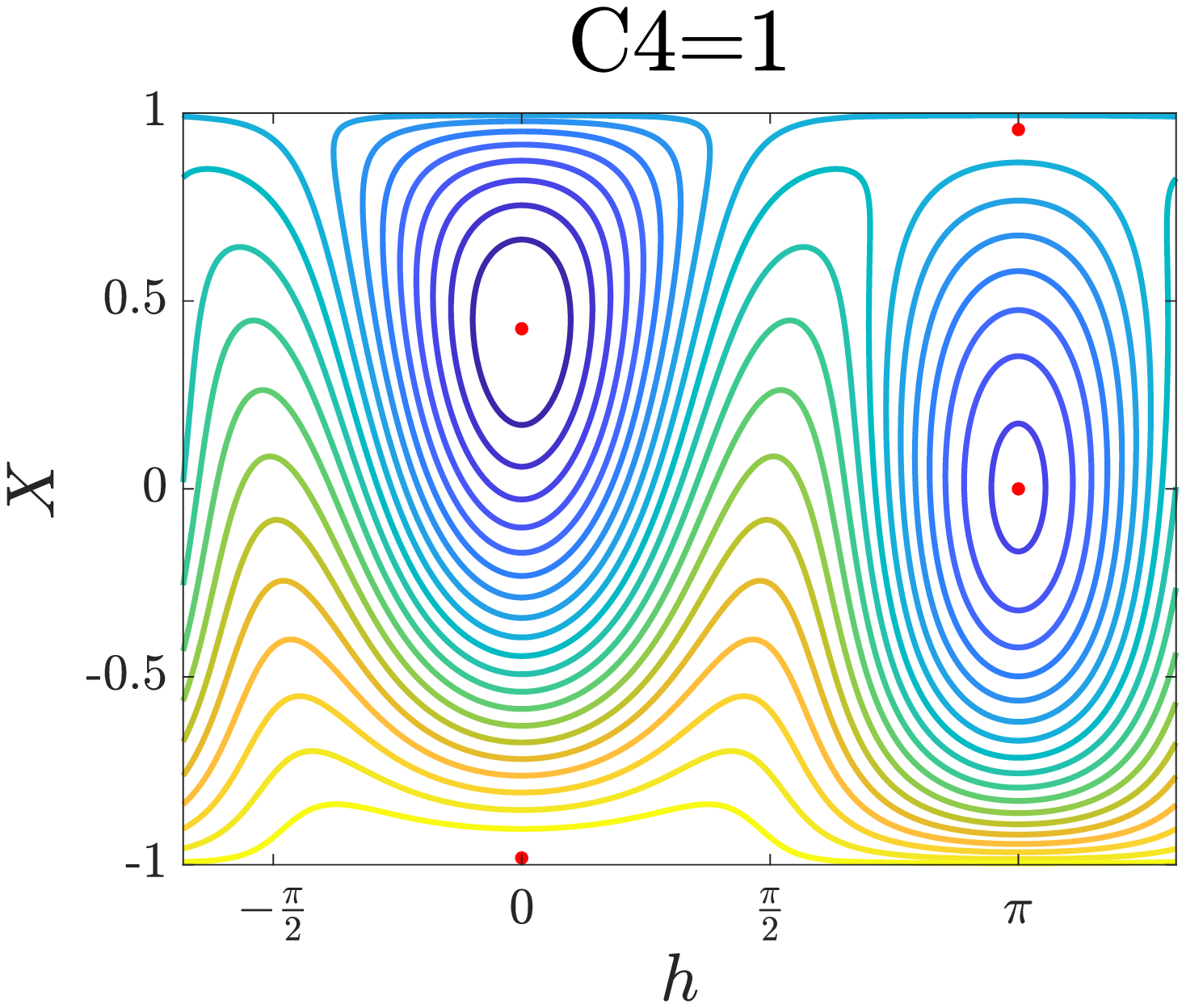}
    \caption{}\label{phase:6}
\end{subfigure}
\medskip
\begin{subfigure}{0.3\textwidth}
    \includegraphics[width=\linewidth]{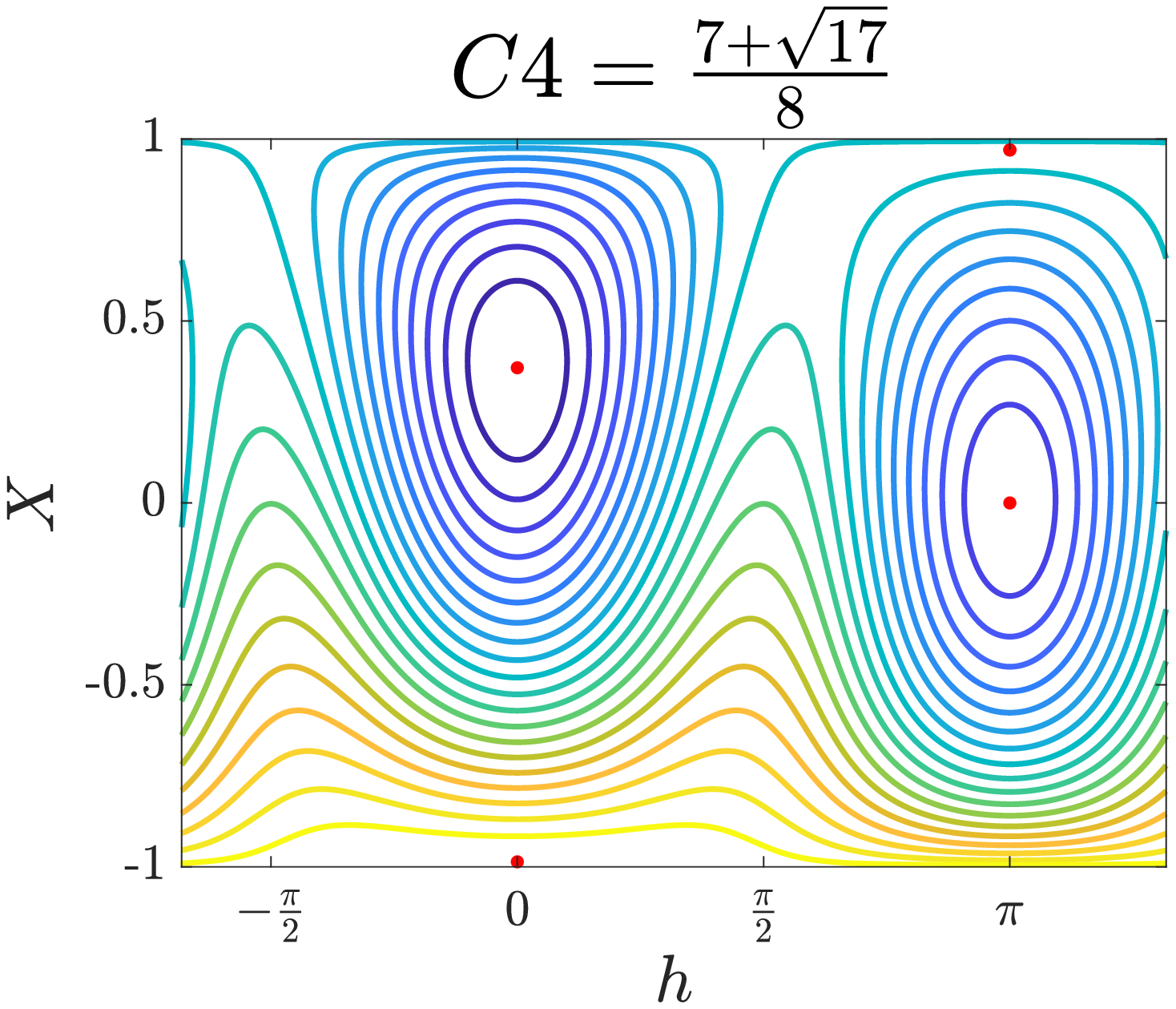}
    \caption{}\label{phase:7}
\end{subfigure}\hspace*{\fill}
\begin{subfigure}{0.3\textwidth}
    \includegraphics[width=\linewidth]{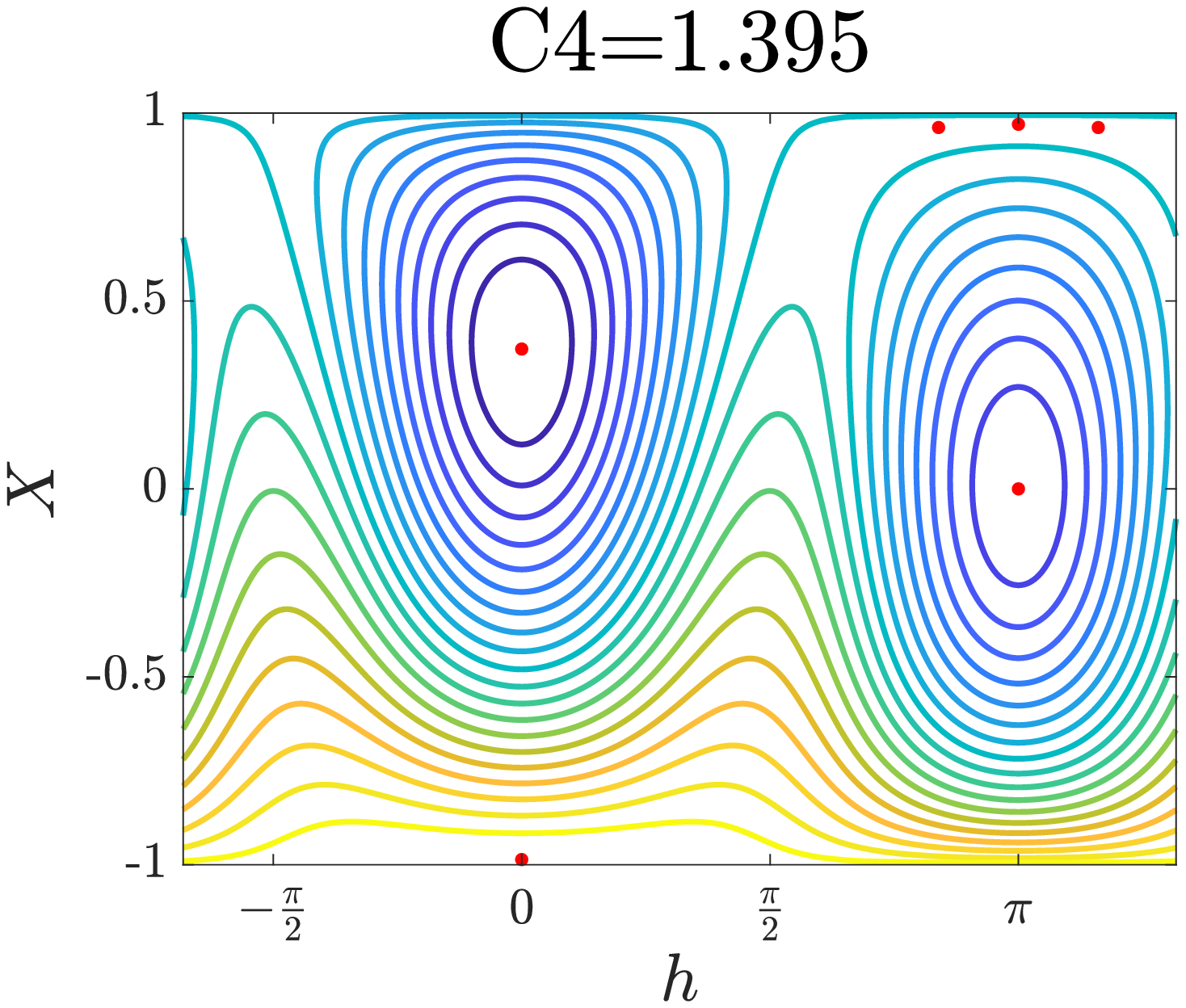}
    \caption{}\label{phase:8}
\end{subfigure}\hspace*{\fill}
\begin{subfigure}{0.3\textwidth}
    \includegraphics[width=\linewidth]{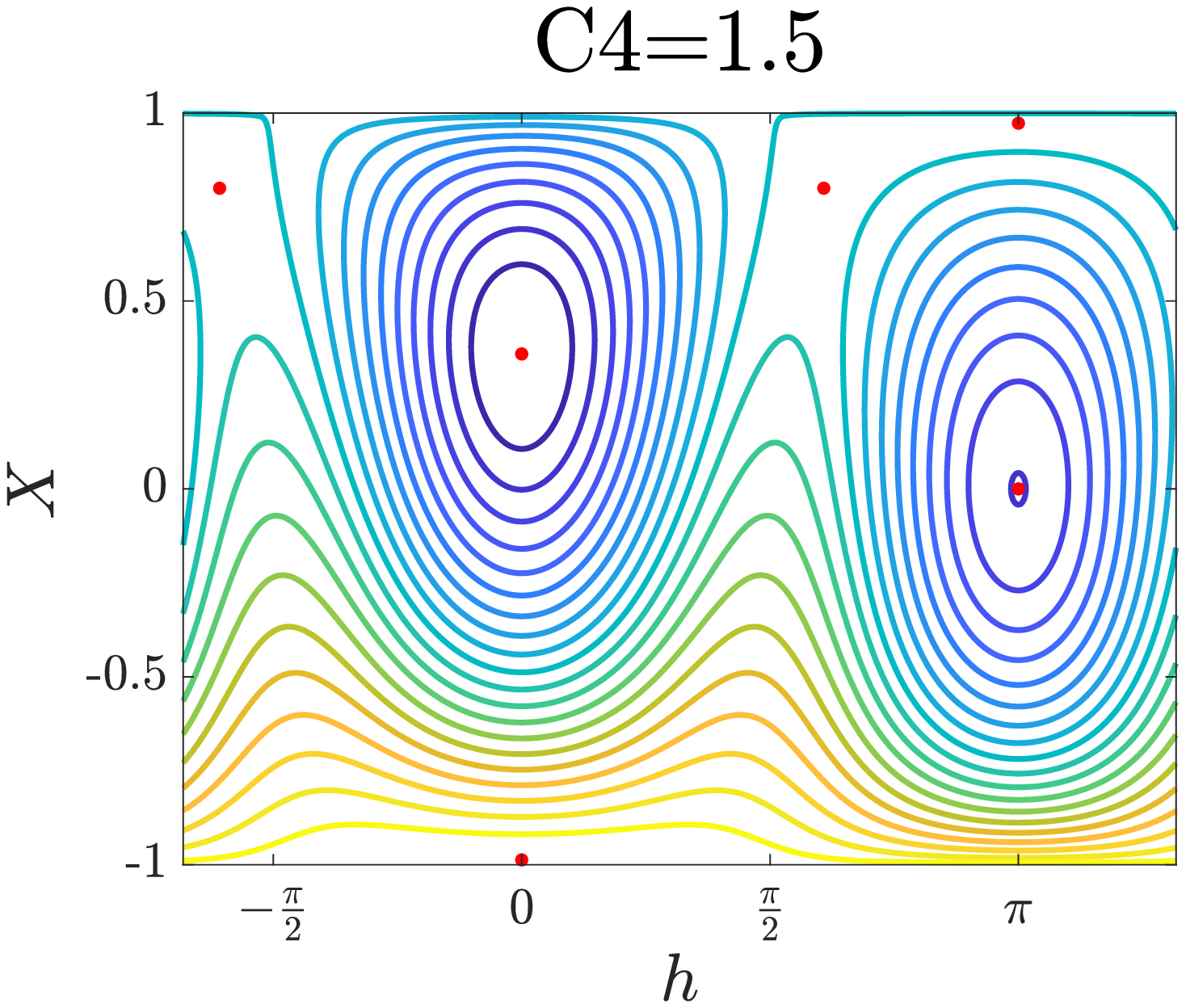}
    \caption{}\label{phase:9}
\end{subfigure}
\medskip
\begin{subfigure}{0.3\textwidth}
    \includegraphics[width=\linewidth]{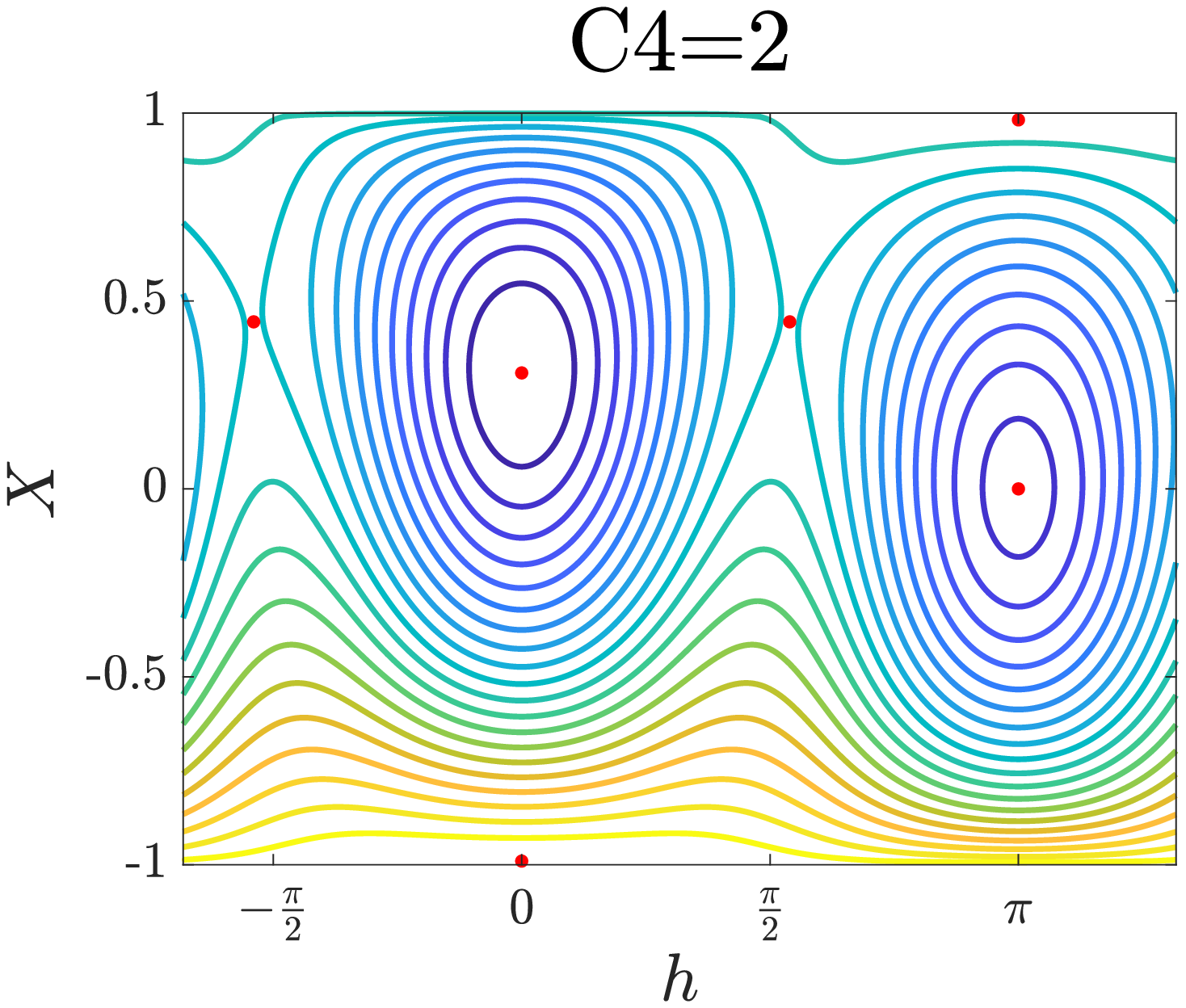}
    \caption{}\label{phase:10}
\end{subfigure}\hspace*{\fill}
\begin{subfigure}{0.3\textwidth}
    \includegraphics[width=\linewidth]{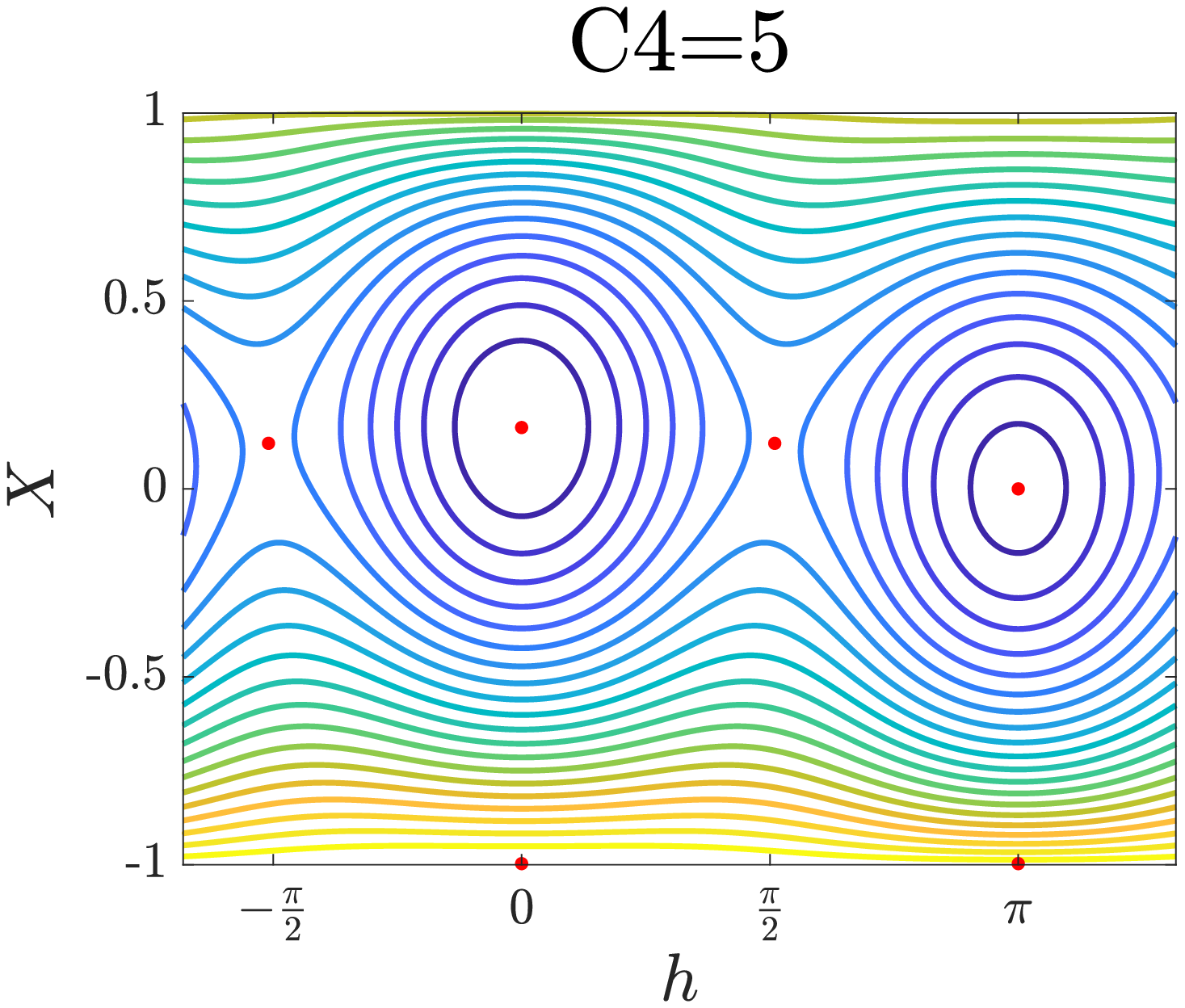}
    \caption{}\label{phase:11}
\end{subfigure}\hspace*{\fill}
\begin{subfigure}{0.3\textwidth}
    \includegraphics[width=\linewidth]{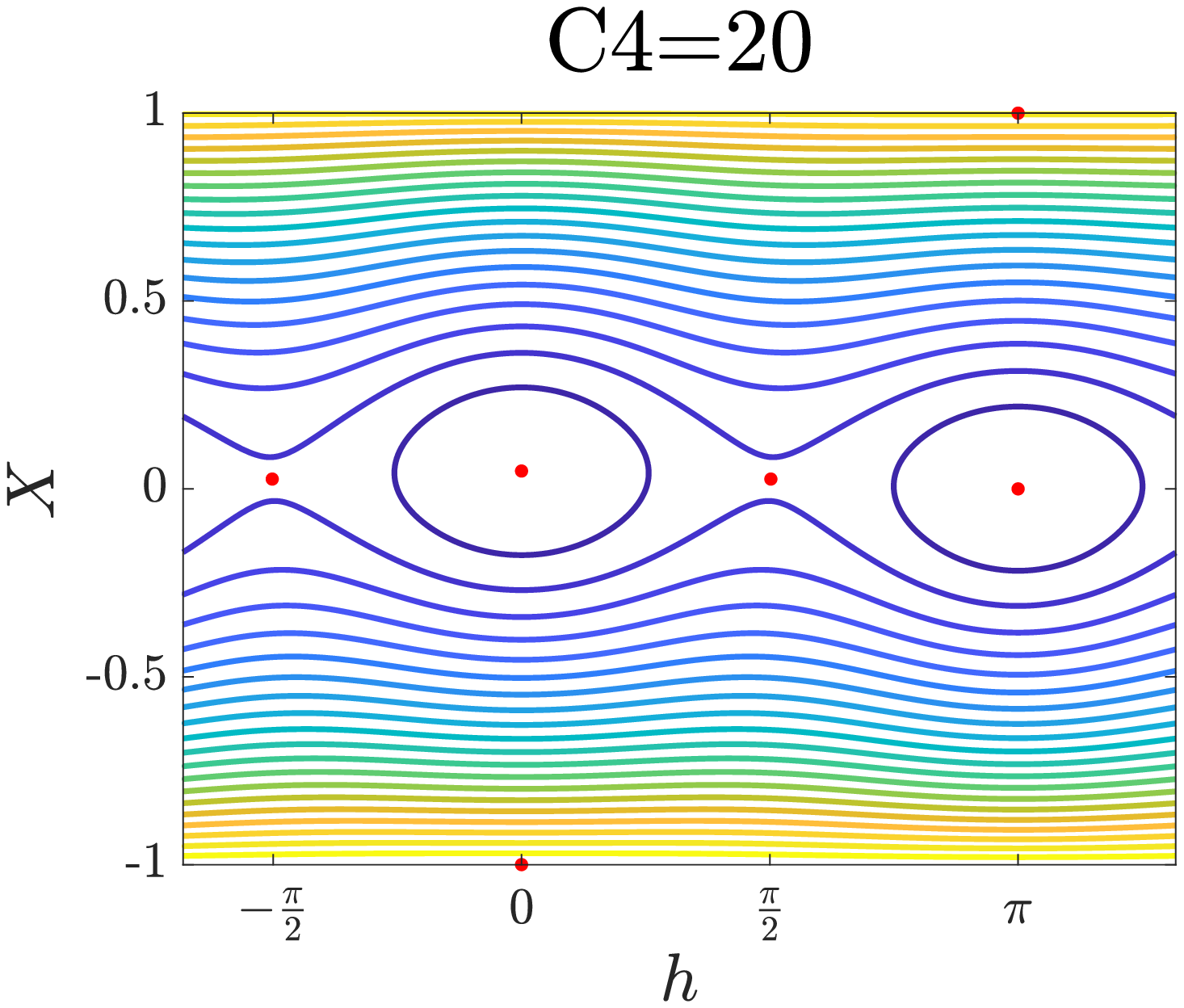}
    \caption{}\label{phase:12}
\end{subfigure}
\medskip
\caption{\cref{phase:1,phase:2,phase:4,phase:8,phase:9,phase:10,phase:11,phase:12} have $6$ fixed points; \cref{phase:5} has $5$ fixed points; \cref{phase:3,phase:5,phase:7} have $4$ fixed points. are phase portraits where bifurcations take place.}\label{fig:bifurcations_phase_portraits}
\end{figure}






\bsp	
\label{lastpage}

\end{document}